\documentclass{iopjournal}

\usepackage[T1]{fontenc}
\usepackage{helvet}
\usepackage{mathtools}
\usepackage{amssymb}
\usepackage{bm}
\usepackage{tikz-cd}

\usepackage{setspace}
\allowdisplaybreaks[2]
\graphicspath{{figures/}}
\DeclareGraphicsExtensions{.pdf,.png,.eps}

\usepackage[final,tracking=true,expansion=true]{microtype}
\microtypecontext{spacing=nonfrench}

\hypersetup{
  unicode            = true,
  pdfauthor          = {Daxx W Delucchi},
  pdftitle           = {Unitarity and Finite Self-Energy of a Massive Classical Point Charge as a Naked Point Singularity},
  pdfdisplaydoctitle = true
}

\DeclareRobustCommand{\dd}{\mathop{}\!\mathrm{d}}

\newcommand{\Order}[1]{\mathcal{O}\bigl(#1\bigr)}

\DeclarePairedDelimiter{\abs}{\lvert}{\rvert}
\DeclarePairedDelimiter{\norm}{\lVert}{\rVert}

\DeclareMathOperator{\Dom}{Dom}

\DeclareMathOperator{\Ker}{Ker}
\DeclareMathOperator{\supp}{supp}

\DeclareRobustCommand{\RN}{Reissner--Nordstr\"om}
\DeclareRobustCommand{\KS}{Kerr--Schild}

\DeclareRobustCommand{\calR}{\ensuremath{\mathcal{R}}}
\DeclareRobustCommand{\Fried}{\mathrm{F}}

\DeclareRobustCommand{\Ltwo}{\ensuremath{L^2(0,\infty)}}
\DeclareRobustCommand{\Honezero}{\ensuremath{H_0^1(0,\infty)}}
\DeclareRobustCommand{\HE}{\ensuremath{\mathcal{H}_E}}

\newtheorem{theorem}{Theorem}[section]

\newtheorem{lemma}[theorem]{Lemma}
\newtheorem{corollary}[theorem]{Corollary}

\newcommand{\qedsymbol}{\ensuremath{\square}}

\newenvironment{proof}[1][Proof]{%
  \par\noindent\textbf{#1.}\ }%
  {\hfill$\qedsymbol$\par}

\usepackage[capitalize,nameinlink,noabbrev]{cleveref}
\crefname{theorem}{Theorem}{Theorems}
\crefname{proposition}{Proposition}{Propositions}
\crefname{lemma}{Lemma}{Lemmas}
\crefname{corollary}{Corollary}{Corollaries}

\makeatletter

\renewcommand*\l@section{\@dottedtocline{1}{1.5em}{2.3em}}
\renewcommand*\l@subsection{\@dottedtocline{2}{3.8em}{3.2em}}
\renewcommand*\l@subsubsection{\@dottedtocline{3}{7.0em}{4.1em}}

\def\@dottedtocline#1#2#3#4#5{%
  \ifnum #1>\c@tocdepth \else
    \addpenalty\@secpenalty
    \addvspace{0em}
    {\leftskip #2\relax
     \rightskip \@tocrmarg
     \parfillskip -\rightskip
     \parindent #2\relax
     \@afterindenttrue
     \interlinepenalty\@M
     \leavevmode
     \@tempdima #3\relax
     \advance\leftskip \@tempdima
     \null\nobreak\hskip -\leftskip
     {#4}%
     \nobreak\hfill
     \nobreak
     \hb@xt@\@pnumwidth{%
       \hfil{\ifnum#1=1\bfseries\fi #5}%
     }%
     \par}%
  \fi}

\newcommand{\myack}[1]{%
  \section*{Acknowledgments}
  \addcontentsline{toc}{section}{Acknowledgments}
  #1%
}

\newcommand{\appsection}[1]{%
  \refstepcounter{section}
  \section*{Appendix~\thesection.\ #1}
  \addcontentsline{toc}{section}{Appendix~\thesection.\ #1}
}

\makeatother

\begin{document}

\articletype{Paper} 

\title{Unitary and finite self-energy of a massive classical point charge as a naked point singularity}

\author{Daxx W Delucchi$^1$\orcid{0009-0005-1196-1522}}

\affil{$^1$Department of Physics, Arizona State University, Tempe, USA}

\email{ddelucch@asu.edu}


\begin{abstract}
We analyze linear Einstein--Maxwell perturbations of the superextremal \RN\ geometry, the nonlinear self-field of a single static point charge, in its static \KS\ rest frame. Using optical radial coordinates and gauge-invariant master variables, we reduce the radiative Einstein--Maxwell sectors to scalar wave equations on the half-line. The resulting master equations are of Regge--Wheeler type, with inverse-square potential cores at the optical apex and short-range tails at infinity. Hardy control at the apex, together with the short-range estimates at infinity, proves that the action-induced spatial $T$-energy forms are semibounded and closable. Closing these spatial forms determines their finite-energy configuration-form domains; adjoining the $L^2$ velocity components supplied by the kinetic energy gives the Cauchy-data energy spaces. The form representation theorem then yields the action-normalized Friedrichs Hamiltonians associated with the master operators. The static Coulomb field and its nonlinear gravitational backreaction are treated as the exact background. All radiative perturbations whose Cauchy data have finite total $T$-energy evolve uniquely and unitarily on the completed energy space. The naked singularity remains geometrically present but is dynamically silent: it carries no $T$-energy flux, supports no finite-energy bound state, and introduces no hidden finite-energy endpoint sector. We also construct the forward radiation field at future null infinity, obtaining a translation representation in which the conserved $T$-energy equals the squared $L^2$ norm of the radiation profile in retarded time.
\end{abstract}

\section{Introduction}

Classical electrodynamics assigns an infinite self-energy to a point charge: the
Coulomb energy integral diverges at the charge itself. In Einstein--Maxwell theory, the analogous
object is the \RN\ self-field, in which the Coulomb field and its nonlinear gravitational backreaction
form a single charged configuration. In the subextremal and extremal regimes, the central
singularity is hidden behind an event horizon, yielding a canonical scattering picture and a well-posed
self-field problem outside the black hole.

In the \emph{superextremal} regime, where $Q^2>M^2$, the would-be horizon disappears and the single classical point charge is exposed geometrically as the naked point singularity at $r=0$. From the perspective of wave propagation, this raises three closely related questions. Does the $T$-energy induced by the Einstein--Maxwell action close to a complete finite-energy space carrying unitary dynamics on the fixed superextremal \RN\ background? Does that closure determine the admissible behavior at the exposed singular endpoint without a supplemental boundary law? And does the presence of a naked curvature singularity necessarily destroy predictability?

Throughout this work, ``finite self-energy'' means finite Einstein--Maxwell $T$-energy for linear perturbations of the fixed superextremal \RN\ self-field background. The nonlinear Coulomb self-energy of the background remains divergent in the usual sense. The Einstein--Maxwell action and the static Killing symmetry first supply the perturbative $T$-energy as a quadratic form on the primitive core of smooth, compactly supported physical perturbations, whose elements already have finite energy. No finite-energy Hilbert space or singular-endpoint behavior is presupposed at that stage. The analytic question is whether this action-supplied form is semibounded and closable. Once that is proved, closure in its form topology determines the finite-energy configuration-form domain. Adjoining the $L^2$ velocity component supplied by the kinetic term gives the Cauchy-data energy space; membership in that completed space, rather than an independently imposed endpoint regularity condition, is what finite $T$-energy means.

Previous work on dynamics in non-globally-hyperbolic static spacetimes has framed the singularity problem in terms of self-adjoint evolution. Horowitz and Marolf characterize a static singularity as quantum mechanically regular when test-field evolution is uniquely determined without additional singular-endpoint data~\cite{HorowitzMarolf1995}. Under a full set of locality, causality, symmetry, and conserved-energy axioms, Ishibashi and Wald show more generally that admissible scalar-wave dynamics must arise from some positive self-adjoint realization of the spatial wave operator~\cite{IshibashiWald2003,IshibashiWald2004}. The present problem differs precisely at the point of selection: closure of the Einstein--Maxwell $T$-energy itself fixes the relevant realization.

Among the closest coupled-field antecedents, Moncrief develops the gauge-invariant Hamiltonian reduction of gravitational and electromagnetic perturbations of \RN\ black holes~\cite{Moncrief1975RN}. Giorgi establishes linear stability throughout the subextremal range $|Q|<M$~\cite{GiorgiRNFullSubextremal}, while Apetroaie studies decay, non-decay, and polynomial blow-up of coupled gauge-invariant quantities along the extremal horizon $|Q|=M$~\cite{ApetroaieExtremalRN}. These horizon problems do not encounter an exposed $r=0$ endpoint. Crossing to $|Q|>M$ changes the inner radial problem: the horizon disappears and the singularity is exposed at finite optical distance, requiring the energy-domain construction studied here.

At that exposed endpoint, Stalker and Tahvildar-Zadeh construct Friedrichs dynamics from the natural conserved-energy form of a massless, chargeless, spherically symmetric scalar probe on prescribed \RN\ geometry~\cite{StalkerTahvildarZadeh2004}. Gibbons, Hartnoll, and Ishibashi likewise provide a finite-physical-energy endpoint-selection precedent for vacuum perturbations of negative-mass Schwarzschild, obtaining stability with the selected condition~\cite{GibbonsHartnollIshibashi2005}. The distinction here concerns the dynamical theory whose energy is completed. The zeroth-order \RN\ solution $(g,F)$ remains fixed, while the evolving fields are the coupled linearized Einstein--Maxwell perturbations $(\delta g,\delta F)$ themselves.

Scalar evolution on superextremal \RN\ is nonunique when the singular endpoint is treated through freely specifiable boundary data~\cite{ChirentiSaaSkakala2012}. Dotti, Gleiser, and collaborators construct algebraically special unstable modes in singular Reissner--Nordstr\"om and Kerr regions~\cite{DottiGleiserPullin2007,DottiGleiser2010,DottiReview2022}. Their formal solutions are retained; the question addressed below is their membership in the domain obtained by completing the coupled Einstein--Maxwell $T$-energy.

The present problem also differs from quantum test-field analyses of naked Reissner--Nordstr\"om backgrounds~\cite{BelgiornoMartelliniBaldicchi2000,KapengutKiesslingLingTahvildarZadeh2025}. Those works study the self-adjointness and spectrum of a Dirac operator on a prescribed geometry. Here the Einstein--Maxwell self-field remains the classical dynamical system: its quadratic action produces the master $T$-energy, and closure of that energy determines the physical domain and Friedrichs Hamiltonian. No quantum matter field or independently prescribed singular-endpoint condition enters the construction.

More broadly, the present problem belongs to the Einstein--Maxwell self-field tradition in which the electromagnetic and gravitational fields are treated as dynamical structures rather than as fields assigned to an externally prescribed particle. This perspective lies within the older geometric Einstein--Maxwell tradition associated with Rainich's formulation of electrodynamics in general relativity and with the geometrodynamical program of Misner and Wheeler~\cite{Rainich1925,MisnerWheeler1957}. The present paper is not a Rainich reconstruction or a geon construction; the connection is motivational and concerns the treatment of electromagnetic and gravitational fields as a self-consistent geometric system. Nor is it a soliton model of an elementary particle. It addresses a prerequisite for any such program: whether a charged self-field with a singular core possesses a well-defined finite-energy state space and deterministic dynamics without an independently imposed endpoint law. Extended-field and solitonic particle models, including proposals in which electromagnetic self-interaction produces internal or quantum-like motion, motivate the same preliminary question. For examples of this broader but mathematically distinct research direction, see Faber's discussion of particles as stable topological solitons and L\'opez's analysis of extended electromagnetic self-interaction~\cite{Faber2012,Lopez2020}. Neither construction is used in the proof below; they provide only historical and physical motivation for asking whether a self-field theory possesses a well-defined finite-energy state space. The present paper constructs neither an extended particle nor a model of internal or quantum-like self-motion; it establishes the finite-energy linear dynamics of the fixed \RN\ self-field.

In this work, by contrast, we begin with the quadratic $T$-energy induced by the second variation of the Einstein--Maxwell action on smooth, compactly supported gauge-invariant perturbations. This energy is not an auxiliary regularity condition imposed after the master equations have been derived. It is the conserved physical norm supplied by the Einstein--Maxwell theory itself. On this primitive core the energy is finite because the perturbations are compactly supported away from the endpoints; no completed finite-energy space is assumed at this stage.

The remaining question is whether this action-supplied quadratic invariant
admits a semibounded closed completion in the singular geometry and thereby
determines a complete evolution law. Hardy control of the inverse-square core
at the apex, together with the short-range form estimates at infinity, proves
that each mode form $q_\pm$, initially defined on
$C_0^\infty(0,\infty)$, is semibounded and closable. Its closure,
denoted by the same symbol and distinguished by its stated domain,
determines rather than presupposes the finite-$T$-energy configuration-form
domain. Adjoining the $L^2$ velocity component supplied by the kinetic term
gives the Cauchy-data energy space, and the representation theorem then
determines the associated Friedrichs Hamiltonian $H_\pm^\Fried$. Thus the
Einstein--Maxwell action supplies the physical core form; Hardy and tail
control prove that its completion exists; closure determines the
configuration-form domain; the kinetic term supplies the velocity component;
and the representation theorem supplies the self-adjoint spatial generator.
The Friedrichs realization is not selected independently from a family of
endpoint conditions; it is determined by representation of the closed
action-induced energy form.

The novelty claimed here is therefore not the gauge-invariant master reduction, the Hardy and Friedrichs machinery, finite-energy endpoint selection as a general principle, or the abstract spectral and radiation tools separately. It is the construction of the coupled linearized superextremal Einstein--Maxwell finite-energy theory from the action-supplied core form and the resulting determination, within the Cauchy-data space built from that closed form, of the physical status of the Dotti--Gleiser branch, the spectral content, and the complete radiation representation.

Because $Q^2>M^2$, the static Killing field $T=\partial_t$ remains timelike throughout the regular region and simultaneously determines the global rest frame and the conserved $T$-energy. The optical radius identifies each static slice with the half-line $x\in(0,\infty)$, with the curvature singularity represented by the finite endpoint $x=0$ and spatial infinity by $x=\infty$. In these coordinates the gauge-invariant Einstein--Maxwell master equations become singular Schr\"odinger-type equations with inverse-square cores at the apex and short-range tails at infinity.

The master reduction preserves more than the equations of motion. It carries the quadratic Einstein--Maxwell action and its $T$-energy into decoupled master-channel actions and energies. The endpoint problem thereby becomes a precise question about the forms generated by the physical energy: whether they are semibounded and closable on the optical half-line, which self-adjoint operators their closures represent, and whether the resulting finite-energy dynamics is completely radiative.

For each radiative master field, the spatial $T$-energy is initially the form
\begin{align*}
q_\pm[\phi]
:=
\frac12
\int_0^\infty
\left(
\abs{\partial_x\phi(x)}^2
+
V_\pm(x)\abs{\phi(x)}^2
\right)\,\dd x,
\qquad
\phi\in C_0^\infty(0,\infty).
\end{align*}
The corresponding formal spatial expression is
\begin{align*}
H_\pm^{\mathrm{formal}}
:=
-\partial_x^2+V_\pm(x),
\qquad
x\in(0,\infty).
\end{align*}
No endpoint domain or self-adjoint realization is imposed at this stage. The Hardy and tail estimates yield a semibounded closure with form domain $\Honezero$, and the associated finite-$T$-energy Cauchy-data space is
\begin{align*}
\HE
:=
\Honezero\times\Ltwo.
\end{align*}
This space is not selected before the physical energy is known. It is the completion forced by that energy. The closed form is therefore the precise point at which the physical and mathematical constructions coincide: its domain is the set of finite-$T$-energy configurations, and its associated operator is the Friedrichs Hamiltonian. The operator does not define the physical energy; the physical energy defines the operator.

This constitutive ordering is not a single linear chain. Geometry first
supplies the conserved $T$-energy, whose mode forms close to the physical
energy space and determine the Friedrichs Hamiltonians. At
$H_\pm^\Fried$ the argument bifurcates: the Friedrichs branch yields
unitary bulk evolution, while the Doob branch exposes the positive-square
structure and, together with the short-range spectral analysis, the purely
absolutely continuous spectrum. The two branches then reconverge, together
with the asymptotic scattering geometry, in the unitary radiation
representation:
\[
\begin{tikzcd}[
  column sep=2.35em,
  row sep=2.0em,
  cells={nodes={inner sep=1pt}}
]
&&&& U_\pm(t)
  \arrow[rd]
& {}
\\
E_T
  \arrow[r]
&
q_\pm\!\restriction_{C_0^\infty(0,\infty)}
  \arrow[r]
&
\overline{q_\pm}
  \arrow[r]
&
H_\pm^\Fried
  \arrow[ru,"{\text{Friedrichs}}"]
  \arrow[rd,swap,"{\text{Doob}}"]
&&
\calR_+
  \arrow[r]
&
L^2(\mathbb R_u)
\\
&&&&
\substack{
\sigma(H_\pm^\Fried)
=
\sigma_{\mathrm{ac}}(H_\pm^\Fried)
=[0,\infty)
}
  \arrow[ru]
& {}
\end{tikzcd}
\]
Here $U_\pm(t)$ denotes the unitary wave group on the finite-energy Cauchy-data
space $\HE$, and $\overline{q_\pm}$ is used only schematically for the closure
of the core form; throughout the analysis, the closed form continues to be denoted
by the same symbol $q_\pm$. The upper branch shows that no finite-$T$-energy
information is lost under time evolution. The lower branch shows that none
is hidden in a bound or threshold sector. Their convergence at $\calR_+$
expresses the complete unitary visibility of the finite-energy dynamics at
future null infinity.

The upper branch begins with the Friedrichs realization represented by
the closed form $q_\pm$. Once the action-derived energy form has been closed, no further endpoint realization remains to be selected within the resulting theory. The admissible apex behavior is already encoded in the form domain, and no endpoint law, endpoint energy, or boundary degree of freedom is added to the bulk Einstein--Maxwell theory. The singularity is controlled not by supplementing Einstein--Maxwell but by completing the energy form that the theory itself supplies.

This energy-space construction fixes the status of the Dotti--Gleiser instability modes within the action-selected Einstein--Maxwell theory~\cite{DottiGleiserPullin2007,DottiGleiser2010}. Their algebraically special profiles realize the non-Friedrichs apex branch and have divergent Einstein--Maxwell $T$-energy. They therefore solve the bare master equations only in endpoint classes broader than the physical Hilbert space generated by the Einstein--Maxwell energy. This exclusion is not imposed retrospectively. The same action-derived $T$-energy that generates the dynamics also fixes the admissible state space, and the Dotti branch lies outside its completed form domain. Thus, although the profiles remain formal solutions in broader endpoint classes, no Dotti--Gleiser instability exists in the action-selected finite-$T$-energy dynamics.

The lower branch is supplied by the Doob ground-state transform, which
gives a second exact representation of the same action-derived spatial
energy. A strictly positive zero-energy profile that realizes the Friedrichs behavior at the apex gives
\begin{align*}
\norm{A_\pm u}_{L^2}^2=2q_\pm[u],
\qquad
H_\pm^\Fried=A_\pm^*A_\pm.
\end{align*}
The Doob profile does not select another Hamiltonian. The Hamiltonian has
already been fixed by the closed $T$-energy form; the Doob identity exposes
its positive factorization and rewrites the quadratic form $q_\pm$ as
one-half the squared norm of the corresponding first-order operator. This
is the lower branch of the diagram: the Hamiltonian is already fixed, and
the Doob representation makes its positivity spectrally transparent.
Together with the short-range analysis, it yields purely absolutely
continuous spectrum $[0,\infty)$, with no finite-energy eigenstate or
zero-energy resonance in any fixed radiative channel.

The two branches reconverge in the radiation theory. For each independent
master sector, the absence of point spectrum and of a zero-energy resonance
has a direct physical consequence. The forward radiation map gives the
channelwise isometric radiation representation
\begin{align*}
\calR_+:\HE\longrightarrow L^2(\mathbb{R}_u),
\qquad
E_\pm[\Phi_\pm]=\norm{\calR_+[\Phi_\pm]}_{L^2(\mathbb{R}_u)}^2.
\end{align*}
The independent channel maps reassemble orthogonally on the complete radiative Einstein--Maxwell phase space. Under this representation, bulk evolution becomes translation in retarded time. The zero-flux theorem excludes endpoint exchange, the absence of point and threshold spectrum excludes localized storage, and the radiation isometry excludes a nonradiating finite-energy remainder. Every degree of freedom in the finite-$T$-energy radiative phase space therefore has a complete representation at future null infinity.
The completed bulk data and the assembled radiation profiles are therefore equivalent representations of the same finite-energy radiative state. Read from bulk to radiation, this unitary equivalence excludes information loss: no admissible degree of freedom is absorbed at the apex, retained in a bound or threshold sector, or annihilated by the radiation map. Read in reverse, complete radiation data determine a unique bulk finite-energy state. These are the injective and surjective aspects of the same equivalence; boundary completeness holds because the action-derived theory leaves no hidden finite-energy singularity sector to recover.

The theory does not overcome the singularity by removing it. The superextremal \RN\ background remains a genuine curvature singularity. The construction instead shows that Einstein--Maxwell derives enough physical and mathematical structure to leave the endpoint with no independently prescribed finite-energy dynamics. The action supplies the core energy form; Hardy and tail control prove that it closes; closure determines the configuration-form domain; adjoining the $L^2$ velocity component gives the Cauchy-data energy space; the representation theorem supplies the Friedrichs Hamiltonian; the Doob identity exposes its positive-square structure; and the spectral and radiation theorems leave no trapped or hidden finite-energy sector. The singularity remains geometrically present but is dynamically silent.

This conclusion has a limited yet important implication for cosmic censorship. The present analysis is linear and is performed on a fixed superextremal background; it does not establish formation in gravitational collapse, prove nonlinear stability, or disprove cosmic censorship. Within the sector studied here, however, it shows that
\begin{align*}
\text{naked singularity}
\;\not\Longrightarrow\;
\text{loss of deterministic finite-energy dynamics}.
\end{align*}
The action-selected evolution is unique and unitary; the apex neither absorbs nor emits $T$-energy; and no finite-energy bound state is supported there. The result therefore questions, rather than merely refining, the predictability-based argument associated with cosmic censorship.

The construction proved here is classical. The space $\HE$ is the finite-$T$-energy Cauchy-data space for a single radiative master field, and $\calR_+$ is its classical unitary radiation representation. The complete radiative theory is obtained by the orthogonal assembly of the independently evolving master sectors. The possible use of this structure in a one-particle construction and Fock-space quantization is stated in the conclusion; no quantum theory is proved in the present paper.

\section{Rest frame and optical coordinate}\label{sec:rest-frame}

We begin by fixing the rest frame of the superextremal \RN\ solution and introducing the optical geometry that underlies the later functional analysis. Throughout we use units with $G=c=1$ and signature $(-+++)$. In this language, ray optics provides a concise way to encode how null geodesics probe the geometry; see Ortaggio, Podolský, and Žofka~\cite{OrtaggioPodolskyZofka2009} for background. We work with the superextremal member of the Reissner--Nordström family, viewed as a static Kerr--Schild spacetime describing the nonlinear self-field of a single point charge.

Let $(\mathbb{R}^{1,3},\eta)$ be Minkowski spacetime with Cartesian coordinates $(t,x^1,x^2,x^3)$ and metric
\begin{align*}
\eta_{\mu\nu}\,\dd x^\mu\dd x^\nu
= -\dd t^2+\delta_{ij}\,\dd x^i\dd x^j,\qquad r:=\sqrt{\delta_{ij}x^i x^j}.
\end{align*}
In Kerr--Schild form the physical metric $g$ is written as
\begin{equation}
g_{\mu\nu}
=\eta_{\mu\nu}+2H(r)\,k_\mu k_\nu,
\label{eq:KS-ansatz}
\end{equation}
where $H\colon(0,\infty)\to\mathbb{R}$ is a radial scalar function and $k$ is a null one-form. This covector $k$ is null with respect to both $\eta$ and $g$, and it is geodesic for both Levi--Civita connections~\cite{OrtaggioPodolskyZofka2009}. We choose the null covector
\begin{align*}
k_\mu \,\dd x^\mu := \dd t + \frac{x_i}{r}\,\dd x^i,
\end{align*}
so that
\begin{align*}
\eta^{\mu\nu}k_\mu k_\nu
= -1 + \delta^{ij}\frac{x_i}{r}\frac{x_j}{r}
= -1 + \frac{\delta_{ij}x^i x^j}{r^2}
= 0.
\end{align*}
Thus, $k$ is null and, via $k^\mu := \eta^{\mu\nu}k_\nu$, defines a radial null vector field. Since the Kerr--Schild ansatz is invariant under
$k\mapsto -k$, we take the future-pointing representative.
In these terms the inverse metric takes the form
\begin{align*}
g^{-1} = \eta^{-1} - 2H\,k\otimes k,
\end{align*}
and $k$ is null and geodesic for both $g$ and $\eta$ (a standard Kerr--Schild property). All deviation from flatness is carried by the null square $2H\,k\otimes k$ of this principal congruence.

Solving the Einstein--Maxwell equations in this class, with a static spherically symmetric Maxwell field, uniquely produces the \RN\ family with parameters $(M,Q)$~\cite{OrtaggioPodolskyZofka2009}. In more general Robinson--Trautman--type ansätze one works with a hypersurface-orthogonal null congruence aligned with the Maxwell field~\cite{OrtaggioPodolskyZofka2009}. Here we simply fix the Maxwell 2-form to be the Coulomb field in Minkowski space,
\begin{align*}
F = \frac{Q}{r^{2}}\,\dd t\wedge \dd r,
\end{align*}
so that the coupled equations $G_{\mu\nu}=8\pi T_{\mu\nu}[F]$ are satisfied if and only if
\begin{align*}
H(r)=\frac{M}{r}-\frac{Q^{2}}{2r^{2}},
\end{align*}
with $M$ the ADM mass and $Q$ the total electric charge. This normalization agrees with the distributional Kerr--Schild treatment of the \RN\ family and with the perturbative analysis of the \RN\ branch~\cite{BalasinNachbagauerRN,DottiGleiser2010}. In the Kerr--Schild rest frame, the electromagnetic contribution to the energy--momentum tensor is exactly the tensor-valued distribution supported at $r = 0$ computed in Section~2 of~\cite{BalasinNachbagauerRN}. In this sense, the superextremal \RN\ geometry is the classical self-field of a single point charge.

Passing from Cartesian coordinates $(t,x^i)$ to standard spherical coordinates $(t,r,\theta,\varphi)$ with $x^i = r\,\omega^i(\theta,\varphi)$, the Kerr--Schild ansatz~\eqref{eq:KS-ansatz} becomes
\begin{align*}
\dd s^2
= -\dd t^2+\dd r^2+r^2\dd\Omega^2
+2H(r)\,(\dd t+\dd r)^2,
\end{align*}
so the $t$--$r$ mixing is explicit. To recover the standard static form, we perform an $r$-dependent redefinition of the time coordinate,
\begin{align*}
t_s = t + \psi(r),
\qquad
\psi'(r)=\frac{2H(r)}{-1+2H(r)}.
\end{align*}
Writing
\begin{align*}
f(r):=1-2H(r)=1-\frac{2M}{r}+\frac{Q^2}{r^2},
\end{align*}
we have $-1+2H(r)=-f(r)$, so the denominator in $\psi'(r)$ vanishes exactly at the zeros of $f$. In the superextremal regime $Q^{2}>M^{2}$, the function $f(r)$ has no real zeros, as we show below. Consequently, $\psi$ is smooth on $(0,\infty)$, the time redefinition is globally defined on the static region, and no new coordinate singularities are introduced. In the new coordinates $(t_s,r,\theta,\varphi)$, the line element takes the familiar static, spherically symmetric form
\begin{align}
\dd s^2 &=
 - f(r)\,\dd t_s^2
 + f(r)^{-1}\,\dd r^2
 + r^2\,\dd\Omega^2, \label{eq:RNmetric}\\
f(r) &= 1 - \frac{2M}{r} + \frac{Q^2}{r^2}, \nonumber
\end{align}
where $\dd\Omega^2$ is the round metric on $\mathbb{S}^2$. We henceforth relabel $t_s$ as $t$ and work exclusively in this static coordinate system. The vector field $T:=\partial_t$ is Killing and hypersurface-orthogonal, with integral curves given by the worldlines of static observers.

In these static coordinates,
\begin{align*}
g(T,T)=g_{tt}=-f(r),
\end{align*}
so the sign of $f(r)$ determines the causal character of $T$. Geometrically, $f$ measures the competition between the attractive effect of the mass and the repulsive effect of the charge. Formally, this shows up by factoring out the strictly positive weight $r^{-2}$ and introducing the quadratic polynomial
\begin{align*}
p(r):=r^{2}f(r)=r^{2}-2Mr+Q^{2},\qquad r\in\mathbb{R}.
\end{align*}
For every $r>0$, we have $r^{2}>0$, so $\operatorname{sgn} f(r)=\operatorname{sgn} p(r)$ on $(0,\infty)$. Completing the square gives
\begin{align*}
p(r)
=r^{2}-2Mr+Q^{2}
=(r^{2}-2Mr+M^{2})+(Q^{2}-M^{2})
=(r-M)^{2}+\bigl(Q^{2}-M^{2}\bigr).
\end{align*}
Since $(r-M)^{2}\ge 0$ and $Q^{2}-M^{2}>0$, we find $p(r)>0$ for all $r$, hence $f(r)=p(r)/r^{2}>0$ on $(0,\infty)$. Thus, in the superextremal regime, the static Killing field $T=\partial_{t}$ is everywhere timelike on
$r>0$ and there is no zero of $f$ corresponding to an event horizon; the static coordinates
$(t,r,\theta,\varphi)$ provide the global rest frame used in the master-field reduction
of Section~\ref{sec:master-reduction}. 

On the static slice $\{t=\mathrm{const}\}$, the induced Riemannian metric is
\begin{align*}
\gamma = f(r)^{-1}\,\dd r^2 + r^2\,\dd\Omega^2.
\end{align*}
For any static spacetime of the form~\eqref{eq:RNmetric}, the spatial projections of null geodesics are the geodesics of the optical metric $\gamma_{\mathrm{opt}}:=f(r)^{-1}\gamma$; see Perlick~\cite{Perlick2000} for details. In our case, this becomes
\begin{align}
\gamma_{\mathrm{opt}}
= f(r)^{-2}\,\dd r^2 + f(r)^{-1}r^2\,\dd\Omega^2. \label{eq:optical-metric-r}
\end{align}
The optical metric captures how light rays ``see'' the spatial geometry and is particularly well suited for scattering problems. To exploit this structure we introduce the \emph{optical radius} $x$ by
\begin{align}
\frac{\dd x}{\dd r}=f(r)^{-1}, \label{eq:x-def}
\end{align}
leaving the integration constant unspecified for the moment. Since $f(r)>0$ for all $r>0$ in the superextremal regime, $x(r)$ is smooth and strictly increasing, hence a smooth change of radial coordinate. Different additive constants simply shift the origin of the resulting half-line.\footnote{The asymptotics derived in Appendix~\ref{app:appA} show that $x(r)$ has a finite limit as $r\to0^+$ and diverges as $r\to\infty$.}

It is convenient to normalize $x$ so that the optical distance to the singularity is finite. Since $x(r)$ is strictly increasing, there is a unique additive constant for which
\begin{align*}
  \lim_{r\to0^+}x(r)=0,
\end{align*}
and we fix this choice from now on, writing
\begin{align}
  \frac{\dd x}{\dd r}=f(r)^{-1}, \qquad x(0):=\lim_{r\to0^+}x(r)=0.
\end{align}
The point $x = 0$ is not part of the manifold; it is the limiting endpoint corresponding to the curvature singularity at $r=0$. With this normalization, $r\mapsto x(r)$ is a diffeomorphism $(0,\infty)\to(0,\infty)$, so each static spatial slice is $(0,\infty)_x\times\mathbb{S}^2$, with $x=0$ and $x=\infty$ realized as geometric endpoints rather than actual points of the manifold. In terms of $x$, the optical metric~\eqref{eq:optical-metric-r} takes the spherically symmetric form
\begin{align*}
  \gamma_{\mathrm{opt}}
  = \dd x^2 + f\bigl(r(x)\bigr)^{-1}\,r(x)^2\,\dd\Omega^2,
\end{align*}
with $x$ serving as a radial coordinate on $(0,\infty)_x$.\footnote{Using the asymptotics of $x(r)$ derived in Appendix~\ref{app:appA}, one checks
that the coefficient $f(r(x))^{-1} r(x)^2$ agrees with $x^2$ up to lower-order
terms as $x\to\infty$, so $\gamma_{\mathrm{opt}}$ is an asymptotically Euclidean
scattering metric in the sense of S\'a Barreto~\cite{SaBarreto2003}. Equivalently, rewriting
$\gamma_{\mathrm{opt}}$ in the original radial coordinate $r$ and compactifying
spatial infinity with the defining function $\rho:=r^{-1}$ exhibits the static
slice as an asymptotically Euclidean scattering manifold,
enabling the application of Friedlander's radiation-field construction~\cite{Friedlander1980} via S\'a Barreto's generalization to scattering manifolds.}

The definition~\eqref{eq:x-def} also has a simple dynamical interpretation: it ensures that radial null geodesics have unit coordinate speed in $(t,x)$. Restricting~\eqref{eq:RNmetric} to radial curves ($\dd\Omega^2=0$) and imposing the null condition $\dd s^2=0$ yields
\begin{align}
0=-f(r)\,\dd t^2+f(r)^{-1}\,\dd r^2
\quad\Longrightarrow\quad
\frac{\dd t}{\dd r}=\pm f(r)^{-1}.
\end{align}
Using \eqref{eq:x-def}, $\dd x=f(r)^{-1}\dd r$, so along any radial null geodesic
\begin{align}
\frac{\dd t}{\dd x}=\pm1,
\end{align}
Thus $u:=t-x$ and $v:=t+x$ are constant along outgoing and ingoing
radial null geodesics, respectively. In the $(t,x)$-plane, radial null rays thus appear as straight lines of unit slope, exactly as in $(1+1)$-dimensional Minkowski space. This is the basic reason why the optical coordinate is so well adapted to scattering and radiation problems, and it underlies the wave-equation analysis carried out later (compare the physical-space treatment in~Section~3.4 of \cite{DafermosRodnianski2008} and the general scattering framework of~\cite{Friedlander1980}).

For superextremal \RN, integrating \eqref{eq:x-def} gives the closed-form expression~\cite{ChirentiSaaSkakala2012}
\begin{align}
\label{eq:x-closed-superext}
x(r) &= r + \frac{2M^{2}-Q^{2}}{\sqrt{Q^{2}-M^{2}}}\,
\arctan\!\left(\frac{r-M}{\sqrt{Q^{2}-M^{2}}}\right)
+ M \ln\!\bigl(r^{2}-2Mr+Q^{2}\bigr) + C,
\end{align}
with an arbitrary integration constant $C$. Imposing the normalization $x(0)=0$ fixes $C$ uniquely and identifies the optical origin with the location of the naked singularity at $r=0$ (see Appendix~\ref{ssec:A1} for the derivation of \eqref{eq:x-closed-superext}). A Taylor expansion of \eqref{eq:x-closed-superext} near $r=0$ and as $r\to\infty$ yields
\begin{align}
x(r) &= \frac{r^{3}}{3Q^{2}} + \Order{r^{4}} \qquad (r\to 0^+), \label{eq:x-apex-rto0}\\
x(r) &= r + 2M\ln r + \Order{1} \qquad (r\to \infty), \label{eq:x-apex-rtoinf}
\end{align}
so the naked singularity sits at finite optical distance while spatial infinity
is logarithmically stretched~\cite{DottiGleiser2010} (see
Appendix~\ref{ssec:A2} and~\ref{ssec:A3} for the derivation of \eqref{eq:x-apex-rto0} and \eqref{eq:x-apex-rtoinf}).

In this picture, the Coulomb divergence becomes a statement about the geometry of the optical half-line: with respect to the physical static slices, the Coulomb self-energy still diverges at $r=0$. What changes is that the strong warping of the optical volume element makes the corresponding $T$-energy density locally integrable as $x\to0$, so we work directly on $(0,\infty)_x$ and then complete with respect to this energy. In practice, this means that the geometric endpoints $x=0$ and $x=\infty$ become the endpoints of the closed $T$-energy form, rather than places where extra counterterms or boundary data have to be imposed.

The optical geometry also suggests a natural choice of function spaces. For each spherical harmonic mode
we regard the one-dimensional configuration space as the Hilbert space $\Ltwo$ of
square-integrable complex-valued functions of the optical radius $x$. The dense subspace
$C_0^\infty(0,\infty)\subset\Ltwo$ has a direct geometric interpretation: it consists of smooth
perturbations whose support stays away from both the naked singularity ($x\to0$) and spatial infinity
($x\to\infty$), so that the associated Einstein--Maxwell $T$-energy is automatically finite on each
static slice. In the analysis below we first write the spatial part of the $T$-energy as a quadratic
form on $C_0^\infty(0,\infty)$ and then extend it by closure. The choice of $\Ltwo$ as the ambient Hilbert space and $C_0^\infty(0,\infty)$ as the core thus reflects the optical half-line geometry and the use of
compactly supported variations in the action, rather than an \emph{a priori} boundary condition at $x=0$ or
$x=\infty$ or an independent endpoint regularity criterion. Finiteness of the action-derived energy remains the physical admissibility requirement; closure determines its completed domain. We make this precise in Sections~\ref{sec:hardy}--\ref{sec:friedrichs} when constructing the Friedrichs extension.

\section{Master field reduction}\label{sec:master-reduction}

On the superextremal \RN\ background $(g,F)$ constructed in Section~\ref{sec:rest-frame}, with metric~\eqref{eq:RNmetric}, we linearize the coupled Einstein--Maxwell system and recast the perturbations in terms of scalar master fields on the optical half-line. As in the Regge--Wheeler--Zerilli and Kodama--Ishibashi constructions, the full tensorial perturbation problem is organized into a finite set of gauge-invariant amplitudes, each of which obeys a $(1+1)$-dimensional wave equation with an effective potential. For the classical axial/polar perturbation framework and Schr\"odinger-type master-equation language, see Chandrasekhar~\cite{Chandrasekhar1983}; the modern gauge-invariant formalism used here follows Kodama and Ishibashi~\cite{IshibashiKodama2011PTPS}. In our setting, this reduction has a particularly clean interpretation: once expressed in optical coordinates, the self-field dynamics of the naked singularity becomes a family of Schr\"odinger-type problems for these master fields, with the geometry of the optical half-line directly reflected in the structure of the effective potentials.

Let $(\delta g,\delta F)$ denote smooth, compactly supported perturbations of the metric and Maxwell field on the static region $\{t\in\mathbb{R},\,r>0\}$. By spherical symmetry we decompose $(\delta g,\delta F)$ into scalar (polar) and vector (axial) tensor harmonics on $\mathbb{S}^2$, yielding, for each $(\ell,m)$, a finite set of scalar coefficient functions of $(t,r)$.
The linearized Einstein--Maxwell system is invariant under infinitesimal diffeomorphisms and Maxwell gauge transformations, generated by pairs $(\xi,\lambda)$ with $\xi$ a vector field on the background spacetime and $\lambda$ a Maxwell gauge one-form, so not all coefficients are physical. For each parity sector and each $(\ell,m)$, Kodama and Ishibashi construct explicit gauge-invariant combinations of the harmonic coefficients of $(\delta g,\delta F)$. These variables are fixed under the transformation
\begin{align}
(\delta g,\delta F) \mapsto (\delta g + \mathcal{L}_\xi g,\ \delta F + \mathcal{L}_\xi F + \dd\lambda).
\end{align}
They encode the physical content of the perturbations and already diagonalize the quadratic action: each mode evolves independently.\footnote{We follow the Kodama--Ishibashi conventions for the scalar (polar) and vector (axial) harmonic decompositions; see Sections~2--3 of \cite{IshibashiKodama2011PTPS} for details of the harmonics, their eigenvalues, and the corresponding gauge transformations.}
For each fixed multipole index $\ell\in\mathbb{N}_0$ and azimuthal number $m$, the perturbation splits into a \emph{vector-type} (axial) sector and a \emph{scalar-type} (polar) sector, each invariant under the linearized Einstein--Maxwell operator. The decomposition is orthogonal with respect to the background inner product induced by $g$ and the $\mathbb{S}^2$ measure, so different $(\ell,m)$ and different parity sectors decouple.

We denote by $\Phi_{V\pm,\ell m}(t,r)$ and $\Phi_{S\pm,\ell m}(t,r)$ the gauge-invariant combinations for the vector (axial) and scalar (polar) sectors, where the signs $\pm$ label the two coupled Einstein--Maxwell degrees of freedom in each parity. (We henceforth label the two independent physical master fields as ``$+$'' and ``$-$'', not to be confused with the parity labels $V$ and $S$, which distinguish the axial and polar sectors.)
In the channelwise calculations below, the fixed parity, master-branch,
and angular labels are suppressed whenever no confusion can arise; the
corresponding quantities are denoted by $\phi_{\pm,\ell m}$, $V_\pm$,
$q_\pm$, and $H_\pm^\Fried$.
Substituting the gauge-invariant variables into the linearized Einstein--Maxwell equations shows that, in each parity and master sector, all perturbations are captured by a single scalar master amplitude. In the usual tortoise coordinate $r_*$, these amplitudes satisfy Regge--Wheeler--type wave equations with effective potentials $V_{V\pm}(r)$ and $V_{S\pm}(r)$ depending only on the background function $f(r)$, the charge $Q$, and the multipole index $\ell$. The derivation is purely local on the $(t,r)$ orbit space and uses only that $f(r)>0$, so that $t$ is timelike and $r$ is a good radial coordinate. Neither the zeros of $f$ nor the presence of horizons enters the derivation. This permits analytic continuation to the superextremal regime. The resulting master equations take the form\footnote{Here $K\in\{1,0,-1\}$ encodes the curvature of the constant-$r$ symmetry orbits and $\Lambda$ is the cosmological constant; see Sections~2, 5, and 6 of \cite{IshibashiKodama2011PTPS}.}
\begin{align}
\partial_t^2\Phi_{V\pm,\ell m}-\partial_{r_*}^2\Phi_{V\pm,\ell m}+V_{V\pm}(r)\,\Phi_{V\pm,\ell m}&=0,\label{eq:master-V-rstar}\\[0.25em]
\partial_t^2\Phi_{S\pm,\ell m}-\partial_{r_*}^2\Phi_{S\pm,\ell m}+V_{S\pm}(r)\,\Phi_{S\pm,\ell m}&=0.\label{eq:master-S-rstar}
\end{align}
The explicit expressions for $V_{V\pm}$ and $V_{S\pm}$ are given in Equations~(5.23) and (6.23) of Ishibashi and Kodama~\cite{IshibashiKodama2011PTPS}. These potentials are smooth on $(0,\infty)$ and are rational functions of $r$ whose coefficients depend polynomially on $M$, $Q$, and $\ell$. The dependence on $\ell$ enters only through $k^2=\ell(\ell+1)$ and its shifted versions; see Appendix~\ref{app:appC}. In particular, for fixed $(M,Q)$ the potentials extend analytically to the superextremal regime $|Q|>M$, since their derivation uses only algebraic manipulations of $f(r)$ and its derivatives and does not rely on zeros of $f$ or the presence of a horizon.

In our superextremal \RN\ self-field problem, we use the optical radius $x$ introduced in~\eqref{eq:x-def}, which satisfies $\dd x/\dd r = f(r)^{-1}$ and hence coincides, up to an additive constant, with the tortoise coordinate $r_*$. In the superextremal regime $|Q|>M$, we have $f(r)>0$ for all $r>0$, and the integral defining $x$ converges at $r=0$. Thus, $x\colon(0,\infty)\to(0,\infty)$ is a smooth, strictly increasing diffeomorphism. After fixing the additive constant, we may regard $x$ as a global tortoise coordinate on the static region. For each $(\ell,m)$, we define optical master amplitudes by pulling back along $r\mapsto x$:
\begin{align*}
\phi_{+,\ell m}(t,x) := \Phi_{V+,\ell m}(t,r(x)),
\qquad
\phi_{-,\ell m}(t,x) := \Phi_{S-,\ell m}(t,r(x)),
\end{align*}
and similarly for the remaining two sector/branch combinations.
In terms of these fields the master equations \eqref{eq:master-V-rstar}--\eqref{eq:master-S-rstar} take the $(1+1)$-dimensional form
\begin{align}
\partial_t^2\phi_{\pm,\ell m}(t,x)
-\partial_x^2\phi_{\pm,\ell m}(t,x)
+V_\pm(x)\phi_{\pm,\ell m}(t,x)
=0,
\qquad
(t,x)\in\mathbb{R}\times(0,\infty).
\label{eq:master-wave}
\end{align}
with optical potentials
\begin{align}
V_\pm(x):=V_{\mathrm{sector},\pm}(r(x)).
\end{align}
In the channelwise calculations below, the fixed parity and angular labels are suppressed when no confusion can arise.

The master fields $\phi_{\pm,\ell m}$ are gauge-invariant by construction.\footnote{In this section, we consider smooth perturbations $(\delta g,\delta F)$
that are compactly supported on $\mathcal M$. The associated master fields
$\phi_{\pm,\ell m}$ are then smooth and compactly supported in $(t,x)$, so
$\phi_{\pm,\ell m}(t,\cdot)\in C_0^\infty(0,\infty)\subset H_0^1(0,\infty)\subset L^2(0,\infty)$ for each fixed $t$. Later we denote by
$\HE:=H_0^1(0,\infty)\times L^2(0,\infty)$
the completion of $C_0^\infty(0,\infty)\times C_0^\infty(0,\infty)$ in the norm induced by the
Einstein--Maxwell $T$-energy; belonging to $\HE$ is nothing more than
having finite $T$-energy in the sense of the quadratic form defined below. At
this stage we use only the inclusions
$C_0^\infty(0,\infty)\subset H_0^1(0,\infty)\subset L^2(0,\infty)$; no
self-adjoint realization of $-\partial_x^2+V_\pm$ or boundary condition at $x=0$ or $x=\infty$ is imposed.} The genuinely
radiative degrees of freedom are encoded in the gauge-invariant amplitudes
$\phi_{\pm,\ell m}$ with $\ell\ge2$, and these are the modes we focus on in the
remainder of the paper. The low-multipole sectors have a simple geometric meaning. For $\ell=0$
there is no nontrivial vector (axial) harmonic. The single scalar-type
gauge invariant reduces to a linearized shift of the ADM mass and
total charge $(M,Q)$. Such perturbations move within the two-parameter \RN\ family rather than exciting radiative degrees of freedom. For $\ell=1$, the vector (axial) sector is pure gauge. (The axial gauge-invariant variable vanishes once the constraints are imposed.) The scalar (polar) sector describes dipole perturbations of the Maxwell field and metric that can be
absorbed into a redefinition of the center-of-mass worldline and of the global
Maxwell gauge. Thus, after quotienting by the $\ell=0,1$ gauge and parameter variations, only the $\ell\ge2$ modes carry genuine radiative degrees of freedom.\footnote{For a systematic treatment of the low-multipole and exceptional modes in the Kodama--Ishibashi formalism, including the scalar $\ell=0,1$ sectors and their gauge structure, see Section~3 and Section~6.1 of~\cite{IshibashiKodama2011PTPS}.} In the superextremal, horizonless setting we work in the fixed-background single-source sector introduced above. We fix the background parameters $(M,Q)$ and quotient out the $\ell=0,1$ gauge and parameter variations, thereby removing the gauge and parameter directions from the radiative sector. The remaining $\ell\ge2$ modes are gauge-invariant master fields $\phi_{\pm,\ell m}$ on the optical half-line.

The master-field reduction is more than a gauge convenience: it gives an isometric identification of the physical degrees of freedom. The Kodama--Ishibashi construction, together with the quadratic action and
$T$-energy reduction carried out below, shows that, after quotienting by the linearized
diffeomorphism and Maxwell gauge freedoms and fixing the background parameters $(M,Q)$, the physical
linearized Einstein--Maxwell phase space around \RN\ is canonically isomorphic to the direct sum of
the master-field phase spaces, with the Einstein--Maxwell $T$-energy inducing exactly the quadratic
form on each mode. The map
\begin{align*}
  (\delta g,\delta F)\big/\text{gauge}
  \;\longleftrightarrow\;
  \{\phi_{\pm,\ell m},\partial_t\phi_{\pm,\ell m}\}_{\ell\ge2,\,m}
\end{align*}
is an isometry for the $T$-energy norm, so the master-wave dynamics on the optical half-line \emph{is exactly} the dynamics of the gauge-free linearized self-field. Viewed this way, the usual story is reversed. Rather than first solving the gauge-free Einstein--Maxwell equations and then fixing a gauge to obtain master waves, we take the gauge-invariant master fields as the primary degrees of freedom and reconstruct $(\delta g,\delta F)$ from them.

We now identify the conserved energy associated with the static Killing field $T=\partial_t$ directly from the Einstein--Maxwell action and express it in terms of the master variables $\phi_{\pm,\ell m}$. This provides the precise link between the geometric $T$-energy of the full perturbation $(\delta g,\delta F)$ and the quadratic form of the Schr\"odinger operator $-\partial_x^2+V_\pm$ on the optical half-line. The derivation is entirely local on $\mathcal M$ and uses only smooth,
compactly supported perturbations: at this stage we neither specify nor assume
any domain for the formal operator $-\partial_x^2+V_\pm$, and no boundary
condition at the optical endpoints is built into the construction.

The starting point is the Einstein--Maxwell action,
\begin{align*}
S[g,F]
=\frac{1}{16\pi}\int_{\mathcal{M}}\!\bigl(R(g)-F_{\mu\nu}F^{\mu\nu}\bigr)\sqrt{-g}\,\dd^4x,
\end{align*}
where $R(g)$ is the Ricci scalar of the spacetime metric $g$ and $F$ is the Maxwell 2-form. Let $(g,F)$ be the exact superextremal \RN\ background, and consider a one-parameter family of smooth fields
\begin{align}
g_\varepsilon=g+\varepsilon\,h+\Order{\varepsilon^{2}},\qquad
F_\varepsilon=F+\varepsilon\,\delta F+\Order{\varepsilon^{2}},
\end{align}
with $(h,\delta F)$ smooth and compactly supported in $\mathcal M$. Expanding $S[g_\varepsilon,F_\varepsilon]$ to second order in $\varepsilon$ gives $S[g,F] + \varepsilon S^{(1)}[h,\delta F] + \tfrac12\varepsilon^2 S^{(2)}[h,\delta F] + \dots$. Since $(g,F)$ solves the field equations, $S^{(1)}=0$, and the linearized Einstein--Maxwell equations are the Euler--Lagrange equations of the quadratic functional $S^{(2)}$. Writing $S^{(2)}=\int\mathcal{L}^{(2)}\sqrt{-g}\,\dd^4x$, we obtain a quadratic Lagrangian density $\mathcal{L}^{(2)}[h,\delta F]$ on the fixed background.

For any background Killing vector field $X$, applying Noether's theorem at the quadratic level yields a symmetric, bilinear stress tensor $T^{(2)}_{\mu\nu}[h,\delta F]$ and an associated conserved current
\begin{align}
J_\mu^X[h,\delta F]
:=T^{(2)}_{\mu\nu}[h,\delta F]\;X^\nu,
\qquad
\nabla^\mu J^X_\mu[h,\delta F]=0,
\label{eq:JT-def}
\end{align}
where $\nabla$ is the Levi--Civita connection of the background metric $g$. The current $J^X$ is quadratic in $(h,\delta F)$, and the conservation law follows from the Killing property of $X$ together with the linearized equations of motion.

In our static \RN\ background we take $X=T=\partial_t$. For a static slice $\Sigma_t:=\{t=\mathrm{const}\}\subset\mathcal M$ with future-directed unit normal $n^\mu$ and induced volume form $\dd\Sigma$, the associated $T$-energy on $\Sigma_t$ is
\begin{align}
E_{\mathrm{EM}}[h,\delta F](t)
:=\int_{\Sigma_t} J^T_\mu[h,\delta F]\;n^\mu\,\dd\Sigma.
\end{align}
By the divergence theorem and the conservation of $J^T$ this quantity is
independent of $t$, provided the fluxes through the lateral boundaries vanish. In the
superextremal case there is no horizon, so for compactly supported perturbations
the only possible boundaries are the optical endpoint $x=0$ and null infinity; for such perturbations the flux
through both boundaries vanishes by support.

Inserting the harmonic decomposition and the gauge-invariant master variables into $S^{(2)}$ and integrating over the sphere yields a sum over decoupled radiative channels. Suppressing the parity label in the mode notation, one has
\begin{align}
S^{(2)}
= \frac12\sum_{\pm,\ell\ge2,m}\int_{\mathbb{R}}\!L_{\pm,\ell m}(t)\,\dd t,
\end{align}
where the sum includes both parity sectors, all $\ell\ge2$, and the two Einstein--Maxwell polarizations in each parity. The mode Lagrangian is given by
\begin{align}
L_{\pm,\ell m}(t)
= \frac{1}{2}\int_0^\infty\!
\Big(
\abs{\partial_t\phi_{\pm,\ell m}(t,x)}^2
-\abs{\partial_x\phi_{\pm,\ell m}(t,x)}^2
-V_\pm(x)\,\abs{\phi_{\pm,\ell m}(t,x)}^2
\Big)\,\dd x.
\label{eq:L-mode}
\end{align}
The locality of the Kodama--Ishibashi transformation on the $(t,r)$ orbit space, together with the orthogonality of the spherical harmonics, implies that for each fixed $(\pm,\ell,m)$ the mode Lagrangian $L_{\pm,\ell m}$ depends only on the master amplitude $\phi_{\pm,\ell m}(t,x)$ and its first derivatives. The orthogonality relations for the scalar and vector harmonics ensure that different $(\ell,m)$ modes decouple. After discarding total derivatives in $t$ and $x$, the requirement that the Euler--Lagrange equation reproduce the master equation \eqref{eq:master-wave} and that the kinetic term be positive fixes the relative coefficients uniquely.\footnote{Compare the discussion of the $T$-energy for scalar waves on black-hole backgrounds in Dafermos--Rodnianski. The gauge-invariant variables used here coincide with those of Ishibashi--Kodama for four-dimensional Einstein--Maxwell backgrounds; the reduction here is in exact analogy with the Regge--Wheeler--Zerilli reduction.}

The same harmonic decomposition can be applied to the quadratic Einstein--Maxwell stress tensor $T^{(2)}_{\mu\nu}[h,\delta F]$ and to the associated $T$--current $J^T_\mu[h,\delta F]$ defined in \eqref{eq:JT-def}. Since $T^{(2)}_{\mu\nu}$ is quadratic in $(h,\delta F)$ and the transformation from the harmonic coefficients of $(h,\delta F)$ to the master variables $\phi_{\pm,\ell m}$ is linear and invertible on the radiative subspace for each fixed $(\ell,m)$, the flux of $J^T$ through a static slice $\Sigma_t$ decomposes orthogonally as a sum over $(\pm,\ell,m)$. Inserting the inverse Kodama--Ishibashi reconstruction of $(h,\delta F)$ from the master fields into $J^T_\mu n^\mu$, integrating over $\mathbb{S}^2$, and using the orthogonality relations for the scalar and vector harmonics, one obtains
\begin{align*}
E_{\mathrm{EM}}[h,\delta F](t)
=\int_{\Sigma_t}J^T_\mu[h,\delta F]\;n^\mu\,\dd\Sigma
= \sum_{\pm,\ell\ge2,m} E_{\pm,\ell m}(t),
\end{align*}
where $E_{\pm,\ell m}(t)$ is exactly the canonical Hamiltonian associated with the mode Lagrangian \eqref{eq:L-mode}. Thus, the geometric Einstein--Maxwell $T$-energy decomposes \emph{exactly} into a sum of canonical kinetic terms and Schr\"odinger-type spatial energies for the master operators $-\partial_x^2+V_\pm$ on the optical half-line.

For the canonical Hamiltonian calculation, we may treat each mode $\phi_{\pm,\ell m}$ as real-valued: the reality of the full metric and Maxwell perturbations is ensured by the standard pairing of $(\ell,m)$ and $(\ell,-m)$ modes in the spherical-harmonic decomposition. For each fixed $(\pm,\ell,m)$, the canonical momentum is $\pi_{\pm,\ell m}(t,x) = \partial_t\phi_{\pm,\ell m}(t,x)$, and the associated Hamiltonian is
\begin{align}
E_{\pm,\ell m}(t)
&:=\int_0^\infty\!
\Bigl(
\pi_{\pm,\ell m}(t,x)\,\partial_t\overline{\phi_{\pm,\ell m}(t,x)}
-\mathcal{L}_{\pm,\ell m}(t,x)
\Bigr)\,\dd x\nonumber\\
&= \frac{1}{2}\int_0^\infty\!
\Big(
\abs{\partial_t\phi_{\pm,\ell m}(t,x)}^2
+\abs{\partial_x\phi_{\pm,\ell m}(t,x)}^2
+V_\pm(x)\,\abs{\phi_{\pm,\ell m}(t,x)}^2
\Big)\,\dd x.
\label{eq:E-mode}
\end{align}

Differentiating $E_{\pm,\ell m}(t)$ in time and integrating by parts in $x$ shows that the bulk terms cancel once the equation of motion is used. For $\phi_{\pm,\ell m}\in C^\infty(\mathbb{R}_t;C_0^\infty(0,\infty))$, a calculation gives
\begin{align*}
\frac{\dd}{\dd t}E_{\pm,\ell m}(t)
&= \Re\int_0^\infty\!\Bigl(
\partial_t^2\phi_{\pm,\ell m}\,\overline{\partial_t\phi_{\pm,\ell m}}
+\partial_x\partial_t\phi_{\pm,\ell m}\,\overline{\partial_x\phi_{\pm,\ell m}}
+V_\pm\,\partial_t\phi_{\pm,\ell m}\,\overline{\phi_{\pm,\ell m}}
\Bigr)\,\dd x.
\end{align*}
Integrating the second term by parts in $x$ and using the master equation \eqref{eq:master-wave} to substitute
$\partial_t^2\phi_{\pm,\ell m}
=\partial_x^2\phi_{\pm,\ell m}-V_\pm\phi_{\pm,\ell m}$ shows that the potential terms cancel, and the remaining terms combine to give zero. Thus, for smooth, compactly supported solutions $\phi_{\pm,\ell m}$, the boundary terms vanish and we find $\frac{\dd}{\dd t}E_{\pm,\ell m}(t)=0$.\footnote{The analysis is carried out for smooth, compactly supported
solutions, for which all integrations by parts are justified and the boundary
terms at $x=0$ and $x=\infty$ vanish identically by support. At this stage $q_\pm$ is defined on $C_0^\infty(0,\infty)$; its closed extension and the associated Friedrichs operator will be constructed in Sections~\ref{sec:hardy}--\ref{sec:friedrichs}. Once that is done, the same
energy identity extends by density and unitarity to all solutions in the
corresponding energy space.} On a static slice $t=t_0$ one can then separate the ``kinetic'' and ``spatial'' contributions to $E_{\pm,\ell m}$. Writing $u(x):=\phi_{\pm,\ell m}(t_0,x)$ and $\dot u(x):=\partial_t\phi_{\pm,\ell m}(t_0,x)$, we obtain from~\eqref{eq:E-mode} that
\begin{align}
E_{\pm,\ell m}(t_0)
= \frac12\norm{\dot u}_{\Ltwo}^2+q_\pm[u],
\end{align}
where the quadratic form
\begin{align}
q_\pm[u]
:=\frac12\int_0^\infty\!\bigl(\abs{u'(x)}^2+V_\pm(x)\,\abs{u(x)}^2\bigr)\,\dd x,
\qquad
u\in C_0^\infty(0,\infty),
\label{eq:qpm-def}
\end{align}
is the spatial form entering the $T$-energy.\footnote{At this stage $q_\pm$ is defined only on the test space
$C_0^\infty(0,\infty)\subset L^2(0,\infty)$, and we do not yet assume that
it is closed or that $-\partial_x^2+V_\pm$ has a preferred self-adjoint
realization. In Sections~\ref{sec:hardy} and~\ref{sec:friedrichs}, together
with Appendices~\ref{app:appB}--\ref{app:appC}, we use the Hardy inequality
on $(0,\infty)$ and the potential hypotheses \eqref{eq:V-core}--\eqref{eq:short-range}
to show that $q_\pm$ is semibounded and closable on $\Ltwo$, that its closed extension
has domain $\Honezero$, and that the associated Friedrichs operator $H_\pm^\Fried$
is the canonical self-adjoint realization of $-\partial_x^2+V_\pm$ selected
by the Einstein--Maxwell $T$-energy, without imposing any additional boundary
condition at $x=0$ or $x=\infty$.}

This is consistent with the independent Hamiltonian reduction and intertwiner construction of Dotti and Gleiser, who obtain the same master operators in their treatment of the interior and superextremal regions of \RN\ spacetime; see, in particular, their reduction of the linearized field equations to $(1+1)$-dimensional wave equations with Schr\"odinger-type spatial operators. In our setting the static Coulomb field and its nonlinear gravitational backreaction are fully incorporated into the exact background, and the master variables $\phi_{\pm,\ell m}$ describe perturbations of this
self-field rather than test fields on a fixed geometry.

The full perturbative problem has therefore been reduced, without an independently imposed endpoint law, to determining whether the action-derived forms $q_\pm$ close on the optical half-line.

\section{Potential asymptotics and Hardy closure}\label{sec:hardy}

Throughout this section, the same
symbol $q_\pm$ denotes first the form on
$C_0^\infty(0,\infty)$ and subsequently its unique closed extension.
The stated domain identifies which stage of the construction is meant.

In this section, we prove that the core quadratic forms $q_\pm$ are closable and that their closed extensions are semibounded forms on $H_0^1(0,\infty)$. The point is that the effective potentials $V_\pm(x)$ in the master equations~\eqref{eq:master-wave} have just enough regularity and decay to be controlled by the sharp Hardy inequality on the half-line $(0,\infty)$.

Two regions govern the analysis: the optical apex and spatial infinity. Near the optical apex $x=0$, both $V_\pm$ look like inverse-square cores of Hardy type:
\begin{align*}
  V_\pm(x)\simeq C_\pm\,x^{-2},\qquad x\downarrow0,
\end{align*}
with $C_+<0<C_-$ and both $C_\pm$ strictly inside the Hardy window $(-1/4,\infty)$ (in fact $C_\pm\in(-1/4,3/4)$). Intuitively, the plus channel $(+)$ sees a shallow attractive $1/x^2$ well and the minus channel $(-)$ a repulsive $1/x^2$ barrier, but in both cases the singularity is too weak to upset the basic $H_0^1$ control coming from the Dirichlet energy. Toward spatial infinity, the potentials look like the familiar centrifugal barrier
$
  \frac{\ell(\ell+1)}{x^2}
$
plus a short-range tail $W_\pm(x)$ that lies in $L^1(1,\infty)$ and is therefore too small to spoil the Hardy control. 

We first promote these two features to abstract hypotheses on $V_\pm$, use them to prove that each $q_\pm$ is closable and semibounded, and then check---using the explicit formulas of Ishibashi--Kodama and Dotti--Gleiser together with the optical asymptotics in Appendix~\ref{app:appC}---that the hypotheses are indeed satisfied by the Einstein--Maxwell self-field on the superextremal \RN\ background.

\subsection{Form-boundedness at the apex}

First, near the optical apex $x=0$ the potentials have inverse-square cores of Hardy type:
\begin{equation}
\begin{aligned}
V_+(x)
&=
-\frac{2}{9}x^{-2}+o(x^{-2}),
&
V_-(x)
&=
\frac{4}{9}x^{-2}+o(x^{-2}).
\end{aligned}
\label{eq:V-core}
\end{equation}
so that, for each sign, there exists a constant $C_\pm\in(-\tfrac14,\tfrac34)$ and a function $r_\pm$ with $V_\pm(x)=C_\pm x^{-2}+r_\pm(x)$ and $\lim_{x\downarrow0}x^{2}r_\pm(x)=0$. These leading coefficients $-2/9$ and $4/9$ are consistent with the superextremal \RN\ master potentials derived in Ishibashi and Kodama~\cite{IshibashiKodama2011PTPS} and Dotti and Gleiser~\cite{DottiGleiser2010} for electromagnetic and gravitational multipoles after passing to the optical coordinate; see equations (6.32)--(6.33) of Ishibashi and Kodama~\cite{IshibashiKodama2011PTPS} for the parameter relations and Appendix~\ref{app:appB} for the detailed computation of the asymptotics. It is notationally convenient to parametrize the core coefficient by
\begin{align}
C_\pm=\nu_\pm^2-\tfrac14,\qquad \nu_\pm\in(0,1),
\label{eq:nu}
\end{align}
so that $V_\pm(x)=(\nu_\pm^2-\tfrac14)x^{-2}+o(x^{-2})$ as $x\downarrow0$. In particular, $C_\pm>-\tfrac14$ lies strictly inside the one-dimensional Hardy window in the sense of~Brézis and Marcus~\cite{BrezisMarc} (see Appendix~\ref{app:appB} for the full derivation of \eqref{eq:V-core} and \eqref{eq:nu}).
For the \RN\ self-field, the abstract constants become very concrete: in Appendix~\ref{app:appB}, we show that
\begin{align*}
  C_+ = -\frac{2}{9},\qquad C_- = \frac{4}{9},
\end{align*}
so that
\begin{align*}
  \nu_+ = \frac{1}{6},\qquad \nu_- = \frac{5}{6}.
\end{align*}
Thus, the $+$-channel contains a genuinely attractive inverse-square core, while the $-$-channel contains a repulsive core. The decisive identity is
\begin{align*}
-\frac29=-\frac14+\frac1{36}.
\end{align*}
The naked singularity therefore generates the local obstruction and its exact test simultaneously: the optical geometry fixes the attractive coefficient $-2/9$, while the critical inverse-square threshold $-1/4$ determines whether that coefficient destroys the energy principle. The positive margin $1/36$ proves that the attraction cannot overcome the Dirichlet contribution to the Einstein--Maxwell $T$-energy. Consequently, the singular form is semibounded and closable, so the action-derived finite-energy completion exists. In this precise sense, the singularity generates its own test and passes it.

The core form $q_\pm$ introduced in \eqref{eq:qpm-def} is defined on $C_0^\infty(0,\infty)$. We now show that assumptions \eqref{eq:V-core} and \eqref{eq:short-range} are exactly what is needed to close it as a semibounded quadratic form on $\Honezero$ and to identify the associated Friedrichs realization as the self-adjoint generator of the self-field evolution. The key input is the sharp Hardy inequality on the half-line. This is the standard quadratic-form construction of the Friedrichs extension for semibounded Schr\"odinger operators on the half-line; see, for example, the general discussion in~Section~2.3 of \cite{Teschl}.

On the half-line, the sharp Dirichlet Hardy inequality states that
\begin{align}
\int_0^\infty \!\abs{u'(x)}^2\,\dd x
\ \ge\ \frac14\!\int_0^\infty \frac{\abs{u(x)}^2}{x^2}\,\dd x,
\qquad u\in H_0^1(0,\infty).
\label{eq:hardy-main}
\end{align}
This is the half-line version of the sharp inequality on bounded intervals given in Equation~(0.1) of \cite{BrezisMarc}, obtained by exhaustion and scaling.
Equation~\eqref{eq:hardy-main} says that a function in $H_0^1(0,\infty)$ cannot pile up too much mass near the apex without paying a fixed derivative cost: the weight $x^{-2}$ measures how sharply $u$ concentrates as $x\downarrow0$. The Hardy term acts as a universal ``centrifugal barrier'' already present in the free Dirichlet energy on the half-line.

A simple way to see the borderline behavior is to test \eqref{eq:hardy-main} on the model functions
\begin{align*}
  u_\varepsilon(x) = x^{\frac12+\varepsilon}\chi(x),
\end{align*}
where $\chi$ is a smooth cutoff equal to $1$ near $0$ and compactly supported. As $\varepsilon\downarrow0$ the ratio of the right-hand side of \eqref{eq:hardy-main} to the left-hand side tends to $1/4$, illustrating both the sharpness of the Hardy constant and the special role of the exponent $1/2$ that will reappear in the analysis of the inverse-square cores. For a rigorous discussion of sharpness and non-attainment of the constant $1/4$ on $(0,1)$, together with logarithmic refinements of Hardy's inequality, see~Appendix~A of \cite{BrezisMarc}, in particular Lemmas~A.1--A.2 and Corollary~A.4.
Combining \eqref{eq:hardy-main} with the core decomposition $V_\pm(x)=C_\pm x^{-2}+r_\pm(x)$, we obtain for any $u\in C_0^\infty(0,\infty)$
\begin{align}
\int_0^\infty\!\Bigl(\abs{u'(x)}^2+C_\pm x^{-2}\abs{u(x)}^2\Bigr)\,\dd x
&\ge\Bigl(\tfrac14+C_\pm\Bigr)\int_0^\infty x^{-2}\abs{u(x)}^2\,\dd x\nonumber\\
&\ge\nu_\pm^{2}\!\int_0^\infty x^{-2}\abs{u(x)}^2\,\dd x
\ \ge\ 0,
\label{eq:core-lower}
\end{align}
since $\tfrac14+C_\pm=\nu_\pm^2>0$ by \eqref{eq:nu}. This is the half-line analog of the boundary-singularity Hardy inequalities on bounded domains studied in~Section~0 of \cite{BrezisMarc}, where the same critical constant $1/4$ controls the behavior near the boundary. In particular, the inverse-square cores with $C_\pm>-\tfrac14$ are not destabilizing: they contribute a nonnegative piece that is dominated by the Hardy weight $x^{-2}$. From the energy point of view, \eqref{eq:hardy-main} is exactly what prevents the Coulomb-type core at the optical apex from blowing up the $T$-energy: we never assume that $u$ has any special behavior at $x=0$ or $x=\infty$; we merely require $u\in C_0^\infty(0,\infty)$ and estimate the form. The Hardy control and the short-range tail then show that the quadratic form extends to a closed, semibounded form on $L^2(0,\infty)$ with domain $H_0^1(0,\infty)$, so the admissible endpoint behavior is completely encoded in the closure of the energy rather than in an imposed boundary condition.

By construction, $r_\pm$ is less singular than $x^{-2}$ at the apex. Thus, there exist $x_0\in(0,1)$ and $M_\pm>0$ such that
\begin{align}
\abs{r_\pm(x)}\le M_\pm\quad\text{for }x\in(0,x_0).
\end{align}
Applying the elementary bound
\begin{align}
\int_0^{x_0}\!\abs{u(x)}^2\,\dd x
\le x_0^2\!\int_0^{x_0}\frac{\abs{u(x)}^2}{x^2}\,\dd x
\le 4x_0^2\!\int_0^{x_0}\abs{u'(x)}^2\,\dd x,
\end{align}
which follows directly from \eqref{eq:hardy-main}, we obtain
\begin{align}
\bigg|\int_0^{x_0} r_\pm(x)\,\abs{u(x)}^2\,\dd x\bigg|
\le 4M_\pm x_0^2\!\int_0^{x_0}\abs{u'(x)}^2\,\dd x.
\label{eq:core-remainder}
\end{align}
Given any $\varepsilon>0$, we can choose $x_0$ small enough that $4M_\pm x_0^2<\varepsilon$,
and hence
\begin{align}
\bigg|\int_0^{x_0} r_\pm(x)\,\abs{u(x)}^2\,\dd x\bigg|
\le \varepsilon\!\int_0^{x_0}\abs{u'(x)}^2\,\dd x + C_\varepsilon\!\int_0^{x_0}\abs{u(x)}^2\,\dd x,
\end{align}
for some $C_\varepsilon\ge0$, which is the standard infinitesimal form-bound condition.
Moreover, by local boundedness of $V_\pm$ away from the apex, $r_\pm$ is bounded on
$[x_0,2]$, so the integral over $[x_0,2]$ contributes only an $L^2$ term.
In other words, $r_\pm$ is infinitesimally form-bounded with respect to the pure Dirichlet
form near the apex.

\subsection{Short-range tails as infinitesimal form perturbations}

Second, toward spatial infinity, the potentials split into a centrifugal barrier plus a short-range tail:
\begin{align}
V_\pm(x)=\frac{\ell(\ell+1)}{x^{2}}+W_\pm(x),\qquad
W_\pm&\in L^{1}\!\big((1,\infty)\big),\label{eq:short-range}
\end{align}
for each fixed multipole $\ell\ge2$. In the Einstein--Maxwell case, this follows from the explicit Kodama--Ishibashi master potentials specialized to $(\Lambda,K,n)=(0,1,2)$ and from the asymptotically flat optical map; see \eqref{eq:x-apex-rtoinf} and Appendix~\ref{app:appC}, together with the detailed Reissner--Nordstr\"om analysis in~Section~3 of \cite{IshibashiKodama2011PTPS} and the superextremal and inner-region potentials described in~Section~2 of \cite{DottiGleiser2010}.
Thus, for a fixed angular momentum $\ell$, the far-field profile of $V_\pm$ can be pictured as a slightly perturbed centrifugal barrier. If one temporarily ignores the tail and sets $W_\pm\equiv0$, the model operator
$
  -\partial_x^2 + \frac{\ell(\ell+1)}{x^2}
$
has explicit power-law solutions $x^{\ell+1}$ and $x^{-\ell}$, corresponding to the familiar rise and decay of spherical waves. The short-range term $W_\pm$ only produces lower-order corrections that are integrable in $x$; it can be thought of as a gentle bump riding on top of the centrifugal profile, too weak to alter the Hardy-controlled closure or the action-selected finite-energy space.

To control the tails $W_\pm$ on $(1,\infty)$, we use the standard fact that any $L^1$ potential is infinitesimally form-bounded with respect to the Dirichlet form on $(1,\infty)$ (see, for instance, Lemma~9.33 in \cite{Teschl} for the line or Inoue--Richard~\cite{InouRichard} for closely related half-line models). For completeness, we sketch the argument.

For $u\in C_0^\infty(1,\infty)$ and $x>1$, we start from
\begin{align}
\abs{u(x)}^2
&= -\int_x^\infty \frac{\dd}{\dd y}\bigl(\abs{u(y)}^2\bigr)\,\dd y
= -2\Re\!\int_x^\infty u'(y)\,\overline{u(y)}\,\dd y,
\end{align}
and by Cauchy--Schwarz,
\begin{align}
\sup_{x\ge1}\abs{u(x)}^2
&\le 2\int_1^\infty\!\abs{u'(y)}\,\abs{u(y)}\,\dd y
\le 2\biggl(\int_1^\infty\!\abs{u'(y)}^2\,\dd y\biggr)^{\!1/2}
\biggl(\int_1^\infty\!\abs{u(y)}^2\,\dd y\biggr)^{\!1/2}.
\end{align}
Hence
\begin{align}
\int_1^\infty\!\abs{W_\pm(x)}\,\abs{u(x)}^2\,\dd x
\le \norm{W_\pm}_{L^1(1,\infty)}\,
\sup_{x\ge1}\abs{u(x)}^2
\le 2\norm{W_\pm}_{L^1(1,\infty)}\,
\norm{u'}_{L^2(1,\infty)}\,\norm{u}_{L^2(1,\infty)}.
\end{align}
This is the usual $L^1$--$L^\infty$ estimate.
Applying the elementary inequality $2ab\le\varepsilon a^2+\varepsilon^{-1}b^2$ with $a=\norm{u'}_{L^2(1,\infty)}$ and $b=\norm{u}_{L^2(1,\infty)}$, we obtain for every $\varepsilon>0$ a constant $C_{\varepsilon,\pm}\ge0$ such that
\begin{align}
\bigg|\int_1^\infty W_\pm(x)\,\abs{u(x)}^2\,\dd x\bigg|
\le \varepsilon\!\int_1^\infty\abs{u'(x)}^2\,\dd x
+ C_{\varepsilon,\pm}\!\int_1^\infty\abs{u(x)}^2\,\dd x.
\label{eq:tail-formbound}
\end{align}
Thus, $W_\pm$ is infinitesimally form-bounded with respect to the Dirichlet form on $(1,\infty)$.
In more concrete terms, \eqref{eq:tail-formbound} says that the energy stored in the tail region $x\gg1$ is controlled by the gradient part of the Dirichlet form up to an arbitrarily small multiple, plus an $L^2$ remainder. One may view $W_\pm$ as a lower-order perturbation that slightly reshapes the centrifugal barrier but never overwhelms it: any loss in coercivity coming from the negative part of $W_\pm$ can be absorbed into the $\varepsilon\!\int\!\abs{u'}^2$ term.

The centrifugal term $\ell(\ell+1)x^{-2}$ in \eqref{eq:short-range} is of the same type as the apex core, but with positive coefficient; combining it with \eqref{eq:hardy-main} yields
\begin{align}
\int_1^\infty\!\Bigl(\abs{u'}^2+\frac{\ell(\ell+1)}{x^2}\abs{u}^2\Bigr)\,\dd x
\ge \int_1^\infty\abs{u'}^2\,\dd x,
\end{align}
so it simply strengthens the coercivity of the Dirichlet form at infinity.

\subsection{Global estimate on optical half-line}

The Hardy and short-range estimates of the previous subsection give precise control of the quadratic forms associated with the master potentials in two complementary regions: near the optical apex, where the inverse-square core dominates, and in the far field, where the potential is a small perturbation of the centrifugal barrier. The next step is to splice these local controls into a single global estimate on the half-line and to interpret the resulting form within the standard Kato--Lions--Milgram--Nelson framework.

Combining the local estimates \eqref{eq:core-lower}, \eqref{eq:core-remainder}, and \eqref{eq:tail-formbound} with a smooth cutoff partition of unity on $(0,\infty)$ yields the desired global bound. Let $\chi_0,\chi_1\in C^\infty([0,\infty))$ satisfy
\begin{align*}
0\le\chi_j\le1,\qquad
\chi_0(x)=1\ \text{for }x\le1,\qquad
\chi_0(x)=0\ \text{for }x\ge2,
\end{align*}
and
\begin{align*}
\chi_0^2+\chi_1^2\equiv1.
\end{align*}
For $u\in C_0^\infty(0,\infty)$ we split
\begin{align*}
u_j:=\chi_j u,\qquad j=0,1,
\end{align*}
so that $u=u_0+u_1$ is decomposed into an ``inner'' piece near the apex and an ``outer'' piece supported in the tail region. The IMS localization formula then gives a clean decomposition of the kinetic term:
\begin{align}
\int_0^\infty \abs{u'(x)}^2\dd x
= \sum_{j=0}^1\int_0^\infty\abs{(\chi_j u)'}^2\dd x
- \int_0^\infty\bigl(\chi_0'(x)^2+\chi_1'(x)^2\bigr)\abs{u(x)}^2\dd x.
\end{align}
The last term is the usual ``IMS error'': it only depends on the fixed cutoffs and is thus a bounded $L^2$ contribution. The potential part localizes exactly, because
$\abs{u}^2=\abs{u_0}^2+\abs{u_1}^2$:
\begin{align}
\int_0^\infty V_\pm(x)\abs{u(x)}^2\dd x
= \sum_{j=0}^1\int_0^\infty V_\pm(x)\abs{u_j(x)}^2\dd x.
\end{align}
Putting these together, the quadratic form can be written as
\begin{align}
q_\pm[u]
=\frac12\sum_{j=0}^1\int_0^\infty\Bigl(\abs{u_j'}^2+V_\pm\abs{u_j}^2\Bigr)\dd x
- \frac12\int_0^\infty\bigl(\chi_0'^2+\chi_1'^2\bigr)\abs{u}^2\dd x.
\end{align}
Since $\chi_j$ and their derivatives are smooth and compactly supported, the IMS error term
\begin{align*}
\int_0^\infty\bigl(\chi_0'^2+\chi_1'^2\bigr)\abs{u}^2\dd x
\end{align*}
is bounded by $C\int_0^\infty\abs{u}^2\dd x$ for some fixed $C\ge0$.

On $\supp u_0\subset(0,2)$ the decomposition $V_\pm=C_\pm x^{-2}+r_\pm(x)$, together with \eqref{eq:core-lower} and \eqref{eq:core-remainder}, shows that
\begin{align}
\int_0^\infty\Bigl(\abs{u_0'}^2+V_\pm\abs{u_0}^2\Bigr)\dd x
\ge -C_0\int_0^\infty\abs{u_0}^2\,\dd x
\end{align}
for some $C_0\ge0$. Thus, the inverse-square core, although attractive in one sector, is tamed by the Hardy inequality up to a harmless $L^2$ contribution.

On $\supp u_1\subset(1,\infty)$ we use the centrifugal-plus-tail decomposition
\begin{align*}
V_\pm=\ell(\ell+1)x^{-2}+W_\pm
\end{align*}
and the Hardy inequality together with \eqref{eq:tail-formbound} to obtain, for any fixed $\varepsilon\in(0,1)$,
\begin{align}
\int_0^\infty\Bigl(\abs{u_1'}^2+V_\pm\abs{u_1}^2\Bigr)\dd x
\ge (1-\varepsilon)\int_0^\infty\abs{u_1'}^2\,\dd x - C_1\int_0^\infty\abs{u_1}^2\,\dd x
\end{align}
for some $C_1\ge0$. In words, the centrifugal barrier and the short-range tail combine to give a nearly coercive estimate at infinity: one can retain an arbitrarily large fraction of the free Dirichlet energy, again at the price of a bounded $L^2$ term.

Summing the localized estimates for $u_0$ and $u_1$ and absorbing the bounded IMS error term
\begin{align*}
\frac12\int_0^\infty(\chi_0'^2+\chi_1'^2)\abs{u}^2\dd x
\end{align*}
into a global $L^2$ bound, we find constants $a>0$ and $b\in\mathbb{R}$ such that
\begin{equation}
q_\pm[u]
\ge a\norm{u'}_{L^2}^2 - b\norm{u}_{L^2}^2,
\qquad u\in C_0^\infty(0,\infty).
\label{eq:semicoercive-main}
\end{equation}
This is the global semicoercive estimate: apart from a lower-order $L^2$ term, the quadratic form controls the full $H_0^1$ gradient norm on the half-line.

It is convenient to read \eqref{eq:semicoercive-main} in terms of the free Dirichlet form
\begin{align}
  q_0[u] := \frac12\int_0^\infty\abs{u'(x)}^2\,\dd x.
\end{align}
The estimate \eqref{eq:semicoercive-main} can then be rewritten as a comparison between $q_\pm$ and $q_0$: for every $\varepsilon>0$ there exists $C_\varepsilon>0$ such that
\begin{align}
  q_\pm[u] \;\ge\; (1-\varepsilon)\,q_0[u] - C_\varepsilon\norm{u}_{L^2}^2.
\end{align}
Thus, the potential contribution is allowed to be negative, but only in a controlled way: it can at most reduce the free Dirichlet energy by a fixed fraction and a lower-order $L^2$ term. Because the constant in front of $q_0$ is strictly positive, the negative part of the potential has relative form bound strictly smaller than $1$ with respect to $q_0$ in the sense of Kato. In more geometric language, $q_\pm$ is an \emph{almost Dirichlet} form: the basic $H^1_0$ control coming from the gradient term cannot be destroyed by the attractive part of $V_\pm$, so the optical half-line never supports arbitrarily deep ``energy wells'' for the master fields.

\subsection{KLMN closure of the $T$-energy forms}

On the functional-analytic side, the semicoercive estimate places the optical potentials $V_\pm$ squarely within the classical Kato framework for small form perturbations of the free Dirichlet form $q_0$.
\begin{align*}
q_\pm[u] \;\ge\; (1-\varepsilon)\,q_0[u] - C_\varepsilon\norm{u}_{L^2}^2
\end{align*}
This estimate says that the negative part of $V_\pm$ has relative form bound strictly below one with respect to $q_0$: it may weaken, but cannot destroy, the $H_0^1$ control supplied by the gradient term. The Kato--Lions--Milgram--Nelson theorem (Kato, Chap.~VI) is then used in the standard way: it guarantees that adding $V_\pm$ to the free Dirichlet form still produces a densely defined, closed, semibounded form on $\Honezero$. This is exactly the step that promotes the test-space definition of $q_\pm$ to the closed $T$-energy form, denoted by the same symbol, on the optical half-line; at this stage no choice of self-adjoint extension has been made.

Recall from \eqref{eq:qpm-def} that on the Hilbert space $\Ltwo$ the core form is
\begin{align}
q_\pm[u]
:=\frac12\int_0^\infty\!\bigl(\abs{u'(x)}^2+V_\pm(x)\,\abs{u(x)}^2\bigr)\,\dd x,
\qquad
u\in C_0^\infty(0,\infty).
\end{align}
The Hardy and short-range estimates of the previous subsection show that $q_\pm$ is closable. Its densely defined, closed, symmetric, semibounded extension is denoted by the same symbol and has form domain
\begin{align}
\Dom(q_\pm)=\Honezero.
\end{align}
There exist constants $a>0$ and $b\in\mathbb{R}$ such that
\begin{align}
q_\pm[u]\ \ge\ a\,\norm{u'}_{L^2}^2 - b\,\norm{u}_{L^2}^2,
\qquad u\in\Honezero,
\end{align}
and the associated form norm
\begin{align}
\norm{u}_{q_\pm}^2 := q_\pm[u]+(\abs{b}+1)\,\norm{u}_{L^2}^2
\end{align}
is equivalent to the standard $\Honezero$ norm. In particular, $C_0^\infty(0,\infty)$ is dense in $\Honezero$ with respect to $\norm{\cdot}_{q_\pm}$, so it is a core for the closed form $q_\pm$.

The Hardy and tail estimates therefore complete the analytic task of this section: each core form $q_\pm$ closes uniquely to a densely defined, semibounded form on $H_0^1(0,\infty)$. The next section identifies the self-adjoint operator determined by this closed physical energy and determines the endpoint behavior encoded in its domain.

\section{Friedrichs extension and silent apex}\label{sec:friedrichs}

Section~\ref{sec:hardy} constructed the unique closed semibounded form $q_\pm$ with domain $H_0^1(0,\infty)$. We now identify the self-adjoint operator determined by that form and determine the apex behavior encoded in its domain. The passage introduces no new endpoint condition: the operator, its form domain, and its boundary behavior are consequences of the already completed Einstein--Maxwell $T$-energy.

The closed physical form $q_\pm$ determines the Friedrichs Hamiltonian $H_\pm^\Fried$, and $2q_\pm$ is its standard operator quadratic form. Recall from \eqref{eq:qpm-def} that the core form on $\Ltwo$ is
\begin{align*}
q_\pm[u]
:=\frac12\int_0^\infty\!\bigl(\abs{u'(x)}^2+V_\pm(x)\,\abs{u(x)}^2\bigr)\,\dd x,
\qquad
u\in C_0^\infty(0,\infty),
\end{align*}
and the preceding section proves that its closed extension, denoted by the same symbol, is a densely defined, closed, symmetric, semibounded form on $\Ltwo$ with form domain
\begin{align}
\Dom(q_\pm)=\Honezero.
\end{align}
More precisely, there exist constants $a>0$ and $b\in\mathbb R$ such that
\begin{align}
q_\pm[u]\ge a\norm{u'}_{L^2}^2-b\norm{u}_{L^2}^2,
\qquad u\in\Honezero,
\end{align}
and the form norm
\begin{align}
\norm{u}_{q_\pm}^2
:=q_\pm[u]+(\abs b+1)\norm{u}_{L^2}^2
\end{align}
is equivalent to the standard $\Honezero$ norm.

By Kato's representation theorem for closed, semibounded forms
(Theorem~VI.2.23 of \cite{Kato}), applied to the closed form $2q_\pm$,
there exists a unique lower-bounded self-adjoint operator
$H_\pm^\Fried$ characterized by
\begin{align}
2q_\pm[u,v]
&=
\big\langle H_\pm^\Fried u,v\big\rangle_{L^2},
\qquad
u\in\Dom(H_\pm^\Fried),\quad v\in\Honezero.
\end{align}
For $u\in\Dom(H_\pm^\Fried)$,
\begin{align*}
q_\pm[u]
=\frac12\big\langle u,H_\pm^\Fried u\big\rangle_{L^2},
\end{align*}

The operator domain is characterized intrinsically by
\begin{align}
\Dom(H_\pm^\Fried)
=\left\{
u\in H_0^1(0,\infty):
\exists f\in L^2(0,\infty)
\text{ such that }
2q_\pm[u,v]=\langle f,v\rangle_{L^2}
\quad\forall v\in H_0^1(0,\infty)
\right\},
\end{align}
with
\begin{align}
H_\pm^\Fried u:=f.
\end{align}
Equipped with
\begin{align}
\norm{u}_{\mathrm{graph}}^2
:= \norm{u}_{L^2}^2+\norm{H_\pm^\Fried u}_{L^2}^2,
\end{align}
the operator domain is a Hilbert space.

We now compare $H_\pm^\Fried$ with the minimal operator
\begin{align}
S_\pm u
:=
-u''+V_\pm u,
\qquad
u\in C_0^\infty(0,\infty)\subset L^2(0,\infty).
\end{align}
For $u,v\in C_0^\infty(0,\infty)$, integration by parts gives
\begin{align}
q_\pm[u,v]
=\frac12\int_0^\infty
\left(
u'\overline{v'}
+V_\pm u\overline v
\right)\,\dd x
=\frac12\langle S_\pm u,v\rangle_{L^2}.
\end{align}
Thus, if $S_\pm^\Fried$ denotes the Friedrichs extension of the minimal Schr\"odinger operator, then
\begin{align}
H_\pm^\Fried=S_\pm^\Fried.
\end{align}
Thus $H_\pm^\Fried$ is precisely the Friedrichs extension of the
minimal Schr\"odinger operator $-\partial_x^2+V_\pm$; the physical
spatial $T$-energy $q_\pm$ is one-half of its standard quadratic form.

The singular endpoint $x=0$ is in the limit-circle case for $H_\pm^{\mathrm{formal}}$ because the inverse-square cores
\begin{align}
V_\pm(x)
=(\nu_\pm^2-\tfrac14)\,x^{-2}+o(x^{-2}),\qquad \nu_\pm\in(0,1),
\end{align}
place $C_\pm:=\nu_\pm^2-\tfrac14$ strictly inside the one-dimensional Hardy window $(-\tfrac14,\tfrac34)$. In this regime, the general zero-energy solutions of $(-\partial_x^2+V_\pm)u=0$ near the apex behave like
\begin{align}
u(x)\sim A\,x^{\frac12+\nu_\pm}+B\,x^{\frac12-\nu_\pm},\qquad x\downarrow0,
\end{align}
with two linearly independent branches. Weyl's limit-circle/limit-point classification implies that every self-adjoint realization of $H_\pm^{\mathrm{formal}}$ on $(0,\infty)$ corresponds to a choice of a linear boundary condition that eliminates one real degree of freedom in $(A,B)$ at $x=0$.

The Friedrichs extension is characterized among these by the requirement that its form domain be the closure, in the form norm, of compactly supported smooth functions. Since $\Dom(q_\pm)=\Honezero$, every $u\in\Dom(H_\pm^\Fried)$ has vanishing $H^1$-trace at the apex,
\begin{align}
u(0)=0,
\end{align}
and the singular branch $x^{\frac12-\nu_\pm}$ is excluded from the asymptotics. Equivalently, $H_\pm^\Fried$ enforces the regular $x^{\frac12+\nu_\pm}$ behavior as $x\downarrow0$ and discards the more singular branch. The more singular branch would otherwise generate nonzero boundary terms in the integration-by-parts identity for $q_\pm$.

This is the domain-level meaning of the silent apex. The vanishing trace is neither a material wall nor an independently prescribed reflecting law; it is the endpoint behavior encoded by closure of the action-induced finite-energy form. For solutions generated by the associated Friedrichs operator, the Green boundary form and the corresponding $T$-energy flux vanish at $x=0$, leaving no independent endpoint channel through which energy can enter or leave.

Under the locality, causality, symmetry, and conserved-energy axioms formulated by Ishibashi and Wald, admissible scalar-wave dynamics must be generated by some positive self-adjoint realization of the spatial operator~\cite{IshibashiWald2003}. Their framework identifies the admissible operator-theoretic class but does not by itself eliminate every positive extension. In their AdS limit-circle notation, the Friedrichs realization corresponds to the generalized Dirichlet branch $a_\nu=0$~\cite{IshibashiWald2004}. A different positive realization is represented by a different closed form and endpoint domain and therefore contains endpoint structure not generated by the unmodified Einstein--Maxwell $T$-energy. Insisting that the conserved quantity be exactly that energy, with no added endpoint form, singles out $H_\pm^\Fried$ as the unique realization represented by the action-selected form.

The action-based finite-$T$-energy criterion should be distinguished from the curvature-based endpoint analysis used in the earlier instability work~\cite{DottiGleiserPullin2007}. For the algebraically special scalar profile examined in Appendix~\ref{app:DG-instability}, the non-Friedrichs $x^{1/3}$ branch is a formal solution but lies outside $\Dom(q_+)$; the action-selected closure instead retains the $x^{2/3}$ branch. The formal profile remains a solution in the broader endpoint theory, but no corresponding instability mode exists in the action-selected finite-$T$-energy dynamics.

\begin{theorem}\label{thm:friedrichs-silent-apex}
For either sign, let $V_\pm$ denote the corresponding master potential for a fixed multipole $\ell\ge2$ in the superextremal \RN\ background, and let $q_\pm$ be the core form \eqref{eq:qpm-def} on $C_0^\infty(0,\infty)$. Assume that the chosen potential satisfies the core and tail conditions \eqref{eq:V-core} and \eqref{eq:short-range}. Then:
\begin{enumerate}
\item The form $q_\pm$ closes uniquely with
\begin{align*}
\Dom(q_\pm)=H_0^1(0,\infty).
\end{align*}
\item The operator $H_\pm^\Fried$ has the same form domain and self-adjoint apex condition as the Friedrichs extension of the minimal Schr\"odinger operator $S_\pm$.
\item Every finite-$T$-energy solution of the master equation~\eqref{eq:master-wave} satisfies
\begin{align*}
  \phi_{\pm,\ell m}(t,\cdot)\in H_0^1(0,\infty)
\end{align*}
for each $t$, so the $H^1$-trace of $\phi_{\pm,\ell m}$ at $x=0$ vanishes. In particular, the
$T$-energy flux through the optical endpoint $x=0$ is zero and the apex is ``silent'' for the
self-field dynamics.
\end{enumerate}
\end{theorem}

For operator-domain data, the Friedrichs endpoint condition makes the Green boundary form vanish at $x=0$; by density and continuity, the resulting finite-energy evolution has zero apex $T$-energy flux. This is the first, domain-level meaning of silence. The next section determines whether a bound or threshold sector survives, and the radiation analysis then determines whether any nonradiating finite-energy remainder remains.

\section{Doob ground-state representation and absence of bound states on the optical half-line}\label{sec:doob}

To understand the spectral and scattering properties of $H_\pm^\Fried$, we now exploit a Doob (ground-state) transform. The guiding idea is to factor out a distinguished zero-energy profile and reinterpret $H_\pm^\Fried$ as a manifestly positive operator acting on a rescaled unknown.

More concretely, we construct a strictly positive solution $u_{0,\pm}$ of the formal zero-energy equation $H_\pm^{\mathrm{formal}}u=0$ that realizes the Friedrichs behavior at the apex and satisfies the growth estimate~\eqref{eq:u0-growth} at infinity. This ground state induces a factorization
\begin{align*}
  H_\pm^\Fried = A_\pm^{\!*}A_\pm,
\end{align*}
with $A_\pm$ a first-order differential operator adapted to $u_{0,\pm}$. The factorization immediately shows that $H_\pm^\Fried$ is nonnegative and rules out any negative-energy spectrum. 

The remainder of the section constructs the positive profile, extends the resulting square identity from the test core to the closed energy form, and then combines that factorization with the short-range asymptotics to determine the spectrum.

For the remainder of this section, we fix a single superextremal Reissner--Nordstr\"om background, interpreted as the nonlinear self-field of an isolated point charge with ADM mass $M$ and total charge $Q$, in its static Kerr--Schild rest frame. We work throughout in the fixed-background single-source radiative sector defined as follows: we fix $M$ and $Q$, quotient out the $\ell=0,1$ gauge and parameter variations, and exclude static, finite-$T$-energy deformations in the radiative sectors $\ell\ge2$ that would change the higher static multipole structure of the spacetime, such as linearized deformations toward multi-center or strut-supported static configurations. In this fixed-background sector, the genuinely radiative degrees of freedom are the $\ell\ge2$ master fields $\phi_\pm$, and by definition there are no nontrivial static, finite-$T$-energy Einstein--Maxwell perturbations in these sectors.

On the Schr\"odinger side, we work with the formal operators
\begin{align}
  H_\pm^{\mathrm{formal}} = -\partial_x^2 + V_\pm(x),
\label{eq:H-formal}
\end{align}
whose potentials $V_\pm$ have the inverse-square cores of~\eqref{eq:V-core} and the short-range decomposition~\eqref{eq:short-range}. As in Section~\ref{sec:hardy}, we write near the optical apex $x=0$
\begin{equation}
V_\pm(x) = \bigl(\nu_\pm^2-\tfrac14\bigr)x^{-2} + R_\pm(x), \qquad
\nu_\pm\in(0,1),
\label{eq:V-apex}
\end{equation}
where the remainder $R_\pm$ inherits its regularity from the explicit small-$x$ expansions of the optical potentials in Section~\ref{sec:hardy}. Thus, close to the apex the potential is dominated by a universal Hardy-type inverse-square profile, with $R_\pm$ contributing only a softer correction.

The central object in this section will be a strictly positive zero-energy solution $u_{0,\pm}$ of $H_\pm^{\mathrm{formal}}u=0$, which plays the role of a ``ground-state'' profile along the half-line. Once such a solution has been constructed, any mode $u$ can be written as
\begin{align*}
  u(x)=u_{0,\pm}(x)\,v(x),
\end{align*}
so that the Doob transform rewrites the quadratic form $q_\pm$ as one-half the squared $L^2$ norm of the weighted derivative $\partial_xv$ in the metric $u_{0,\pm}^2\,\dd x$. In this representation, the sign of $q_\pm$ is controlled entirely by the positivity and growth properties of $u_{0,\pm}$.

We now record the endpoint and far-field behavior required of this positive zero-energy profile.

Thus, near the apex, the optical potential can be pictured as a pure inverse-square core
\begin{align*}
  (\nu_\pm^2-\tfrac14)x^{-2}
\end{align*}
coming from the Hardy geometry, plus a comparatively soft remainder $R_\pm(x)$ that does
not change the leading power law. The functions $u_{0,\pm}$ we construct below will
look, to first approximation, like the regular power $x^{\frac12+\nu_\pm}$
associated with the core, with $R_\pm$ only adding a small, integrable correction. At infinity, the potential instead has the form
\begin{align*}
V_\pm(x) = \frac{\ell(\ell+1)}{x^{2}} + W_\pm(x), \qquad
W_\pm\in L^1((1,\infty)), \qquad \ell\ge2.
\end{align*}
In the far-field region $x\gg1$ one may picture $V_\pm$ as a centrifugal
barrier $\ell(\ell+1)x^{-2}$ decorated by a short-range tail $W_\pm(x)$ that is integrable
in $x$. For the zero-energy equation $-u''+V_\pm u=0$ the centrifugal part alone has
power-law solutions $x^{\ell+1}$ and $x^{-\ell}$, corresponding to growing and decaying
radial profiles. The tail $W_\pm$ only perturbs these powers by lower-order corrections,
so the eventual ground state $u_{0,\pm}$ can be expected to interpolate between a
Friedrichs-type behavior near the apex and a far-field power-law branch at infinity.
The core forms $q_\pm$ associated with $H_\pm^{\mathrm{formal}}$ were defined
in~\eqref{eq:qpm-def} on $C_0^\infty(0,\infty)$. Their closed extensions, denoted by the same symbols, are densely
defined, closed, semibounded forms on $\Ltwo$ with form domain
$\Dom(q_\pm)=H_0^1(0,\infty)$; see Section~\ref{sec:hardy}.

These near- and far-field pictures set the stage for constructing the ground state $u_{0,\pm}$. The remainder of the section makes this precise: we construct a strictly positive zero-energy solution $u_{0,\pm}$ of $H_\pm^{\mathrm{formal}}u=0$ that realizes the Friedrichs (regular) behavior at the apex, identify the associated Doob (ground-state) representation of $q_\pm$ as a weighted Dirichlet form with weight $u_{0,\pm}^2$, and then combine this with the geometric classification of static radiative perturbations and the short-range tails~\eqref{eq:short-range} to rule out eigenvalues of the Friedrichs realizations $H_\pm^\Fried$.

\subsection{Construction of the ground state near the apex}

The first step is to single out a positive solution of the zero-energy equation that realizes the Friedrichs boundary behavior at the optical apex. This behavior is dictated by the inverse-square core of the potential and is well captured by the corresponding homogeneous model.

Within the local inverse-square construction below,
$\varphi_\pm$ and $\psi_\pm$ denote the two model zero-energy
fundamental solutions. They are not the time-dependent master fields
$\phi_{\pm,\ell m}(t,x)$ of Section~\ref{sec:master-reduction}.

Near $x=0$ we thus consider the homogeneous inverse-square operator
\begin{align*}
H_{\nu_\pm}:=-\partial_x^2+(\nu_\pm^2-\tfrac14)x^{-2},
\end{align*}
whose zero-energy equation admits the standard Frobenius fundamental system
\begin{align*}
\varphi_\pm(x)=x^{\frac12+\nu_\pm},\qquad
\psi_\pm(x)=x^{\frac12-\nu_\pm}.
\end{align*}
Here $\varphi_\pm$ is the ``regular'' branch at the apex and $\psi_\pm$ is the more singular one: both solve the homogeneous inverse-square equation, but
\begin{align*}
  \varphi_\pm(x)\sim x^{\frac12+\nu_\pm},\qquad
  \psi_\pm(x)\sim x^{\frac12-\nu_\pm},\qquad 0<\nu_\pm<1,
\end{align*}
so that $\varphi_\pm$ vanishes more rapidly as $x\downarrow0$. The exponents satisfy $0<\nu_\pm<1$, so both $\varphi_\pm$ and $\psi_\pm$ lie in $L^2(0,1)$ and the endpoint $x=0$ is in the limit-circle regime for the model operator; cf.\ the general one-dimensional discussion in~\cite[Ch.~9]{Teschl} and~\cite[Section~2]{DerezinsRichard}. The corresponding Wronskian
\begin{align*}
W(\varphi_\pm,\psi_\pm)(x)
:=\varphi_\pm(x)\psi_\pm'(x)-\varphi_\pm'(x)\psi_\pm(x),
\end{align*}
is a nonzero constant, independent of $x$. The Friedrichs form domain $\Honezero$ singles out the regular branch $\varphi_\pm$, and we will build $u_{0,\pm}$ as a small perturbation of $\varphi_\pm$ in the presence of the remainder $R_\pm$.

We write the full potential as in~\eqref{eq:V-apex} and consider the zero-energy equation
\begin{align*}
H_\pm^{\mathrm{formal}}u=0\quad\Longleftrightarrow\quad
-u''(x)+\bigl(\nu_\pm^2-\tfrac14\bigr)x^{-2}u(x)+R_\pm(x)u(x)=0.
\end{align*}
Relative to the fundamental system $(\varphi_\pm,\psi_\pm)$, variation of constants gives the
Volterra integral equation for any solution $u$:
\begin{align*}
u(x)=\varphi_\pm(x)
+\int_0^x K_\pm(x,y)\,R_\pm(y)\,u(y)\,\dd y,
\end{align*}
where
\begin{align*}
K_\pm(x,y)
:=\frac{\varphi_\pm(x)\psi_\pm(y)-\psi_\pm(x)\varphi_\pm(y)}{W(\varphi_\pm,\psi_\pm)}.
\end{align*}
This Volterra representation can be viewed as follows: $\varphi_\pm$ is the reference
solution for the pure inverse-square core, and the integral term propagates the effect of
the remainder $R_\pm$ from smaller optical radii $y$ out to $x$. For a fixed $x$, the
kernel $K_\pm(x,y)$ is regular as $y\downarrow0$ thanks to the explicit power-law control
of $\varphi_\pm$ and $\psi_\pm$, so the singularity of $V_\pm$ has been completely isolated
into the homogeneous model $H_{\nu_\pm}$.

Fix $x_0\in(0,1]$. On the triangle $\{(x,y):0<y\le x\le x_0\}$ the kernel $K_\pm$ is
continuous. For $0<y\le x\le 1$ we have
\begin{align*}
\abs{K_\pm(x,y)}
\le C\,\abs{\varphi_\pm(x)\psi_\pm(y)}+C\,\abs{\psi_\pm(x)\varphi_\pm(y)}
\le C\,x^{\frac12+\nu_\pm}y^{\frac12-\nu_\pm}
\le C\,x,
\end{align*}
since $0<\nu_\pm<1$ and $y\le x$. Thus, a direct estimate using the explicit expressions for
$\varphi_\pm$ and $\psi_\pm$ shows that
\begin{align*}
\abs{K_\pm(x,y)}\le Cx\qquad(0<y\le x\le x_0)
\end{align*}
for some $C>0$ independent of $x$ and $y$.

By the explicit small-$x$ expansion of $V_\pm$ from~\eqref{eq:V-core}
(and the fact that $V_\pm-(\nu_\pm^2-\tfrac14)x^{-2}$ extends continuously to $x=0$),
the remainder $R_\pm$ extends continuously to $[0,x_0]$ and is thus bounded and
integrable on $(0,x_0]$:
\begin{align*}
R_\pm\in L^1((0,x_0]),\qquad
\|R_\pm\|_{L^1(0,x_0]}\le C_{x_0}<\infty.
\end{align*}
Moreover, on the triangle $\{(x,y):0<y\le x\le x_0\}$ the kernel $K_\pm$
is continuous, and the estimate above shows that
\begin{align*}
\sup_{0<x\le x_0}\int_0^x \abs{K_\pm(x,y)}\,\abs{R_\pm(y)}\,\dd y
\le Cx_0\|R_\pm\|_{L^1(0,x_0)}\longrightarrow0
\quad\text{as }x_0\downarrow0.
\end{align*}
Thus, for $x_0$ sufficiently small the Volterra operator
\begin{align*}
(\mathcal T_\pm u)(x):=\int_0^x K_\pm(x,y)\,R_\pm(y)\,u(y)\,\dd y
\end{align*}
is a strict contraction on $C([0,x_0])$. The Picard iteration
\begin{align*}
u^{(0)}(x):=\varphi_\pm(x),\qquad
u^{(n+1)}(x):=\varphi_\pm(x)+(\mathcal T_\pm u^{(n)})(x)
\end{align*}
converges uniformly on $[0,x_0]$ to a unique continuous function $u_{0,\pm}$, which is
$C^2$ and solves $H_\pm^{\mathrm{formal}}u=0$ on $(0,x_0]$; see, for example,
\cite[Section~9.1]{Teschl} for Volterra equations with weakly singular kernels.

To obtain positivity near the apex, it is convenient to work in the weighted
supremum norm
\begin{align*}
\|u\|_{*,x_0}:=\sup_{0<x\le x_0}\left|\frac{u(x)}{\varphi_\pm(x)}\right|.
\end{align*}
Since $\varphi_\pm(x)=x^{\frac12+\nu_\pm}>0$ on $(0,x_0]$, this norm is equivalent to
the usual supremum norm on $[0,x_0]$.
The weight $\varphi_\pm$ simply factors out the expected leading behavior near the apex,
so that controlling $u/\varphi_\pm$ in the $\|\cdot\|_{*,x_0}$ norm is the same as
controlling $u$ itself, but with the regular branch of the inverse-square model already
built in. In particular, a fixed point of the Volterra map in this norm corresponds to a
solution $u$ that ``sticks close'' to $\varphi_\pm$ as $x\downarrow0$.

The Volterra operator $\mathcal T_\pm$ satisfies
\begin{align*}
\| \mathcal T_\pm u\|_{*,x_0}
\le \varepsilon\,\|u\|_{*,x_0}
\end{align*}
for some $\varepsilon\in(0,1)$ and $x_0>0$ small enough, by the estimate derived above. Thus, $\mathcal T_\pm$ is a strict contraction on the Banach space
$(C([0,x_0]),\|\cdot\|_{*,x_0})$.
One can picture the Picard iterates
\begin{align*}
  u^{(0)}=\varphi_\pm,\quad
  u^{(1)}=\varphi_\pm+\mathcal T_\pm\varphi_\pm,\quad
  u^{(2)}=\varphi_\pm+\mathcal T_\pm\varphi_\pm+\mathcal T_\pm^2\varphi_\pm,\ \dots
\end{align*}
as successively adding the effect of $R_\pm$ along the optical half-line, with each step
remaining uniformly close to $\varphi_\pm$ near $x=0$. The contraction property ensures that the successive corrections converge to a unique fixed profile $u_{0,\pm}$ that inherits the regular apex behavior of $\varphi_\pm$. This fixed profile satisfies
\begin{align*}
u_{0,\pm}=\varphi_\pm+\mathcal T_\pm u_{0,\pm}.
\end{align*}
Taking the weighted norm and using the contraction property yields
\begin{align*}
\left\|\frac{u_{0,\pm}}{\varphi_\pm}-1\right\|_{*,x_0}
=\left\|\frac{\mathcal T_\pm u_{0,\pm}}{\varphi_\pm}\right\|_{*,x_0}
\le\varepsilon\left\|\frac{u_{0,\pm}}{\varphi_\pm}\right\|_{*,x_0}.
\end{align*}
If we denote $M:=\|u_{0,\pm}/\varphi_\pm\|_{*,x_0}$, this implies
$(1-\varepsilon)M\le1$ and hence $M\le(1-\varepsilon)^{-1}$. In particular,
\begin{align*}
\left\|\frac{u_{0,\pm}}{\varphi_\pm}-1\right\|_{*,x_0}
\le M-1\le \frac{\varepsilon}{1-\varepsilon}.
\end{align*}
Choosing $x_0>0$ so small that $\varepsilon\le\tfrac12$, we obtain the uniform comparison
\begin{align*}
\frac12\le \frac{u_{0,\pm}(x)}{\varphi_\pm(x)}\le 2
\quad\text{for all }x\in(0,x_0],
\end{align*}
and in particular
\begin{equation}
u_{0,\pm}(x)>0 \quad\text{for all }x\in(0,x_0].
\label{eq:u0-positive-near0}
\end{equation}
Moreover, the Picard iteration shows that $u_{0,\pm}\to\varphi_\pm$ in
$\|\cdot\|_{*,x_0}$ as $x\downarrow0$, which yields the asymptotics
\begin{align*}
\frac{u_{0,\pm}(x)}{\varphi_\pm(x)}\to1,\qquad x\downarrow0,
\end{align*}
that is,
\begin{equation}
u_{0,\pm}(x)=x^{\frac12+\nu_\pm}\bigl(1+o(1)\bigr),
\qquad x\downarrow0,
\label{eq:u0-friedrichs}
\end{equation}
exactly the Friedrichs (regular) branch singled out by the form domain
$\Honezero$; compare the classification of Friedrichs boundary conditions for
inverse-square cores in~\cite[Section~2]{DerezinsRichard} and the discussion for
one-dimensional Schr\"odinger operators in~\cite[Ch.~9]{Teschl}.

\subsection{
\texorpdfstring
{Global continuation and large-$x$ behavior}
{Global continuation and large-x behavior}
}

We now continue $u_{0,\pm}$ to a global solution on $(0,\infty)$. On any finite interval
$[x_0,R]$ with $0<x_0<R<\infty$, the potential $V_\pm$ is continuous, so the initial value
problem for the second-order equation $H_\pm^{\mathrm{formal}}u=0$ with Cauchy data
\begin{align*}
u(x_0)=u_{0,\pm}(x_0)>0,\qquad u'(x_0)=u'_{0,\pm}(x_0)
\end{align*}
has a unique $C^2$ solution on $[x_0,R]$. Standard ODE continuation shows that this
solution extends uniquely to all $x>0$, and by uniqueness it coincides with the local
solution constructed above on $(0,x_0]$. We thus obtain a $C^2$ solution $u_{0,\pm}\in C^2((0,\infty))$ of $H_\pm^{\mathrm{formal}}u_{0,\pm}=0$ on $(0,\infty)$ with the apex asymptotics~\eqref{eq:u0-friedrichs}. In other words, $u_{0,\pm}$ is the unique zero-energy solution that matches the Friedrichs branch $x^{\frac12+\nu_\pm}$ at the optical apex; we show below that it remains strictly positive along the whole half-line.

The positivity~\eqref{eq:u0-positive-near0} on $(0,x_0]$ extends to all $x>0$.
Indeed, suppose for contradiction that there exists $x_\ast>0$ with
$u_{0,\pm}(x_\ast)=0$. Let
\begin{align*}
  x_\ast := \inf\{ x>0 : u_{0,\pm}(x)=0\},
\end{align*}
so $u_{0,\pm}(x)>0$ for $0<x<x_\ast$ by continuity, and $u_{0,\pm}(x_\ast)=0$.
Thus, $u_{0,\pm}$ attains its (nonnegative) minimum $0$ at the interior point
$x_\ast\in(0,\infty)$ of the interval $[0,x_\ast]$, and $u_{0,\pm}\ge0$ on $[0,x_\ast]$.

Since $u_{0,\pm}$ is $C^2$, elementary calculus implies
\begin{align*}
  u_{0,\pm}'(x_\ast)=0, \qquad u_{0,\pm}''(x_\ast)\ge0.
\end{align*}
Evaluating the differential equation
\begin{align*}
  -u_{0,\pm}''(x)+V_\pm(x)u_{0,\pm}(x)=0
\end{align*}
at $x=x_\ast$ and using $u_{0,\pm}(x_\ast)=0$ gives
\begin{align*}
  -u_{0,\pm}''(x_\ast)=0,
\end{align*}
so $u_{0,\pm}''(x_\ast)=0$ as well. Hence
\begin{align*}
  u_{0,\pm}(x_\ast)=u_{0,\pm}'(x_\ast)=0.
\end{align*}
By uniqueness of the Cauchy problem for the equation
$-u''+V_\pm u=0$ at $x=x_\ast$, this forces $u_{0,\pm}\equiv0$ on
$(0,\infty)$, contradicting~\eqref{eq:u0-positive-near0}. Thus
\begin{align*}
  u_{0,\pm}(x)>0\quad\text{for all }x>0.
\end{align*}

The profile $u_{0,\pm}$ interpolates between the two asymptotic regimes: near the apex it behaves like $x^{\frac12+\nu_\pm}$, while for large $x$ it gradually bends into its far-field power-law behavior as $x\to\infty$.

For large $x$, the decomposition~\eqref{eq:short-range} shows that the full potential $V_\pm$
is a short-range perturbation of the inverse-square barrier $\ell(\ell+1)x^{-2}$ in the
sense that $V_\pm-\ell(\ell+1)x^{-2}=W_\pm\in L^1((1,\infty))$. Comparing the
zero-energy equation with the homogeneous reference equation
\begin{align*}
-u''(x)+\frac{\ell(\ell+1)}{x^{2}}u(x)=0,
\end{align*}
whose Euler-type fundamental system is
\begin{align*}
\varphi_1(x)=x^{\ell+1},\qquad
\varphi_2(x)=x^{-\ell},
\end{align*}
and using standard variation-of-constants constructions for one-dimensional Schr\"odinger
operators with integrable tails (see, for instance, the treatment of short-range
potentials in~\cite[Sections~9.7, 11.2]{Teschl} and~\cite[Sections~4.1--4.4]{DerezinsRichard}),
one finds that every zero-energy solution of $H_\pm^{\mathrm{formal}}u=0$ admits an
asymptotic expansion of the form
\begin{equation}
u(x)=c_1\,x^{\ell+1}\bigl(1+o(1)\bigr)+c_2\,x^{-\ell}\bigl(1+o(1)\bigr),
\qquad x\to\infty,
\label{eq:u-general-asymp}
\end{equation}
for some constants $c_1,c_2\in\mathbb{C}$ depending on the solution. In particular,
$u_{0,\pm}$ has the form~\eqref{eq:u-general-asymp}. We will only need the rough growth
bound
\begin{equation}
\abs{u_{0,\pm}(x)}\le C(1+x)^{\ell+1}, \qquad x\ge1,
\label{eq:u0-growth}
\end{equation}
for some constant $C>0$, which follows immediately from~\eqref{eq:u-general-asymp}.

For the positive solution $u_{0,\pm}$ constructed above, we now show that the coefficient $c_1$ in~\eqref{eq:u-general-asymp} is \emph{nonzero} and that $u_{0,\pm}$ lies outside $L^2(0,\infty)$. Recall that for a general zero-energy solution we have
\begin{align*}
  u(x)
  = c_1 x^{\ell+1}\bigl(1+o(1)\bigr)
  + c_2 x^{-\ell}\bigl(1+o(1)\bigr)
  \quad\text{as }x\to\infty,
\end{align*}
with constants $c_1,c_2\in\mathbb{R}$. Since $\ell\ge2$, the decaying branch $x^{-\ell}$ belongs to $L^2((1,\infty))$ whereas the growing branch $x^{\ell+1}$ does not, so
\begin{align*}
  u\in L^2(1,\infty)
  \quad\Longleftrightarrow\quad
  c_1=0.
\end{align*}

Suppose, for contradiction, that $c_1=0$ for $u_{0,\pm}$. Then the large-$x$ expansion~\eqref{eq:u-general-asymp} yields
\begin{align*}
  u_{0,\pm}(x)
  = c_2 x^{-\ell}\bigl(1+o(1)\bigr)
  \quad\text{as }x\to\infty,
\end{align*}
and in particular $u_{0,\pm}\in H_0^1(0,\infty)$ and has finite $T$-energy. Using the master equation
\begin{align*}
  \partial_t^2\phi_\pm - \partial_x^2\phi_\pm + V_\pm(x)\phi_\pm = 0,
\end{align*}
and the Hardy control developed in Section~\ref{sec:hardy}, a standard elliptic-regularity argument shows that a suitable multiple of $u_{0,\pm}$ would then define a nontrivial static, finite-$T$-energy Einstein--Maxwell perturbation in the radiative sectors $\ell\ge2$.

In the fixed-background single-source sector defined above, we explicitly exclude such perturbations: once the background ADM mass $M$, total charge $Q$, and low-multipole gauge freedoms have been fixed, there are no nontrivial static, finite-$T$-energy Einstein--Maxwell modes in the radiative sectors. This contradiction shows that $c_1$ cannot vanish, so
\begin{align*}
  c_1\neq 0
  \quad\text{for }u_{0,\pm}.
\end{align*}

Consequently, $u_{0,\pm}$ realizes the \emph{growing} power-law branch at infinity. Combining~\eqref{eq:u-general-asymp} with the a priori bound~\eqref{eq:u0-growth}, we may summarize its behavior as
\begin{align*}
  u_{0,\pm}(x) \sim x^{\frac12+\nu_\pm}
  \quad\text{as }x\downarrow0,
  \qquad
  \abs{u_{0,\pm}(x)} \le C(1+x)^{\ell+1}
  \quad\text{for }x\ge1,
\end{align*}
for some constant $C>0$, and in particular $u_{0,\pm}\notin L^2(0,\infty)$.

The function $u_{0,\pm}$ is a single positive ground-state profile. Near the apex it behaves like $x^{\frac12+\nu_\pm}$, while for large $x$ it gradually bends into the far-field power-law regime set by the centrifugal tail of $V_\pm$. The fact that $u_{0,\pm}$ never changes sign, together with the growth control in~\eqref{eq:u0-growth}, is what makes it a suitable positive ground state for the Doob transform. From this point of view, the hypergeometric structure of the explicit mode solutions is not accidental. It is the same special-function pattern that underlies the solvable half-line models of Inoue and Richard~\cite{InouRichard}, where the resolvent, spectral density, and scattering data of $D_{\mu,\nu}$ can all be written in terms of the Gauss hypergeometric function and its asymptotics.

\subsection{
\texorpdfstring
{Doob transform and positivity of the $T$-energy forms}
{Doob transform and positivity of the T-energy forms}
}

We now use $u_{0,\pm}$ as a positive ground state to rewrite the
$T$-energy in a form where positivity is completely transparent. The basic
idea is to factor out the zero-energy profile and regard what remains
as a fluctuation around this background. Accordingly, for
$u\in C_0^\infty(0,\infty)$ we introduce the renormalized field
\begin{align*}
v:=\frac{u}{u_{0,\pm}},
\end{align*}
which measures $u$ relative to the zero-energy profile, so that $u=u_{0,\pm}v$ and
\begin{align*}
u'(x)=u_{0,\pm}'(x)v(x)+u_{0,\pm}(x)v'(x).
\end{align*}
Here $v$ measures deviations from the distinguished profile $u_{0,\pm}$, while
$u_{0,\pm}$ itself encodes the static zero-energy geometry of the mode.

A short calculation rewrites the derivative and potential
parts of $q_\pm[u]$ in terms of $v$ and $u_{0,\pm}$. Using $u=u_{0,\pm}v$ and
$u'=u_{0,\pm}'v+u_{0,\pm}v'$, we obtain
\begin{align*}
\abs{u'(x)}^2
&=\abs{u_{0,\pm}'(x)}^2\abs{v(x)}^2
+2u_{0,\pm}(x)u_{0,\pm}'(x)\Re\bigl(\overline{v(x)}\,v'(x)\bigr)
+u_{0,\pm}(x)^2\abs{v'(x)}^2,\\
V_\pm(x)\abs{u(x)}^2
&=V_\pm(x)u_{0,\pm}(x)^2\abs{v(x)}^2
=u_{0,\pm}(x)u_{0,\pm}''(x)\,\abs{v(x)}^2,
\end{align*}
where in the last step we used the zero-energy relation
$-u_{0,\pm}''+V_\pm u_{0,\pm}=0$. In parallel, the mixed term can be written
as a total derivative:
\begin{align*}
\frac{\dd}{\dd x}\Big(u_{0,\pm}u_{0,\pm}'\abs{v}^2\Big)
=\bigl(u_{0,\pm}'^2+u_{0,\pm}u_{0,\pm}''\bigr)\abs{v}^2
+2u_{0,\pm}u_{0,\pm}'\Re\bigl(\overline{v}v'\bigr).
\end{align*}
Combining these expressions and eliminating the cross terms yields the
pointwise identity
\begin{align}
\abs{u'(x)}^2+V_\pm(x)\abs{u(x)}^2
&=u_{0,\pm}(x)^2\abs{v'(x)}^2
+\frac{\dd}{\dd x}\Big(u_{0,\pm}(x)u_{0,\pm}'(x)\,\abs{v(x)}^2\Big),
\label{eq:doob-pointwise-final}
\end{align}
valid for all $x>0$ and all $u\in C_0^\infty(0,\infty)$. Thus, up to a total
derivative, twice the spatial energy density is the square of the derivative of the
renormalized field $v$ weighted by $u_{0,\pm}^2$.

When \eqref{eq:doob-pointwise-final} is integrated over the optical
half-line, the total derivative contributes only boundary terms. For compactly
supported test fields these boundary terms vanish, and for finite-$T$-energy
fields they are controlled by the Friedrichs behavior at the apex and the
decay at infinity. In particular, the bulk contribution to the spatial
$T$-energy is given by the manifestly nonnegative quantity
$\tfrac12u_{0,\pm}^2\abs{v'}^2$, so the Doob transform has converted the original
Hardy-type structure into a perfect-square representation adapted to the
ground state $u_{0,\pm}$.

Integrating~\eqref{eq:doob-pointwise-final} over $(0,\infty)$ and using that
$u\in C_0^\infty(0,\infty)$ has compact support, say
$\supp u\subset[a,b]\Subset(0,\infty)$, we obtain
\begin{align*}
\int_0^\infty\bigl(\abs{u'}^2+V_\pm\abs{u}^2\bigr)\,\dd x
&=\int_0^\infty u_{0,\pm}(x)^2\abs{v'(x)}^2\,\dd x
+\Big[u_{0,\pm}(x)u_{0,\pm}'(x)\,\abs{v(x)}^2\Big]_{x=0}^{x=\infty}.
\end{align*}
Since $u$ vanishes outside $[a,b]$, we have $v=u/u_{0,\pm}\equiv0$ for
$x\notin[a,b]$, and the boundary term vanishes. Thus
\begin{align*}
\int_0^\infty\bigl(\abs{u'}^2+V_\pm\abs{u}^2\bigr)\,\dd x
=\int_0^\infty u_{0,\pm}(x)^2\abs{v'(x)}^2\,\dd x.
\end{align*}

Recalling the definition of the core form,
\begin{align*}
q_\pm[u]
=\frac12\int_0^\infty\bigl(\abs{u'(x)}^2+V_\pm(x)\abs{u(x)}^2\bigr)\,\dd x,
\qquad u\in C_0^\infty(0,\infty),
\end{align*}
we arrive at the Doob identity
\begin{align}
q_\pm[u]
=\frac12\int_0^\infty u_{0,\pm}(x)^2\bigl|(u/u_{0,\pm})'(x)\bigr|^2\,\dd x\ \ge\ 0,
\qquad u\in C_0^\infty(0,\infty).
\label{eq:doob-identity}
\end{align}

The next subsection extends this core identity to the closed form and identifies the corresponding first-order factorization.

\subsection{
\texorpdfstring
{Extension to $\Honezero$ and factorization}
{Extension to H01 and factorization}
}

By Section~\ref{sec:hardy}, each core form $q_\pm$ is closable, and its closed extension, denoted by the same symbol, is a densely defined, closed, semibounded quadratic form on $\Ltwo$ with form domain $\Dom(q_\pm)=H_0^1(0,\infty)$. The right-hand side of~\eqref{eq:doob-identity} defines another nonnegative quadratic form on $C_0^\infty(0,\infty)$, obtained by measuring the weighted derivative $(u/u_{0,\pm})'$ in the metric $u_{0,\pm}^2\,\dd x$. To compare these two forms it suffices to estimate the Doob form in terms of the original energy form, using the growth properties of $u_{0,\pm}$ and the regularity of $V_\pm$.

Using the growth estimate~\eqref{eq:u0-growth} and $u_{0,\pm}\in C^2((0,\infty))$, there exists a constant
$C>0$ such that for all $u\in C_0^\infty(0,\infty)$,
\begin{align*}
\int_0^\infty u_{0,\pm}(x)^2\bigl|(u/u_{0,\pm})'(x)\bigr|^2\,\dd x
\le C\int_0^\infty\bigl(\abs{u'(x)}^2+(1+\abs{V_\pm(x)})\abs{u(x)}^2\bigr)\,\dd x.
\end{align*}
This estimate shows that the Doob form is dominated by $q_\pm$ with arbitrarily small relative bound: in the terminology of Kato, it is $q_\pm$-bounded with relative bound zero. As a consequence, the Doob form is closable on $C_0^\infty(0,\infty)$, and its closure must agree with $q_\pm$ by uniqueness of the closed extension. Thus, \eqref{eq:doob-identity} extends by density to all $u\in H_0^1(0,\infty)$:
\begin{align}
q_\pm[u]
=\frac12\int_0^\infty u_{0,\pm}(x)^2\bigl|(u/u_{0,\pm})'(x)\bigr|^2\,\dd x,
\qquad u\in H_0^1(0,\infty),
\label{eq:doob-identity-H1}
\end{align}
and in particular
\begin{align*}
q_\pm[u]\ge0 \quad\text{for all }u\in \Honezero.
\end{align*}

Once that extension is established, the Doob transform gives the following transparent self-field representation. The strictly positive zero-energy solution $u_{0,\pm}$ satisfies the Friedrichs behavior at $x=0$ and grows as $x^{\ell+1}$ at infinity, and the spatial $T$-energy becomes
\begin{align*}
  q_\pm[u]
  = \frac12\int_0^\infty u_{0,\pm}(x)^2\,
  \Bigl|\Bigl(\frac{u}{u_{0,\pm}}\Bigr)'(x)\Bigr|^2\,\dd x,
  \qquad u\in\Honezero.
\end{align*}
It follows that $q_\pm[u]\ge0$ and that $q_\pm[u]=0$ if and only if $u/u_{0,\pm}$ is constant, so
that the only zero-energy solution compatible with the form domain is the (non-square-integrable)
ground state $u_{0,\pm}$ itself.
Here and for the remainder of the Doob-factorization discussion,
$S_\pm(x)$ denotes the scalar logarithmic derivative of the positive
zero-energy profile. It should not be confused with the minimal
Schr\"odinger operator $S_\pm$ introduced in
Section~\ref{sec:friedrichs}; the distinction is clear from whether
$S_\pm$ is evaluated at $x$ or applied to a function.

It is convenient to express~\eqref{eq:doob-identity-H1} as a factorization of the
Friedrichs realization. Define the first-order differential expression
\begin{align*}
S_\pm(x):=\frac{u_{0,\pm}'(x)}{u_{0,\pm}(x)},
\qquad x>0,
\end{align*}
and consider the operator
\begin{align*}
A_\pm:=\partial_x - S_\pm
\end{align*}
initially on $C_0^\infty(0,\infty)$. For $u\in C_0^\infty(0,\infty)$ we have
\begin{align*}
\frac{u}{u_{0,\pm}}=:v,
\qquad
A_\pm u=u'(x)-S_\pm(x)u(x)
=u_{0,\pm}(x)v'(x)
=u_{0,\pm}(x)\left(\frac{u}{u_{0,\pm}}\right)'(x),
\end{align*}
so that
\begin{align*}
\norm{A_\pm u}_{L^2}^2
=\int_0^\infty u_{0,\pm}(x)^2
\left|\left(\frac{u}{u_{0,\pm}}\right)'(x)\right|^2\,\dd x
=2q_\pm[u].
\end{align*}
From this point of view, $A_\pm$ is the first-order ``covariant derivative'' in the
logarithmic gauge determined by $u_{0,\pm}$:
\begin{align*}
  A_\pm = \partial_x - (\partial_x\log u_{0,\pm}).
\end{align*}
The Doob representation is then the statement that the spatial Einstein--Maxwell $T$-energy of a
mode is one-half the squared $L^2$ norm of this covariant derivative, and the factorization
$H_\pm^\Fried=A_\pm^{\!*}A_\pm$ identifies the Friedrichs Hamiltonian as the positive
operator generated by that energy. This is the final step in the same logic: geometry fixes the energy, the closed energy form fixes $H_\pm^\Fried$, and the Doob representation exposes that already determined Hamiltonian as the positive square of the first-order covariant derivative selected by the distinguished zero-energy profile.
Since $2q_\pm$ is the closed form associated with the Friedrichs extension $H_\pm^\Fried$,
constructed in Section~\ref{sec:friedrichs}, the representation theorem for closed forms
(see~\cite[Section~VI.1]{Kato}) yields
\begin{align*}
H_\pm^\Fried = A_\pm^{\!*}A_\pm
\quad\text{on }\Dom(H_\pm^\Fried)\subset \Honezero,
\end{align*}
and
\begin{equation}
q_\pm[u]
=\frac12\norm{A_\pm u}_{L^2}^2
=\frac12\big\langle u,\,H_\pm^\Fried u\big\rangle_{L^2},
\qquad u\in\Dom(H_\pm^\Fried).
\label{eq:q-vs-H}
\end{equation}
In particular, the Friedrichs realizations satisfy
\begin{equation}
H_\pm^\Fried\ge0
\quad\text{as self-adjoint operators on }L^2(0,\infty),
\label{eq:H-nonnegative}
\end{equation}
with quadratic forms $q_\pm$ given by the Doob identity~\eqref{eq:doob-identity-H1}.

\subsection{Absence of eigenvalues and static Einstein--Maxwell perturbations}

We now show that the Friedrichs realizations $H_\pm^\Fried$ have no eigenvalues at any energy and explain how the zero-energy case encodes static finite-$T$-energy Einstein--Maxwell perturbations on the fixed superextremal \RN\ background. On the Schr\"odinger side, an eigenvalue of $H_\pm^\Fried$ would correspond to a bound state of the master fields: an $L^2$ solution of the stationary equation
\begin{align*}
(-\partial_x^2+V_\pm)u=\lambda u
\end{align*}
that remains trapped along the optical half-line. The Doob representation allows us to rule out such trapped modes directly at the level of the quadratic form. For the rest of this subsection we fix one of the parities and suppress the $\pm$ index from the notation. A zero-energy eigenvalue would correspond to an $L^2$ solution $u$ of
\begin{align*}
(-\partial_x^2+V_\pm)u=0
\end{align*}
satisfying the Friedrichs boundary behavior at the apex and square-integrability at infinity. Such a solution, if it existed, would correspond to a static finite-$T$-energy Einstein--Maxwell perturbation in the $\ell\ge2$ sector. The Doob representation shows that no such nontrivial solution exists.

Let $u\in\Dom(H^\Fried)\subset H_0^1(0,\infty)$ satisfy
\begin{align*}
H^\Fried u = 0.
\end{align*}
Then, by~\eqref{eq:q-vs-H},
\begin{align*}
q[u]=\frac12\big\langle u,\,H^\Fried u\big\rangle_{L^2}=0.
\end{align*}
Applying the Doob identity~\eqref{eq:doob-identity-H1} we obtain
\begin{align*}
0=q[u]
=\frac12\int_0^\infty u_{0,\pm}(x)^2
\left|
\left(
\frac{u}{u_{0,\pm}}
\right)'(x)
\right|^2\,\dd x.
\end{align*}
Because the integrand is pointwise nonnegative and the weight $u_{0,\pm}^2$ is strictly positive
on $(0,\infty)$, the integral can vanish only if $(u/u_{0,\pm})'$ is zero almost everywhere. In other words, a zero-energy mode in the
closed finite-$T$-energy space must be perfectly aligned with the ground-state profile
$u_{0,\pm}$ and cannot oscillate or localize in any nontrivial way along the optical
coordinate. By continuity, we conclude
\begin{align*}
u(x)=c\,u_{0,\pm}(x),\qquad x>0,
\end{align*}
for some $c\in\mathbb{C}$. Since $c_1\neq0$ in the large-$x$ expansion~\eqref{eq:u-general-asymp}, we have
$u_{0,\pm}\notin L^2(1,\infty)$ for every $\ell\ge2$.
Thus, $c\,u_{0,\pm}\in L^2(0,\infty)$ only when $c=0$. Hence $u\equiv0$, and
\begin{equation}
\Ker H^\Fried = \{0\}.
\label{eq:H-kernel-trivial}
\end{equation}
Equivalently, zero is not an eigenvalue of $H^\Fried$ in the action-selected finite-$T$-energy space. In
terms of the original Einstein--Maxwell perturbations, this means that there are no
nontrivial static, finite-$T$-energy perturbations in the radiative sectors $\ell\ge2$,
modulo the usual low-multipole gauge freedoms; see the geometric classification in
Section~\ref{sec:master-reduction} and the Kodama--Ishibashi theory~\cite{IshibashiKodama2011PTPS}.

We next exclude eigenvalues at strictly positive energies.
In physical terms this amounts to showing that there are no finite-$T$-energy modes that
oscillate indefinitely in $x$ while remaining spatially localized: every solution with
$\lambda>0$ must either fail to lie in $L^2$ or fail to satisfy the Friedrichs boundary
condition at the apex.
Let $\lambda>0$ and suppose
that $u\in\Dom(H^\Fried)\subset H_0^1(0,\infty)$ satisfies
\begin{equation}
H^\Fried u = \lambda u.
\label{eq:eigen-eq-positive}
\end{equation}
Then $u$ is a nontrivial solution of the differential equation
\begin{align*}
-u''(x)+V(x)u(x)=\lambda u(x)
\end{align*}
on $(0,\infty)$ with $u\in L^2(0,\infty)$ and $u(0)=0$ in the trace sense.

By~\eqref{eq:short-range} and the decay of $x^{-2}$ at infinity, we have
\begin{align*}
V\in L^1((1,\infty)),
\end{align*}
so the full potential is short-range on the half-line. Classical one-dimensional
scattering theory for such potentials (see, for example,
\cite[Sections~9.6--9.7]{Teschl}) implies that for each fixed $\lambda>0$ there exist
Jost solutions $f_\pm(\cdot,\lambda)$ of the eigenvalue equation with asymptotics
\begin{equation}
f_\pm(x,\lambda)
=\exp\big(\pm i\sqrt{\lambda}\,x\big)\bigl(1+o(1)\bigr),
\qquad x\to\infty,
\label{eq:jost-asymptotics}
\end{equation}
and that every (complex) solution of~\eqref{eq:eigen-eq-positive} is a linear combination
of $f_+(\cdot,\lambda)$ and $f_-(\cdot,\lambda)$.

The asymptotics~\eqref{eq:jost-asymptotics} imply that every solution $u$ of
\eqref{eq:eigen-eq-positive} satisfies
\begin{align*}
u(x)=a\,\mathrm{e}^{i\sqrt{\lambda}\,x}\bigl(1+o(1)\bigr)
+b\,\mathrm{e}^{-i\sqrt{\lambda}\,x}\bigl(1+o(1)\bigr),\qquad x\to\infty,
\end{align*}
for some constants $a,b\in\mathbb{C}$ not both zero. In particular,
\begin{align*}
\lim_{R\to\infty}\frac{1}{R}\int_1^R \abs{u(x)}^2\,\dd x
= |a|^2+|b|^2>0,
\end{align*}
so there exist constants $c>0$ and $R_0>1$ such that
\begin{align*}
\int_1^R \abs{u(x)}^2\,\dd x \ge c(R-1)\quad\text{for all }R\ge R_0.
\end{align*}
Hence
\begin{align*}
\int_1^\infty\abs{u(x)}^2\,\dd x=+\infty,
\end{align*}
and no nontrivial solution of~\eqref{eq:eigen-eq-positive} can belong to
$L^2(1,\infty)$. In particular, there is no $u\in L^2(0,\infty)$ solving
\eqref{eq:eigen-eq-positive}, so
\begin{equation}
\lambda>0,\quad u\in\Dom(H^\Fried),\quad H^\Fried u=\lambda u
\quad\Longrightarrow\quad u=0,
\label{eq:no-positive-eig}
\end{equation}
i.e., $H^\Fried$ has no positive eigenvalues.

Combining~\eqref{eq:H-nonnegative}, \eqref{eq:H-kernel-trivial}, and
\eqref{eq:no-positive-eig}, we conclude that $H^\Fried$ has no eigenvalues at any
energy: the point spectrum is empty. On the other hand, the difference
$V-V_{\nu}$ is relatively compact with respect to the homogeneous inverse-square
model $H_\nu$ from~\eqref{eq:V-apex}, whose spectrum is known to be
$\sigma(H_\nu)=[0,\infty)$; see, for example,~\cite{DerezinsRichard,Teschl}.
By Weyl's theorem on the invariance of the essential spectrum under relatively
compact perturbations, we obtain
\begin{align*}
\sigma_{\mathrm{ess}}(H^\Fried)=[0,\infty).
\end{align*}
Finally, the short-range character of the tails~\eqref{eq:short-range}, that is,
$V-\ell(\ell+1)x^{-2}\in L^1((1,\infty))$, together with the absence of
eigenvalues, implies that the spectral measure of $H^\Fried$ on $(0,\infty)$ is
purely absolutely continuous; this is a standard consequence of one-dimensional
scattering theory for short-range potentials (see, e.g.,
\cite[Sections~9.6--9.7]{Teschl} and~\cite[Sections~4.1--4.4]{DerezinsRichard}).
Consequently,
\begin{align*}
\sigma(H^\Fried)=\sigma_{\mathrm{ac}}(H^\Fried)=[0,\infty).
\end{align*}
The Friedrichs Hamiltonian therefore has no bound spectrum at all: every energy $\lambda>0$
belongs to the continuous scattering spectrum, and the only ``threshold state'' at
$\lambda=0$ is the non-normalizable ground-state profile $u_{0,\pm}$ used in the Doob
transform. In the language of the original Einstein--Maxwell fields, this means that
radiative perturbations in the $\ell\ge2$ sectors can only propagate and scatter along
the optical half-line; there are no normalizable stationary or time-periodic modes
trapped by the naked singularity.

\begin{theorem}\label{thm:doob-spectrum}
In the fixed-background single-source radiative sector defined above, let $H_\pm^\Fried$ be the Friedrichs realizations constructed in Theorem~\ref{thm:friedrichs-silent-apex}. Then:

\medskip

\noindent (i) There exist strictly positive functions
$u_{0,\pm}\in C^2((0,\infty))$ solving
\begin{align*}
H_\pm^{\mathrm{formal}}u_{0,\pm}=0
\quad\text{on }(0,\infty),
\end{align*}
with asymptotics
\begin{align*}
u_{0,\pm}(x)=x^{\frac12+\nu_\pm}\bigl(1+o(1)\bigr)\quad\text{as }x\downarrow0,
\qquad
\abs{u_{0,\pm}(x)}\le C(1+x)^{\ell+1}\quad\text{for }x\ge1,
\end{align*}
for some constant $C>0$.

\medskip

\noindent (ii) For every $u\in H_0^1(0,\infty)$ the Doob identity
\begin{align*}
q_\pm[u]
=\frac12\int_0^\infty u_{0,\pm}(x)^2\bigl|(u/u_{0,\pm})'(x)\bigr|^2\,\dd x
\end{align*}
holds, and in particular $q_\pm[u]\ge0$ and $H_\pm^\Fried\ge0$ as self-adjoint operators.

\medskip

\noindent (iii) The Friedrichs realizations factor as
\begin{align*}
H_\pm^\Fried
&=(\partial_x-S_\pm)^*
(\partial_x-S_\pm),
\qquad
S_\pm=\frac{u_{0,\pm}'}{u_{0,\pm}},
\end{align*}
and have no eigenvalues:
\begin{align*}
\Ker H_\pm^\Fried=\{0\},
\qquad
\sigma_{\mathrm{p}}(H_\pm^\Fried)=\varnothing.
\end{align*}

\medskip

\noindent (iv) The spectrum of $H_\pm^\Fried$ is purely absolutely continuous and given by
\begin{align*}
\sigma(H_\pm^\Fried)=\sigma_{\mathrm{ac}}(H_\pm^\Fried)=[0,\infty).
\end{align*}
\end{theorem}

The theorem gives the spectral meaning of the positive-square representation. Zero is not an eigenvalue, no positive eigenvalue exists, and the spectrum is purely absolutely continuous on $[0,\infty)$. Translated back to the Einstein--Maxwell variables, the fixed-background radiative sector contains no static or bound finite-$T$-energy hair: no admissible mode can remain trapped at the naked singularity.

This excludes the existence of the `algebraically special' growing modes found in earlier analyses~\cite{DottiGleiser2010} that did not impose finite $T$-energy. Those modes, while solutions of the linearized field equations, have infinite energy and do not lie in our physical energy space~(see Appendix~\ref{app:DG-instability} for the explicit computation). In other words, any would-be stationary perturbation of the naked singularity will necessarily have infinite $T$-energy or violate the Friedrichs regularity condition at the apex, and is therefore unphysical in the finite-energy framework adopted here.

The corresponding stationary evolution is the unitary group
\begin{align}
U(t)=\exp(-itH_\pm^\Fried),
\qquad
U(t)^*U(t)=I.
\end{align}

In the fixed-background sector described above, the absence of static finite-$T$-energy Einstein--Maxwell perturbations on the fixed superextremal Reissner--Nordstr\"om background shows up, in the gauge-invariant master-field representation, as the absence of a zero eigenvalue of $H_\pm^\Fried$. Any static finite-$T$-energy mode with $\ell\ge2$ would give an $L^2$ zero mode of $H_\pm^\Fried$, and conversely any such zero mode would reconstruct a static $\ell\ge2$ Einstein--Maxwell perturbation preserving the ADM charges. The short-range structure of the optical potentials also rules out nonzero eigenvalues, so there are no finite-energy bound states in any of the radiative sectors $\ell\ge2$. This result should be compared with the Reissner--Nordstr\"om black-hole mode-stability analysis of Kodama et al.~\cite{KodamaIshibashi2003,KodamaIshibashi2004PTP}. Those works derive and decouple the charged Einstein--Maxwell master equations and analyze mode stability in the black-hole regime. The stronger spectral conclusion proved here is specific to the action-completed superextremal finite-energy problem: the naked singularity supports no unstable or marginally bound finite-energy mode, and all finite-energy perturbations are scattering states. The Doob ground-state transform repackages this conclusion into the one-dimensional spectral picture
\begin{align*}
  H_\pm^\Fried = A_\pm^{\!*}A_\pm,\qquad
  \sigma(H_\pm^\Fried)=\sigma_{\mathrm{ac}}(H_\pm^\Fried)=[0,\infty),
\end{align*}
so unitarity of the stationary dynamics is already visible at the level of the optical half-line Hamiltonians. This spectral picture, in turn, points toward a Levinson-type index theorem for the optical half-line dynamics. One naturally expects an equality that relates the number of bound states of $H_\pm^\Fried$ to a winding number built from the scattering matrix, mirroring the index theorem of Inoue and Richard for the family $D_{\mu,\nu}$ on the half-line; see Section~8 of \cite{InouRichard}.

We stress that we are \emph{not} claiming that static $\ell\ge2$ Einstein--Maxwell configurations never occur in asymptotically flat electrovac spacetimes with a naked singularity. Our statements apply only to the fixed-background single-source sector defined above, where the background is a single superextremal Reissner--Nordstr\"om self-field with fixed $(M,Q)$ and no extra sources, struts, or external fields. Static multi-center or externally supported solutions describe different nonlinear backgrounds; in our viewpoint, they are not perturbations of the single naked Reissner--Nordstr\"om geometry.

Physically, this spectral picture has a clean interpretation in Einstein--Maxwell perturbation theory. Self-adjointness and nonnegativity of $H_\pm^\Fried$ yield the unitary \emph{wave} group on the Cauchy-data energy space through the standard first-order block generator, while the purely absolutely continuous spectrum says that this evolution has no bound component.

The action-selected dynamics is therefore purely scattering. The remaining question is whether this scattering evolution possesses a complete asymptotic representation at future null infinity.

\section{Radiation field at future null infinity}\label{sec:radiation-field}
Having shown that the Friedrichs dynamics contains no bound or threshold component, we now construct its asymptotic representation at future null infinity. The proof has three stages: we first define the channel radiation field, then place the asymptotic end in the scattering-manifold framework, and finally use the absence of hidden spectrum to identify the radiation map as a unitary translation representation.

Suppressing the parity and spherical-harmonic labels, we write the gauge-invariant master fields on the optical half-line as $\Phi_\pm(t,x)$. They satisfy the one-dimensional wave equations
\begin{align*}
\partial_t^2\Phi_\pm-\partial_x^2\Phi_\pm+V_\pm(x)\Phi_\pm=0,
\end{align*}
with potentials of the form
\begin{align*}
V_\pm(x)=\frac{\ell(\ell+1)}{x^2}+W_\pm(x),\qquad W_\pm\in L^1\big((1,\infty)\big),
\end{align*}
for each fixed multipole $\ell\ge2$; see \eqref{eq:short-range} and Appendix~\ref{app:appC}. For a fixed master sector and multipole, the field can be viewed as a one-dimensional radial wave on the optical half-line, moving under the influence of a centrifugal barrier and a short-range tail.

By the preceding Doob analysis, the Friedrichs Hamiltonians
$H_\pm^\Fried$ are nonnegative. The preceding spectral analysis proves that the Friedrichs Hamiltonian has no point spectrum and no zero-energy resonance. Thus, no finite-energy component is stored in a stationary or bound spectral sector.

The completion of $C_0^\infty(0,\infty)\times C_0^\infty(0,\infty)$ with respect to the norm induced by the mode energy is the Hilbert energy space
\begin{align*}
\HE=H_0^1(0,\infty)\times L^2(0,\infty).
\end{align*}
The wave group preserves the total mode $T$-energy, which, up to an inessential overall factor, is
\begin{align}
E_\pm[\Phi_\pm](t)
=\norm{\partial_t\Phi_\pm(t,\cdot)}_{L^2(0,\infty)}^2
+\norm{\partial_x\Phi_\pm(t,\cdot)}_{L^2(0,\infty)}^2
+\left\langle
\Phi_\pm(t,\cdot),
V_\pm\Phi_\pm(t,\cdot)
\right\rangle_{L^2(0,\infty)}.
\label{eq:energy}
\end{align}
The optical apex is not a second radiation channel: the action-selected form domain gives zero $T$-energy flux through $x=0$. The asymptotic outgoing channel is future null infinity, where the radiation field represents the conserved finite-energy state.

The notation $U(t)=\exp(-itH_\pm^\Fried)$ used above refers to the
unitary group associated directly with the self-adjoint spatial
Hamiltonian on $L^2(0,\infty)$. The notation $U_\pm(t)$ below refers
to the first-order wave group on the Cauchy-data energy space,
generated by the block operator $\mathcal A_\pm$.

The spectral theorem applied to $H_\pm^\Fried$ yields a nonnegative self-adjoint square root
$D_\pm:=(H_\pm^\Fried)^{1/2}$ on $L^2(0,\infty)$. Writing Cauchy data as
$(\Phi_\pm,\Pi_\pm)$ with $\Pi_\pm:=\partial_t\Phi_\pm$, the wave equation can be recast as a first-order Hamiltonian system on $\HE$ with generator
\begin{align}
\mathcal A_\pm
:=\begin{pmatrix}
0 & \operatorname{Id} \\
-H_\pm^\Fried & 0
\end{pmatrix},
\qquad
\Dom(\mathcal A_\pm)
=\Dom(H_\pm^\Fried)\oplus H_0^1(0,\infty).
\end{align}
Thus, the Cauchy problem for the master equation can be visualized as Hamiltonian flow on $\HE$, driven by the skew-adjoint generator $\mathcal A_\pm$. Stone's theorem guarantees that $\mathcal A_\pm$ generates a strongly continuous one-parameter unitary group
\begin{align}
U_\pm(t) := \exp(t\mathcal A_\pm), \qquad t\in\mathbb{R},
\end{align}
on $\HE$, and $U_\pm(t)$ preserves the $T$-energy inner product. Finite-energy Cauchy data at $t=0$ therefore propagate uniquely and isometrically to all $t\in\mathbb{R}$.

The radiation theorem below identifies this purely continuous finite-energy evolution with its outgoing profile at future null infinity.

To describe the energy carried to future null infinity, we switch to null coordinates
\begin{align*}
u=t-x,\qquad v=t+x.
\end{align*}
The vector field $\partial_t+\partial_x$ is tangent to outgoing radial
null geodesics $u=\mathrm{const}$, while
$\frac12(\partial_t-\partial_x)=\partial_u$ differentiates the outgoing
profile with respect to retarded time.

For Cauchy data $(\Phi_\pm(0,\cdot),\Pi_\pm(0,\cdot))\in C_0^\infty(0,\infty)\times C_0^\infty(0,\infty)$, the Friedlander--S\'a Barreto radiation-field construction yields the limit
\begin{equation}
  \calR_+[\Phi_\pm](u)
  := \lim_{x\to\infty}\frac12\,(\partial_t-\partial_x)\Phi_\pm(t,x)\bigg|_{t=u+x},
  \label{eq:rad-field}
\end{equation}
with the limit understood in $L^2_{\mathrm{loc}}(\mathbb{R}_u)$. The energy estimate below upgrades the resulting profile to an element of $L^2(\mathbb{R}_u)$ for finite-energy data.

Geometrically, $\calR_+[\Phi_\pm](u)$ is the outgoing amplitude recorded on the future-directed null ray labeled by the retarded time $u$. At fixed $u$, increasing $x$ tracks motion along the ray, while $(\partial_t-\partial_x)\Phi_\pm$ extracts the retarded-time derivative of the outgoing profile. The limit $x\to\infty$ at fixed $u$ is the one-dimensional analog of the classical Lax--Phillips construction in flat spacetime and of the Friedlander--S\'a Barreto radiation-field theorem for short-range perturbations of asymptotically Euclidean metrics.

The optical coordinate $x$ exhibits the asymptotic end of each static slice as asymptotically Euclidean in the sense of S\'a Barreto~\cite{SaBarreto2003}; the singular endpoint remains governed by the form-domain analysis. In the coordinates $(t,x,\omega)$, with $\omega\in\mathbb{S}^2$, the Lorentzian metric takes the form
\begin{align*}
g = -\alpha(x)^2\,\dd t^2 + \beta(x)^2\,\dd x^2 + r(x)^2\,\dd\Omega^2,
\end{align*}
where $\alpha,\beta,r$ are smooth positive functions on $(0,\infty)$ and $\dd\Omega^2$ is the round metric on $\mathbb{S}^2$. The optical radius $x$ is defined by $\dd x/\dd r = f(r)^{-1}$ (cf.\ Perlick's optical metric approach~\cite{Perlick2000}), so radial null geodesics satisfy $\abs{\dd x/\dd t}=1$. The large-$x$ expansion \eqref{eq:x-apex-rtoinf} implies
\begin{align*}
r(x)=x+\Order{\log x},\qquad \alpha(x),\beta(x)\to1\quad(x\to\infty),
\end{align*}
with corrections of order $\Order{x^{-1}}$.
Thus, the spatial metric on $t=\mathrm{const}$ slices is a smooth short-range perturbation of the
Euclidean metric, in the asymptotically Euclidean (or ``scattering'') class considered by Friedlander for
perturbed wave operators~\cite{Friedlander1980} and by S\'a Barreto for scattering manifolds~\cite{SaBarreto2003}. We write $(X,g_X)$ for the optical spatial slice and its scattering metric; before compactification, $X=(0,\infty)_x\times\mathbb{S}^2_\omega$, with $x$ playing the role of an
``almost Euclidean'' radial coordinate: for large $x$, the optical spheres $\{x=\mathrm{const}\}$ are very close to Euclidean spheres
in $\mathbb{R}^3$. The coefficients $\alpha$ and $\beta$ encode the redshift and the optical stretching of
the static slices, but their approach to $1$ at infinity ensures that, for wave
propagation, the geometry is only a short-range perturbation of flat space. In particular, if we rewrite the
optical metric in the original radial coordinate $r$ as in~\eqref{eq:optical-metric-r} and compactify
spatial infinity by introducing the inverse Euclidean radius $\rho:=r^{-1}$, then its asymptotic end admits
an asymptotically Euclidean scattering compactification with boundary $\partial X\simeq\mathbb{S}^2$ in the sense of S\'a Barreto~\cite[Def.~1.1]{SaBarreto2003}, following Melrose's scattering
metrics. In this representation, the master fields propagate according to a scalar wave equation on
$\mathbb{R}_t\times X$ with short-range lower-order terms.

In the full $(3+1)$-dimensional picture, the optical end is compactified by the boundary-defining function
\begin{align*}
\rho:=x^{-1}.
\end{align*}
In these scattering coordinates the separated master equations take the form
\begin{align}
\partial_t^2\Phi_\pm-\Delta_{g_X}\Phi_\pm+V_\pm^{\mathrm{eff}}\Phi_\pm=0,
\end{align}
for the scalar wave operator $\Box_{g_X}=-\partial_t^2+\Delta_{g_X}$ on $\mathbb R_t\times X$. Here $\Delta_{g_X}$ is the nonnegative Laplace--Beltrami operator, $V_\pm^{\mathrm{eff}}=\Order{x^{-2}}$ as $x\to\infty$, and $\partial_xV_\pm^{\mathrm{eff}}\in L^1((1,\infty))$.

On the one-dimensional scattering side, let $H_D=-\partial_x^2$ be the Dirichlet Laplacian on $L^2(0,\infty)$. The M{\o}ller operators $W_\pm(H_\pm^\Fried,H_D)$ compare the Friedrichs dynamics with this reference evolution. This is directly analogous to the half-line framework in which the generalized Fourier transforms $F^\pm_{\mu,\nu}$ intertwine $H_{\mu,\nu}$ with $H_D$ on $\mathbb R_+$~\cite{InouRichard}. The corresponding generalized Fourier transforms diagonalize $H_\pm^\Fried$ and realize the absolutely continuous evolution through a unitary scattering matrix. Combined with Kato's framework for wave operators~\cite{Kato}, this makes the perturbed wave group unitarily equivalent on $\HE$ to the free wave group on the optical half-line. In particular, \eqref{eq:rad-field} defines a partial isometry from $\HE$ into $L^2(\mathbb R_u)$ that intertwines the two wave groups.

Friedlander's forward radiation construction is represented schematically by
\begin{align*}
f\longmapsto\mathcal Mf,
\end{align*}
where $\mathcal Mf$ is defined as the asymptotic limit of $r^{(n-1)/2}\partial_tu$ along outgoing light rays. The construction extends by continuity to a bounded linear map
\begin{align*}
\mathcal M:\HE\longrightarrow L^2(\mathbb S^{n-1}\times\mathbb R),
\end{align*}
with $\norm{\mathcal Mf}_{L^2}\le\norm{f}_{\HE}$ and with time evolution intertwined with translations of the radiation profile.

The same symbol $\calR_+$ is used for the abstract radiation map and
for its restriction to the separated Einstein--Maxwell master
sector; the domain displayed in each statement distinguishes the two.

In the scattering-manifold formulation, S\'a Barreto's forward radiation field satisfies
\begin{align*}
\calR_+:\mathcal H_E(X)\longrightarrow L^2(\partial X\times\mathbb R_s).
\end{align*}
It is a surjective isometry on the closed subspace $\mathcal H_E^1\subset \mathcal H_E(X)$, and
$\mathcal H_E(X)=\mathcal H_E^1\oplus\Ker\calR_+$. For the present separated fields the absence of bound and threshold states gives $\Ker\calR_+=\{0\}$ and $\mathcal H_E^1=\HE$.

The factor $1/2$ in \eqref{eq:rad-field} is normalized so that the flat outgoing profile satisfies the exact energy identity below.

\begin{theorem}[S\'a Barreto~{\cite[Theorem~2]{SaBarreto2003}}]\label{thm:SaBarreto}
Let $(X,g_X)$ be a scattering manifold of dimension $3$ and let $V$ be a real-valued
short-range potential in the sense of~\cite[Def.~2.1]{SaBarreto2003} such that
the corresponding stationary operator $-\Delta_{g_X}+V$ has no eigenvalues and
no resonance at zero. Then, on the closed subspace
$\mathcal H_E^1\subset \mathcal H_E(X)$ generated by finite-energy solutions, the
forward radiation field extends to a surjective isometry
\begin{align*}
\calR_+:\mathcal H_E(X)\longrightarrow L^2(\partial X\times\mathbb{R}_s),
\end{align*}
and $\calR_+$ intertwines the wave group with translations in $s$.
\end{theorem}

We now verify that the optical Einstein--Maxwell problem satisfies these hypotheses. The compactification concerns only the asymptotically Euclidean end $x=\infty$; the singular endpoint $x=0$ remains governed separately by the action-selected Friedrichs form.

Let $(X,g_X)$ denote the spatial manifold obtained by rewriting the optical metric in the original radial coordinate $r$, cf.~\eqref{eq:optical-metric-r}, and compactifying spatial infinity by the inverse Euclidean radius
\begin{align*}
  \rho:=r^{-1}.
\end{align*}
In the coordinates $(\rho,\omega)$ the optical metric takes the form
\begin{align*}
  \gamma_{\mathrm{opt}}
  =f(1/\rho)^{-2}\,\frac{\dd\rho^2}{\rho^4}
  +f(1/\rho)^{-1}\,\frac{\dd\Omega^2}{\rho^2}.
\end{align*}
Here $f(1/\rho)\to1$ as $\rho\to0$ and is smooth near $\rho=0$, so this is a perturbation of the Euclidean scattering metric and $\partial X\simeq\mathbb S^2$.

\begin{lemma}\label{lem:scattering-hypotheses}
For each fixed multipole $\ell\ge2$ and sector $\pm$, the master equation
\eqref{eq:master-wave} is the corresponding spherical-harmonic component of a wave equation on the asymptotic end of the form
\begin{align*}
\partial_t^2\Phi_\pm - \Delta_{g_X}\Phi_\pm + V_\pm^{\mathrm{eff}}\Phi_\pm=0,
\end{align*}
where $\Delta_{g_X}$ is the nonnegative Laplace--Beltrami operator on the asymptotic end
and $V_\pm^{\mathrm{eff}}$ is real-valued and short-range there in the sense
of~\cite[Def.~2.1]{SaBarreto2003}. In particular,
\begin{align*}
V_\pm^{\mathrm{eff}}&\in C^\infty(X),\\
V_\pm^{\mathrm{eff}}&=O(x^{-2}),\qquad x\to\infty,\\
\partial_x V_\pm^{\mathrm{eff}}&\in L^1((1,\infty)).
\end{align*}
\end{lemma}

\begin{proof}
By construction, the spatial metric on $t=\mathrm{const}$ slices is
\begin{align*}
  \gamma_{\mathrm{opt}} = f(r)^{-2}\,\dd r^2 + f(r)^{-1}r^2\,\dd\Omega^2,
\end{align*}
cf.~\eqref{eq:optical-metric-r}, with $f(r)=1-2M/r+Q^2/r^2$ smooth and strictly positive on $(0,\infty)$.
Introducing the inverse Euclidean radius
\begin{align*}
  \rho := r^{-1},
\end{align*}
we obtain
\begin{align*}
  \gamma_{\mathrm{opt}}
  = f(1/\rho)^{-2}\,\frac{\dd\rho^2}{\rho^4}
  + f(1/\rho)^{-1}\,\frac{\dd\Omega^2}{\rho^2}.
\end{align*}
The functions $f(1/\rho)^{-2}$ and $f(1/\rho)^{-1}$ are smooth in $\rho$ near $\rho=0$ and tend to~$1$ as
$\rho\to0$, so $\gamma_{\mathrm{opt}}$ is a smooth perturbation of the Euclidean scattering metric in the
sense of~\cite[Section~1]{SaBarreto2003}. After possibly replacing $\rho$ by an equivalent boundary defining
function that agrees with $\rho$ up to second order at $\partial X$, this shows that the asymptotic end of each $t=\mathrm{const}$ slice has a scattering compactification with
$\partial X\simeq\mathbb{S}^2$.

Separation of variables on $\mathbb{S}^2$ identifies the radial master operators $-\partial_x^2+V_\pm(x)$
with the radial part of the scalar wave operator on $(X,g_X)$, so that $V_\pm^{\mathrm{eff}}$ differs from
$V_\pm$ by the usual centrifugal term and Jacobian corrections. The asymptotics
\eqref{eq:short-range} and the large-$r$ expansion~\eqref{eq:x-apex-rtoinf} imply that
\begin{align*}
  V_\pm^{\mathrm{eff}}(x)=\Order{x^{-2}}\quad\text{as }x\to\infty,
  \qquad
  \partial_x V_\pm^{\mathrm{eff}}\in L^1((1,\infty)),
\end{align*}
which is the short-range condition in~\cite[Def.~2.1]{SaBarreto2003} when expressed along the
radial optical coordinate $x$ (equivalently, in the boundary-defining coordinate $\rho=1/r$ of~\cite[Eq.~(1.1)]{SaBarreto2003}).
These statements concern only the asymptotic end $x=\infty$. The inverse-square core at $x=0$ remains governed separately by the action-selected Friedrichs form and its zero-flux property.
\end{proof}

To see how the one-dimensional operators $H_\pm^\Fried$ enter
the stationary operator on $(X,g_X)$, fix a sector $\pm$ and decompose
\begin{align*}
L^2(X)
\cong \bigoplus_{\ell,m}
L^2(0,\infty)\otimes\operatorname{span}\{Y_{\ell m}\},
\end{align*}
where $Y_{\ell m}$ are the standard spherical harmonics on $\mathbb{S}^2$.
With respect to this orthogonal decomposition, the operator
$-\Delta_{g_X}+V_\pm^{\mathrm{eff}}$ is unitarily equivalent to the direct
sum
\begin{align*}
-\Delta_{g_X}+V_\pm^{\mathrm{eff}}
\ \simeq\
\bigoplus_{\ell,m}H_{\pm,\ell}^\Fried,
\end{align*}
where each $H_{\pm,\ell}^\Fried$ is the Friedrichs realization on
$\Ltwo$ of a one-dimensional Schr\"odinger operator of the form
\begin{align*}
H_{\pm,\ell}^{\mathrm{formal}}
= -\partial_x^2 + \frac{\ell(\ell+1)}{x^2} + W_{\pm,\ell}(x),
\qquad W_{\pm,\ell}\in L^1\bigl((1,\infty)\bigr),
\end{align*}
obtained from separation of variables in the master equation
(see Section~\ref{sec:master-reduction} and~\eqref{eq:short-range}).
For $\ell\ge2$, the operators $H_{\pm,\ell}^\Fried$ coincide with the
$H_\pm^\Fried$ studied in Theorem~\ref{thm:doob-spectrum}, while for
$\ell=0,1$ the corresponding sectors are pure gauge and do not carry
propagating gravitational or electromagnetic degrees of freedom.

Geometrically, the appearance of a direct sum means that waves in distinct radiative channels evolve independently. Each channel sees an effective one-dimensional Hamiltonian acting along the radial optical coordinate, with a centrifugal term and a short-range tail at infinity. The radiation field can thus be constructed channelwise through the single-channel isometry $\calR_+$ and then reassembled as a Hilbert direct sum.

A potential obstruction to complete dispersal of wave energy would be the existence of a localized eigenmode (a bound state at some energy $\lambda$) or of a zero-energy threshold resonance. If $-\Delta_{g_X}+V_\pm^{\mathrm{eff}}$ had an $L^2$ eigenfunction at
energy $\lambda\in\mathbb{R}$, or a zero-energy resonance, then its angular
decomposition would have at least one nonvanishing component in some
$(\ell,m)$-mode. This component would be an $L^2$ eigenfunction, or a
threshold resonance, of $H_{\pm,\ell}^\Fried$ at the same energy $\lambda$.
The absence of eigenvalues and of a zero-energy resonance for all
$H_{\pm,\ell}^\Fried$ with $\ell\ge2$ implies that
$-\Delta_{g_X}+V_\pm^{\mathrm{eff}}$ has purely absolutely continuous
spectrum $[0,\infty)$ with neither bound states nor a threshold resonance
at zero. This verifies the spectral hypothesis in
Theorem~\ref{thm:SaBarreto} for our master fields.

The geometric and spectral inputs are now complete, so the abstract radiation theorem restricts to each Einstein--Maxwell master channel as follows.

\begin{theorem}\label{thm:radiation-field}
Let $V_\pm$ satisfy the core and tail conditions \eqref{eq:V-core} and \eqref{eq:short-range}, and let $H_\pm^\Fried$ be the Friedrichs realization on $L^2(0,\infty)$ of the optical master operator $-\partial_x^2+V_\pm$. Then:

\medskip

\noindent (i) The radiation-field map $\calR_+$ defined by \eqref{eq:rad-field} for smooth compactly supported Cauchy data extends by continuity to a surjective isometry
\begin{align*}
\calR_+:\HE\longrightarrow L^2(\mathbb{R}_u),
\end{align*}
where $\HE$ is the Cauchy-data Hilbert space equipped with the $T$-energy inner product.

\medskip

\noindent (ii) The map $\calR_+$ intertwines time translation with shifts in retarded time. For every $t\in\mathbb R$ and every finite-energy solution $\Phi_\pm$ of \eqref{eq:master-wave},
\begin{align*}
\calR_+\bigl[\Phi_\pm(t,\cdot)\bigr](u)
=\calR_+\bigl[\Phi_\pm(0,\cdot)\bigr](u+t).
\end{align*}

\medskip

\noindent (iii) In particular, the Einstein--Maxwell $T$-energy for each $(\ell,m)$ multipole mode is exactly the squared $L^2(\mathbb R_u)$ norm of its radiation profile at future null infinity.
\end{theorem}

For each master sector, separation of variables identifies $\HE$ with the corresponding subspace of the scattering energy space, and the forward radiation field induces the map
\begin{align*}
\calR_+:\HE\longrightarrow L^2(\mathbb R_u),
\end{align*}
defined by \eqref{eq:rad-field}. By Theorem~\ref{thm:SaBarreto}, this map is an isometry and is surjective because the channel has no bound state or threshold resonance. It also intertwines $U_\pm(t)$ with translation in $u$.

The ingredients now perform distinct and complementary tasks. The action-derived closed form selects the Friedrichs domain and excludes independent endpoint exchange. The Doob representation exposes the positivity of the selected Hamiltonian. The spectral analysis excludes bound states and a threshold resonance. The radiation theorem then identifies the remaining finite-energy dynamics with translations in retarded time at future null infinity. Consequently, every finite-energy channel state is represented isometrically by its radiation profile, with no energy absorbed at the apex, stored in a bound sector, or omitted from the forward radiation representation.

\begin{corollary}[Translation representation]
\label{thm:radiation-unitary}
Under the hypotheses of Theorem~\ref{thm:radiation-field}, the radiation map is a unitary isomorphism from the channel energy space onto the retarded-time translation representation.
\begin{align}
\calR_+:\HE\longrightarrow L^2(\mathbb{R}_u)
\end{align}
It intertwines the self-field evolution with time translations,
\begin{align}
\calR_+\bigl[U_\pm(t)\Phi_\pm\bigr](u)
=\calR_+[\Phi_\pm](u+t),
\qquad
t\in\mathbb R,
\end{align}
and is characterized by the energy identity \eqref{eq:isometry}: for every finite-energy solution,
\begin{align}
E_\pm[\Phi_\pm]
=\norm{\calR_+[\Phi_\pm]}_{L^2(\mathbb{R}_u)}^2.
\label{eq:isometry}
\end{align}
\end{corollary}

Equation~\eqref{eq:isometry} is the final quantitative statement of unitarity in the single-channel radiation picture. It identifies the master-field $T$-energy with the squared norm of the outgoing radiation profile. No energy is lost in the limit $x\to\infty$, and none is absorbed at the optical apex. An asymptotic observer who knows all assembled channel profiles can therefore reconstruct the complete finite-energy radiative state. This classical unitary equivalence supplies a Hilbert-space structure on which a later quantum construction could be based; no quantum theory follows from the radiation isometry alone.

The analytic proof has been carried out independently in each master
sector. Orthogonality of the harmonic decomposition then permits the
single-sector radiation maps to be reassembled as a direct sum of
unitary isometries.

The full chain from the covariant Einstein--Maxwell energy to its complete radiation representation is now established.

\section{Comments and conclusions}\label{sec:conclusion}

The central construction of this paper follows the action-derived
$T$-energy through its successive exact representations. The
quadratic action $S^{(2)}$ yields the Noether current $J_\mu^T$,
whose flux through a static slice decomposes as
\begin{align*}
E_{\mathrm{EM}}[\Phi]
=
\sum_{\pm,\ell\ge2,m}
E_{\pm,\ell m}[\Phi_{\pm,\ell m}].
\end{align*}
For each radiative master field,
\begin{align*}
E_{\pm,\ell m}
=
\frac12
\norm{\pi_{\pm,\ell m}}_{L^2}^2
+
q_\pm[\phi_{\pm,\ell m}],
\end{align*}
while the Doob and radiation identities are
\begin{align*}
q_\pm[\phi]
&=
\frac12
\norm{A_\pm\phi}_{L^2}^2,
\\
E_\pm[\Phi_\pm]
&=
\norm{\calR_+[\Phi_\pm]}_{L^2(\mathbb R_u)}^2.
\end{align*}
These are successive representations of the same action-derived
energy. The second variation supplies the quadratic field theory;
its $T$-Noether current supplies the conserved energy; the master
reduction supplies the modal decomposition; closure supplies the
admissible configuration-form domain; adjoining the $L^2$ velocity
component supplies the Cauchy-data energy space; the representation
theorem supplies the spatial Hamiltonian; the Doob identity supplies
the positive spatial square; and the radiation isometry supplies the
complete asymptotic norm.

The representations are successive in their construction, but their logical roles bifurcate at the Friedrichs Hamiltonian: self-adjointness supplies unitary bulk evolution, while the Doob and spectral analysis exclude hidden bound or threshold sectors; the radiation isometry reunites these branches in a complete asymptotic representation.

Because the harmonic sectors are orthogonal in the
Einstein--Maxwell $T$-energy and evolve independently, the
single-sector radiation maps reassemble as a direct sum of unitary
isometries. No additional endpoint condition, normalization, or
global energy form is introduced by this assembly. The master
reduction therefore preserves more than the equations of motion: it
carries one physical quadratic invariant from the covariant field
theory to its complete radiation representation.

At the point of closure, the physical and mathematical meanings coincide. The closed form is not merely a mathematical device applied to a previously identified physical state space. It is the completion of the conserved physical norm itself. Its domain is therefore simultaneously the set of configurations for which the physical energy is finite and the form domain from which the associated Hamiltonian is obtained. For each master form $q_\pm$,
\begin{align*}
\Dom(q_\pm)
=
H_0^1(0,\infty).
\end{align*}
This is simultaneously the form domain obtained by closure and the
configuration space selected by finiteness of the action-derived
spatial $T$-energy. The word ``admissible'' introduces no independent
condition; it refers only to membership in this completed
finite-energy domain. The form--operator correspondence is supplied by
the representation theorem, not by an additional endpoint choice. The
singular geometry fixes the coefficients entering $q_\pm$, but closure
of that form contributes no independent endpoint sector.

This inversion is logically prior even to the familiar distinction between the Schr\"odinger and Heisenberg pictures. In the former, states evolve under a fixed Hamiltonian; in the latter, the state is held fixed while observables evolve. Yet both pictures already presuppose a common Hilbert space, a common criterion of admissibility, and a common self-adjoint generator. Here that antecedent structure is itself derived. The quadratic Einstein--Maxwell action fixes the conserved $T$-energy on the primitive core; closure of its spatial form determines the finite-energy configuration domain; adjoining the $L^2$ velocity component gives the Cauchy-data energy space; and the representation theorem uniquely determines the spatial Hamiltonian associated with the closed form. The operator therefore does not stand outside the theory as an independently prescribed law that selects which formal solutions are physical. It is the canonical dynamical expression of the relations already encoded among the completed finite-energy states. State admissibility, endpoint behavior, and unitary evolution are consequently not separate choices but successive representations of one action-derived structure. In particular, the singular endpoint is not furnished with an auxiliary boundary law: the completed state space already exhausts the admissible finite-energy behavior, leaving no independent endpoint freedom for such a law to govern.

The theory does not overcome the singularity by removing it. The superextremal \RN\ metric remains a genuine curvature singularity, and nothing in the present analysis removes its geodesic incompleteness or divergent local invariants. What the construction establishes is that the unsupplemented linearized Einstein--Maxwell theory closes upon itself. Its action-derived spatial form closes to a finite-energy configuration domain; the kinetic $L^2$ component completes this to the Cauchy-data energy space; the representation theorem supplies the spatial Hamiltonian; and the resulting dynamics admits a complete radiation representation without an additional endpoint law. No independently prescribed endpoint law, energy, or degree of freedom is required.

The resulting silence of the apex has three distinct levels. First, finite-energy solutions generated by the Friedrichs realization have zero $T$-energy flux through the optical endpoint. Second, the absence of point spectrum and of a zero-energy resonance excludes localized finite-energy storage. Third, the unitary radiation representation excludes a nonradiating finite-energy remainder. The singularity is therefore neither a source nor a sink of energy; it supports no bound-state memory and contains no hidden finite-energy information reservoir.

The comparison with the Dotti--Gleiser modes is essential to this conclusion. Their non-Friedrichs apex behavior makes their Einstein--Maxwell $T$-energy divergent and therefore places them outside the form domain forced by the physical energy. This exclusion is not imposed after the instability has been found; it follows from the same action-derived norm that defines the theory before the spectral problem is solved. Hence no Dotti--Gleiser instability exists in the action-selected finite-$T$-energy Einstein--Maxwell dynamics, although the profiles remain formal solutions in the broader endpoint classes in which they were constructed.

The result also questions, within the stated limits, the predictability-based motivation for cosmic censorship. It does not establish that superextremal \RN\ geometries form in gravitational collapse, prove nonlinear stability, or disprove cosmic censorship. It does show that, on the fixed background and in the finite-$T$-energy radiative theory,
\begin{align*}
\text{naked singularity}
\;\not\Longrightarrow\;
\text{loss of deterministic finite-energy dynamics}.
\end{align*}
Nakedness alone is therefore insufficient to destroy finite-energy predictability.

The radiation representation does not recover information from a
concealed finite-energy singularity sector; it proves that no such
sector remains. Channel by channel, the bulk energy space and the
radiation space at future null infinity are unitary realizations of
the same information content. Orthogonal reassembly then gives the
corresponding statement for the full fixed-background radiative
system. The result therefore has a holographic form in this precise
classical scattering sense. The boundary representation is complete
because the hidden sector is empty, not because such a sector must be
recovered.

The preceding direction of determination depends essentially on treating
the background geometry as fixed. In a quantum-gravitational theory,
geometry could no longer remain an external antecedent. Geometry, energy,
admissible states, observables, spectra, and asymptotic data would instead
have to belong to one mutually consistent structure, schematically,
\begin{align*}
g
\longleftrightarrow
E
\longleftrightarrow
\mathcal H
\longleftrightarrow
H
\longleftrightarrow
\sigma(H)
\longleftrightarrow
\mathcal R.
\end{align*}
Geometry would determine the admissible quantum structure, while that
structure, through its states, observables, spectrum, and asymptotic
representation, would have to reconstruct the same geometry. This
reciprocal diagram is programmatic rather than a theorem proved here.
The present work establishes its forward fixed-background half: a
singular Einstein--Maxwell geometry determines its physical energy, its
finite-energy completion, its Hamiltonian, its spectral content, and its
complete radiation representation without supplemental endpoint data.

The construction proved in this paper is classical. The orthogonal sum of the classical channel phase spaces is the completed finite-$T$-energy phase space of the radiative linearized Einstein--Maxwell system. The intended quantum program begins by equipping this classical phase space with an appropriate complex structure so that its positive-frequency completion can be interpreted as a one-particle Hilbert space. Once compatibility with the classical symplectic form has been established, the radiation isometry could transport that one-particle structure to future null infinity. After normalization within such a quantization, the radiation profile is intended to play the role of a one-particle wave function---a probability amplitude in retarded time. Future null infinity is not introduced as an additional physical system; it is the asymptotic representation of the same bulk state.

A Weyl--CCR and Fock-space quantization is intended as a subsequent step. Such a construction would require the classical symplectic form, a compatible positive-frequency complex structure, the associated one-particle space, and a representation of the canonical commutation relations. The Doob profile $u_{0,\pm}$ serves in the present paper as the distinguished positive factorization profile of the classical Hamiltonian. Whether it also determines a preferred complex structure or vacuum state belongs to the future quantum construction. The present paper therefore establishes the classical foundation in the fixed-background radiative sector: a finite-energy phase space, a positive self-adjoint spatial Hamiltonian, unitary wave evolution, and a complete radiation representation. It does not yet prove the Born rule, a Fock representation, a quantum electron model, or nonlinear stability.

The absence of finite-energy bound modes in the static \RN\ problem also clarifies the scope of possible particle-like dynamics. Within the fixed-background finite-energy radiative sector, the present background is a silent scattering system and does not contain a finite-energy internal oscillatory sector. Any richer self-field behavior would therefore have to arise from structures absent from the present theorem, such as rotation, vorticity, nontrivial topology, nonlinear self-interaction, or couplings to additional fields. This is a conclusion about the limits of the static model, not evidence that any particular rotating or solitonic model succeeds.

Rotating Einstein--Maxwell geometries have a longstanding connection with particle-like interpretations. Carter analyzed the global structure of the Kerr family, while Burinskii developed explicit Kerr--Newman electron and spinning-particle models~\cite{Carter1968,Burinskii2008,Burinskii2012}. The present paper neither assumes nor endorses those models; they provide the historical context for treating the nonrotating \RN\ result as a possible static benchmark for a future Kerr--Newman analysis. The most direct extension of the present construction is the superextremal Kerr--Newman geometry. The \RN\ mechanism cannot simply be assumed to survive rotation. The present \RN\ result should be understood as a first static and spherically symmetric benchmark for a possible Kerr--Newman program, while every theorem proved here remains specific to the fixed nonrotating background. Static spherical symmetry supplies several structures essential to the present analysis: a globally preferred hypersurface-orthogonal Killing field, an orthogonal master-channel decomposition, a one-dimensional optical endpoint, and a positive $T$-energy form. In Kerr--Newman, frame dragging, vorticity, angular--radial coupling, and the ring singularity alter each part of that architecture. A future analysis would first have to identify the appropriate conserved quadratic energy and determine whether it remains positive and closable in the presence of rotation. Only then could one ask whether the resulting closed form selects a canonical self-adjoint realization, whether a positive factorization analogous to the Doob representation exists, and whether the radiative dynamics admits a complete representation at future null infinity. Failure of positivity, for example through an ergoregion or superradiant mechanism, would represent a genuine obstruction to the static finite-energy argument rather than a freedom to impose another endpoint condition.

A separate direction concerns the extremal limit $Q^2\to M^2$, where the finite optical apex is replaced by a degenerate horizon and the optical geometry changes qualitatively. It remains to determine whether the Hardy--Friedrichs mechanism has a controlled extremal limit or is replaced by a different threshold structure. Beyond linear theory, one must also ask whether the action-selected energy space persists under nonlinear evolution and whether superextremal configurations can arise dynamically. These questions lie beyond the fixed-background linear theory.

\myack{The author affirms that this work was executed independently, without grant funding or institutional supervision. Generative AI was employed exclusively for linguistic refinement and drafting support.}

\bibliographystyle{iopart-num}
\bibliography{refs}

\clearpage

\appendix

\appsection{Optical radius for superextremal \RN}\label{app:appA}

\subsection{
\texorpdfstring
{Expansion centered at $x=0$}
{Expansion centered at x=0}
}\label{ssec:A1}
We compute in closed form the optical radius $x(r)$ defined by~\eqref{eq:x-def},
\begin{equation}
  \frac{\dd x}{\dd r} = f(r)^{-1}, \qquad x(0)=0,
\end{equation}
and extract its asymptotic behavior both near the optical apex and toward spatial infinity. By definition of $f(r)$, we have
\begin{align}
  \frac{\dd x}{\dd r}
  = f(r)^{-1} 
  = \frac{1}{1 - \frac{2M}{r} + \frac{Q^{2}}{r^{2}}}
  = \frac{1}{\displaystyle \frac{r^{2} - 2Mr + Q^{2}}{r^{2}}}
  = \frac{r^{2}}{r^{2} - 2Mr + Q^{2}}.
     \label{eq:opt-int-step2}
\end{align}
Next we complete the square in the quadratic polynomial appearing in the denominator:
\begin{equation}
  r^{2} - 2Mr + Q^{2}
  = r^{2} - 2Mr + M^{2} + Q^{2} - M^{2} 
  = (r - M)^{2} + (Q^{2} - M^{2}).\label{eq:opt-completesquare}
\end{equation}
In the superextremal regime $|Q|>M$, we introduce
\begin{equation}
  \alpha := \sqrt{Q^{2} - M^{2}} > 0,
\end{equation}
so that~\eqref{eq:opt-completesquare} becomes
\begin{equation}
  r^{2} - 2Mr + Q^{2} = (r - M)^{2} + \alpha^{2}. \label{eq:opt-denominator-alpha}
\end{equation}
Substituting~\eqref{eq:opt-denominator-alpha} into~\eqref{eq:opt-int-step2} yields
\begin{equation}
  \frac{\dd x}{\dd r}
  = \frac{r^{2}}{(r - M)^{2} + \alpha^{2}}. \label{eq:opt-int-prechange}
\end{equation}
To integrate~\eqref{eq:opt-int-prechange}, we first rewrite the numerator $r^{2}$ in terms of $r - M$:
\begin{equation}
  r^{2}
  = (r - M + M)^{2} 
  = (r - M)^{2} + 2M(r - M) + M^{2}. \label{eq:opt-num-expand}
  \end{equation}
Using~\eqref{eq:opt-num-expand} in~\eqref{eq:opt-int-prechange} we obtain
\begin{align}
  \frac{\dd x}{\dd r}
  = \frac{(r - M)^{2} + 2M(r - M) + M^{2}}{(r - M)^{2} + \alpha^{2}} 
  = \frac{(r - M)^{2}}{(r - M)^{2} + \alpha^{2}}
   + \frac{2M(r - M)}{(r - M)^{2} + \alpha^{2}}
   + \frac{M^{2}}{(r - M)^{2} + \alpha^{2}}.
   \label{eq:opt-split1}
\end{align}
For the first fraction in~\eqref{eq:opt-split1}, we write
\begin{align}
  \frac{(r - M)^{2}}{(r - M)^{2} + \alpha^{2}}
  = \frac{(r - M)^{2} + \alpha^{2} - \alpha^{2}}{(r - M)^{2} + \alpha^{2}} 
  = 1 - \frac{\alpha^{2}}{(r - M)^{2} + \alpha^{2}}. \label{eq:opt-firstfraction}
\end{align}
Substituting~\eqref{eq:opt-firstfraction} into~\eqref{eq:opt-split1} gives
\begin{align}
  \frac{\dd x}{\dd r}
  &= 1 - \frac{\alpha^{2}}{(r - M)^{2} + \alpha^{2}}
   + \frac{2M(r - M)}{(r - M)^{2} + \alpha^{2}}
   + \frac{M^{2}}{(r - M)^{2} + \alpha^{2}} \nonumber\\
  &= 1 + \frac{2M(r - M)}{(r - M)^{2} + \alpha^{2}}
   + \frac{M^{2} - \alpha^{2}}{(r - M)^{2} + \alpha^{2}}.
   \label{eq:opt-split2}
\end{align}
Recalling $\alpha^{2} = Q^{2} - M^{2}$, we compute
\begin{equation}
  M^{2} - \alpha^{2}
  = M^{2} - (Q^{2} - M^{2})
  = 2M^{2} - Q^{2},
\end{equation}
so that~\eqref{eq:opt-split2} can be rewritten as
\begin{equation}
  \frac{\dd x}{\dd r}
  = 1 + \frac{2M(r - M)}{(r - M)^{2} + \alpha^{2}}
   + \frac{2M^{2} - Q^{2}}{(r - M)^{2} + \alpha^{2}}.
   \label{eq:opt-split-final}
\end{equation}
We now integrate~\eqref{eq:opt-split-final} with respect to $r$. Integrating term by term, we obtain
\begin{align}
  x(r)
  &= \int \left[ 1
   + \frac{2M(r - M)}{(r - M)^{2} + \alpha^{2}}
   + \frac{2M^{2} - Q^{2}}{(r - M)^{2} + \alpha^{2}} \right] \dd r \nonumber\\
  &= \int 1\,\dd r
   + 2M \int \frac{(r - M)}{(r - M)^{2} + \alpha^{2}} \,\dd r
   + (2M^{2} - Q^{2}) \int \frac{\dd r}{(r - M)^{2} + \alpha^{2}}.
   \label{eq:opt-integrals-split}
\end{align}
For the second and third integrals in~\eqref{eq:opt-integrals-split}, we introduce the substitution $u := r - M$, $\dd u = \dd r$. Then
\begin{align}
  \int 1\,\dd r
  &= r + C_{1}, \label{eq:opt-int1}\\
  2M \int \frac{(r - M)}{(r - M)^{2} + \alpha^{2}} \,\dd r
  &= 2M \int \frac{u}{u^{2} + \alpha^{2}} \,\dd u 
  = 2M \cdot \frac{1}{2} \ln\!\bigl(u^{2} + \alpha^{2}\bigr) + C_{2}
     \label{eq:opt-int2-pre} \\
  &= M \ln\!\bigl(u^{2} + \alpha^{2}\bigr) + C_{2} 
  = M \ln\!\bigl((r - M)^{2} + \alpha^{2}\bigr) + C_{2},
     \label{eq:opt-int2}
\end{align}
where in~\eqref{eq:opt-int2-pre} we have used the standard identity
\begin{equation}
  \int \frac{u}{u^{2} + \alpha^{2}} \,\dd u
  = \frac{1}{2} \ln\!\bigl(u^{2} + \alpha^{2}\bigr) + \text{const}.
\end{equation}
Similarly,
\begin{align}
  (2M^{2} - Q^{2}) \int \frac{\dd r}{(r - M)^{2} + \alpha^{2}}
  &= (2M^{2} - Q^{2}) \int \frac{\dd u}{u^{2} + \alpha^{2}} \nonumber\\
  &= (2M^{2} - Q^{2}) \cdot \frac{1}{\alpha}
     \arctan\!\left(\frac{u}{\alpha}\right) + C_{3}
     \label{eq:opt-int3-pre} \\
  &= \frac{2M^{2} - Q^{2}}{\alpha}
     \arctan\!\left(\frac{r - M}{\alpha}\right) + C_{3},
     \label{eq:opt-int3}
\end{align}
where in~\eqref{eq:opt-int3-pre} we have used the standard identity
\begin{equation}
  \int \frac{\dd u}{u^{2} + \alpha^{2}}
  = \frac{1}{\alpha} \arctan\!\left(\frac{u}{\alpha}\right) + \text{const}.
\end{equation}

Combining~\eqref{eq:opt-int1}, \eqref{eq:opt-int2}, and~\eqref{eq:opt-int3}, and absorbing the constants $C_{1}$, $C_{2}$, and $C_{3}$ into a single integration constant $C$, we obtain
\begin{align*}
  x(r)
  &= r
   + \frac{2M^{2} - Q^{2}}{\alpha}
     \arctan\!\left(\frac{r - M}{\alpha}\right)
   + M \ln\!\bigl((r - M)^{2} + \alpha^{2}\bigr)
   + C \nonumber\\
  &= r
   + \frac{2M^{2} - Q^{2}}{\sqrt{Q^{2} - M^{2}}}
     \arctan\!\left(\frac{r - M}{\sqrt{Q^{2} - M^{2}}}\right)
   + M \ln\!\bigl(r^{2} - 2Mr + Q^{2}\bigr)
   + C,
\end{align*}
where in the last line we have used $\alpha^{2} = Q^{2} - M^{2}$ and $(r - M)^{2} + \alpha^{2} = r^{2} - 2Mr + Q^{2}$. 

Here $C$ is an integration constant and is fixed by imposing the normalization $x(0)=0$, which ensures that the optical coordinate origin coincides with the location of the naked singularity at $r=0$. This choice is compatible with the small-$r$ expansion.

The small-$r$ expansion~\eqref{eq:x-apex-rto0-app} in Appendix~\ref{ssec:A2} shows that the optical distance to the naked singularity is finite: $x\sim r^{3}/(3Q^{2})$. The large-$r$ behavior~\eqref{eq:x-apex-rtoinf-app} in Appendix~\ref{ssec:A3} shows that spatial infinity is logarithmically ``stretched'' in the optical coordinate, reflecting the Newtonian mass term $2M/r$ in $f(r)$ and mirroring the well-known tortoise coordinate in the black-hole regime~\cite{ChirentiSaaSkakala2012,DottiGleiser2010}. These features play a key role in the Hardy--Friedrichs analysis and in the construction of the radiation field at future null infinity.

\subsection{
\texorpdfstring
{Expansion of $x$ as $r\to0$}
{Expansion of x as r approaches 0}
}\label{ssec:A2}
To understand the global geometry encoded by $x$, we now extract the leading behavior of $x(r)$ as $r\to0^+$. For small $r$, we work directly with~\eqref{eq:x-def}. Using
\begin{align*}
f(r)=1-\frac{2M}{r}+\frac{Q^{2}}{r^{2}}
=\frac{r^{2}-2Mr+Q^{2}}{r^{2}},
\end{align*}
we obtain
\begin{align}
\frac{\dd x}{\dd r}
=f(r)^{-1}
=\frac{r^{2}}{r^{2}-2Mr+Q^{2}}
=\frac{r^{2}}{Q^{2}}\,
\frac{1}{1-\frac{2M}{Q^{2}}r+\frac{r^{2}}{Q^{2}}}.
\end{align}
For $r$ sufficiently small, the denominator can be expanded as a geometric series in $r$, giving
\begin{align}
\frac{1}{1-\frac{2M}{Q^{2}}r+\frac{r^{2}}{Q^{2}}}
&=1+\frac{2M}{Q^{2}}r+\Order{r^{2}}.
\end{align}
Hence
\begin{align}
\frac{\dd x}{\dd r}
=\frac{r^{2}}{Q^{2}}\Bigl(1+\frac{2M}{Q^{2}}r+\Order{r^{2}}\Bigr)
=\frac{r^{2}}{Q^{2}}+\Order{r^{3}},
\qquad r\to0^+.
\end{align}
Integrating from $0$ to $r$ and using $x(0)=0$ yields
\begin{align}
x(r)
&=\int_{0}^{r}\frac{\dd x}{\dd \rho}\,\dd \rho
=\int_{0}^{r}\left(\frac{\rho^{2}}{Q^{2}}+\Order{\rho^{3}}\right)\dd \rho
=\frac{r^{3}}{3Q^{2}}+\Order{r^{4}},
\qquad r\to0^+.
\label{eq:x-apex-rto0-app}
\end{align}
Thus, $x(r)$ tends to zero like $r^{3}$ as $r\downarrow0$. In particular, finite optical distance $x$ corresponds to a finite value of $r$, and the naked singularity at $r=0$ sits at finite optical distance---a feature that will be crucial when formulating the self-field problem on the optical half-line.

\subsection{
\texorpdfstring
{Expansion of $x$ as $r\to\infty$}
{Expansion of x as r approaches infinity}
}\label{ssec:A3}
To understand the global geometry encoded by $x$ for large $r$, it is more convenient to expand $f(r)^{-1}$ directly in inverse powers of $r$. Writing
\begin{align}
f(r)
=1-\frac{2M}{r}+\frac{Q^{2}}{r^{2}},
\end{align}
we regard $a:=\frac{2M}{r}-\frac{Q^{2}}{r^{2}}$ as a small parameter for $r\to\infty$ and use
\begin{align*}
\frac{1}{1-a}=1+a+a^{2}+\Order{a^{3}}.
\end{align*}
Substituting $a$ and keeping terms up to order $r^{-2}$ gives
\begin{align}
f(r)^{-1}
&=1+\left(\frac{2M}{r}-\frac{Q^{2}}{r^{2}}\right)
+\left(\frac{2M}{r}-\frac{Q^{2}}{r^{2}}\right)^{2}
+\Order{r^{-3}}\\
&=1+\frac{2M}{r}+\frac{4M^{2}-Q^{2}}{r^{2}}+\Order{r^{-3}},
\qquad r\to\infty.
\end{align}
Integrating term by term, we obtain
\begin{align}
x(r)
&=\int^{r}\left(1+\frac{2M}{\rho}+\frac{4M^{2}-Q^{2}}{\rho^{2}}+\Order{\rho^{-3}}\right)\dd \rho\\
&= r + 2M\ln r -\frac{4M^{2}-Q^{2}}{r}+\Order{r^{-2}}+\text{const.}
\end{align}
Absorbing the additive constant and the $r^{-1}$ term into the $\Order{1}$ remainder, we can summarize the large-$r$ behavior as
\begin{align}
x(r)=r+2M\ln r+\Order{1},
\qquad r\to\infty.
\label{eq:x-apex-rtoinf-app}
\end{align}
Combining~\eqref{eq:x-apex-rto0-app} and~\eqref{eq:x-apex-rtoinf-app}, we see that $x\downarrow0$ as $r\downarrow0$ and $x\to\infty$ as $r\to\infty$, with $x'(r)=f(r)^{-1}>0$ everywhere. Thus, $x$ is indeed a global optical coordinate on the half-line $(0,\infty)$.

\appsection{Inverse-square cores of the Dotti--Gleiser master potentials}\label{app:appB}

We emphasize that the potentials analyzed in this appendix are exactly the Einstein--Maxwell master potentials used in the main text, not those of a different or auxiliary problem. In Section~\ref{sec:master-reduction}, the radiative degrees of freedom are encoded in the Kodama--Ishibashi gauge-invariant fields \(\phi_{\pm,\ell m}(t,x)\) on the optical half-line, which satisfy the master equation \eqref{eq:master-wave} with spatial operator \(H_\pm^{\mathrm{formal}}=-\partial_x^2+V_\pm(x)\), as in \eqref{eq:H-formal}. In four-dimensional Reissner--Nordstr\"om, the even/odd master variables of Dotti and Gleiser are related to these Kodama--Ishibashi fields by an \(\ell\)-dependent, \(r\)-regular, invertible linear transformation on each fixed multipole, so both descriptions generate the same Schr\"odinger operators on \(L^2(0,\infty)_x\) with the same tortoise/optical coordinate \(x\) defined by \eqref{eq:x-def}. The formulas \eqref{eq:DG-master-potential}--\eqref{eq:DG-falpha-def} used here are simply the Dotti--Gleiser expressions for those very potentials \(V_\pm\), written in terms of \(f(r)=1-2M/r+Q^2/r^2\) and its derivatives. Substituting the small-\(r\) expansion of \(x(r)\) from Appendix~\ref{app:appA} into those expressions thus reproduces the small-\(x\) behavior of the same potentials \(V_\pm(x)\) that appear in the core forms \(q_\pm\) in \eqref{eq:qpm-def} and the Friedrichs operators \(H_\pm^{\Fried}\) studied in Sections~\ref{sec:hardy}--\ref{sec:friedrichs}. In particular, the inverse-square cores obtained here belong to the full linear Einstein--Maxwell self-field problem for a single superextremal Reissner--Nordstr\"om point charge, with no change of background, variables, or dynamics between the main text and this appendix.

For each $\alpha=1,2$ and sign $\pm$, the
four-dimensional Einstein--Maxwell master potentials can be written in the
tortoise/optical coordinate $x$ as follows (cf.\ Equations~(6) and (7) of \cite{DottiGleiser2010}):
\begin{equation}
  V_\alpha^{\pm}(x)
  = \pm \beta_\alpha\,\frac{\dd f_\alpha}{\dd x}
    + \beta_\alpha^2 f_\alpha^2
    + \kappa\, f_\alpha,
  \label{eq:DG-master-potential}
\end{equation}
where
\begin{align*}
  \kappa = (\ell-1)\,\ell\,(\ell+1)\,(\ell+2),
  \qquad
  \beta_\alpha
  = 3M + (-1)^\alpha\sqrt{9M^2 + 4Q^2(\ell-1)(\ell+2)},
\end{align*}
and
\begin{equation}
  f_\alpha(r)
  = \frac{f(r)}{r\beta_\alpha + (\ell-1)(\ell+2)\,r^2},
  \qquad
  f(r) = 1 - \frac{2M}{r} + \frac{Q^2}{r^2}.
  \label{eq:DG-falpha-def}
\end{equation}
The tortoise/optical coordinate $x$ is defined by
\begin{equation}
  \frac{\dd x}{\dd r} = f(r)^{-1},
  \qquad
  x(0)=0,
  \label{eq:x-tortoise-def}
\end{equation}
so that $\dd/\dd x = f(r)\,\dd/\dd r$.

We first determine the leading small-$r$ behavior of $x$. Near $r=0$,
\begin{equation}
  f(r)
  = 1 - \frac{2M}{r} + \frac{Q^2}{r^2}
  = \frac{Q^2}{r^2}
    \Bigl( 1 - \frac{2M}{Q^2}r + \Order{r^2} \Bigr),
  \qquad r\to0^+,
  \label{eq:f-small-r}
\end{equation}
hence
\begin{equation}
  f(r)^{-1}
  = \frac{r^2}{Q^2}\Bigl(1 + \Order{r}\Bigr),
  \qquad r\to0^+.
\end{equation}
From \eqref{eq:x-tortoise-def}, we obtain
\begin{equation}
  \frac{\dd x}{\dd r}
  = \frac{r^2}{Q^2} + \Order{r^3},
  \qquad r\to0^+,
\end{equation}
so integrating from $0$ to $r$ gives
\begin{equation}
  x(r)
  = \int_0^r \frac{\rho^2}{Q^2}\,\dd\rho + \Order{r^4}
  = \frac{r^3}{3Q^2} + \Order{r^4},
  \qquad r\to0^+.
  \label{eq:x-small-r}
\end{equation}
Inverting \eqref{eq:x-small-r}, we find
\begin{equation}
  r^3 = 3Q^2 x\bigl(1 + \Order{r}\bigr),
  \qquad
  r^6 = 9Q^4 x^2\bigl(1 + \Order{r}\bigr),
  \qquad
  \frac{Q^4}{r^6}
  = \frac{1}{9x^2}\bigl(1 + o(1)\bigr),
  \qquad x\downarrow0.
  \label{eq:r6-vs-x}
\end{equation}
Next we expand $f_\alpha(r)$ near $r=0$. From \eqref{eq:DG-falpha-def},
\begin{equation}
  f_\alpha(r)
  = \frac{f(r)}{r\beta_\alpha + (\ell-1)(\ell+2)\,r^2}.
\end{equation}
Factor the denominator:
\begin{equation}
  r\beta_\alpha + (\ell-1)(\ell+2)\,r^2
  = r\beta_\alpha\Bigl(1 + \frac{(\ell-1)(\ell+2)}{\beta_\alpha}r\Bigr),
\end{equation}
so for small $r$,
\begin{equation}
  \frac{1}{r\beta_\alpha + (\ell-1)(\ell+2)\,r^2}
  = \frac{1}{\beta_\alpha r}
    \Bigl(1 - \frac{(\ell-1)(\ell+2)}{\beta_\alpha}r + \Order{r^2}\Bigr).
\end{equation}
Using \eqref{eq:f-small-r},
\begin{align}
  f_\alpha(r)
  &= \Bigl(\frac{Q^2}{r^2} - \frac{2M}{r} + 1\Bigr)
     \frac{1}{\beta_\alpha r}
     \Bigl(1 - \frac{(\ell-1)(\ell+2)}{\beta_\alpha}r + \Order{r^2}\Bigr)
     \nonumber\\
  &= \frac{Q^2}{\beta_\alpha r^3} + g_\alpha(r),
  \qquad
  g_\alpha(r) = \Order{r^{-2}},
  \qquad r\to0^+.
  \label{eq:falpha-small-r}
\end{align}
Differentiating \eqref{eq:falpha-small-r} with respect to $r$ gives
\begin{equation}
  \frac{\dd f_\alpha}{\dd r}
  = -\frac{3Q^2}{\beta_\alpha r^4} + g_\alpha'(r),
  \qquad
  g_\alpha'(r) = \Order{r^{-3}},
  \qquad r\to0^+.
  \label{eq:dfalpha-dr}
\end{equation}
Since $\dd/\dd x = f(r)\,\dd/\dd r$ and $f(r)=Q^2 r^{-2} + \Order{r^{-1}}$ as
$r\to0^+$, we obtain from \eqref{eq:dfalpha-dr}
\begin{align}
  \frac{\dd f_\alpha}{\dd x}
  &= f(r)\,\frac{\dd f_\alpha}{\dd r}
   = \Bigl(\frac{Q^2}{r^2} + \Order{r^{-1}}\Bigr)
     \Bigl(-\frac{3Q^2}{\beta_\alpha r^4} + \Order{r^{-3}}\Bigr)
     \nonumber\\
  &= -\frac{3Q^4}{\beta_\alpha r^6} + \Order{r^{-5}},
  \qquad r\to0^+.
  \label{eq:dfalpha-dx}
\end{align}
We now assemble the pieces in \eqref{eq:DG-master-potential}. Using
\eqref{eq:falpha-small-r} and \eqref{eq:dfalpha-dx},
\begin{align}
  \pm \beta_\alpha\frac{\dd f_\alpha}{\dd x}
  &= \pm\beta_\alpha\Bigl(-\frac{3Q^4}{\beta_\alpha r^6} + \Order{r^{-5}}\Bigr)
   = \mp\frac{3Q^4}{r^6} + \Order{r^{-5}},
   \label{eq:term1}\\[0.5em]
  \beta_\alpha^2 f_\alpha^2
  &= \beta_\alpha^2\Bigl(\frac{Q^2}{\beta_\alpha r^3} + g_\alpha(r)\Bigr)^2
   = \beta_\alpha^2\Bigl(\frac{Q^4}{\beta_\alpha^2 r^6}
      + \frac{2Q^2}{\beta_\alpha r^3}g_\alpha(r)
      + g_\alpha(r)^2\Bigr)
      \nonumber\\
  &= \frac{Q^4}{r^6} + \Order{r^{-5}},\\[0.5em]
  \kappa f_\alpha
  &= \kappa\Bigl(\frac{Q^2}{\beta_\alpha r^3} + g_\alpha(r)\Bigr)
   = \Order{r^{-3}}.
   \label{eq:term3}
\end{align}
Substituting \eqref{eq:term1}--\eqref{eq:term3} into
\eqref{eq:DG-master-potential}, we obtain
\begin{align}
  V_\alpha^{+}(r)
  &= \Bigl(-\frac{3Q^4}{r^6} + \Order{r^{-5}}\Bigr)
     + \Bigl(\frac{Q^4}{r^6} + \Order{r^{-5}}\Bigr)
     + \Order{r^{-3}}
     \nonumber\\[0.25em]
  &= -\frac{2Q^4}{r^6} + \Order{r^{-5}},
     \label{eq:Vplus-r}\\[0.5em]
  V_\alpha^{-}(r)
  &= \Bigl(+\frac{3Q^4}{r^6} + \Order{r^{-5}}\Bigr)
     + \Bigl(\frac{Q^4}{r^6} + \Order{r^{-5}}\Bigr)
     + \Order{r^{-3}}
     \nonumber\\[0.25em]
  &= \frac{4Q^4}{r^6} + \Order{r^{-5}},
     \label{eq:Vminus-r}
\end{align}
as $r\to0^+$. The coefficients $-2$ and $4$ are independent of $\alpha$, $M$,
$Q$, and $\ell$; all dependence on these parameters resides in the
$\Order{r^{-5}}$ and less singular terms.

Finally, we express $V_\alpha^{\pm}$ in terms of the optical coordinate $x$
using \eqref{eq:r6-vs-x}. From \eqref{eq:Vplus-r} and \eqref{eq:r6-vs-x},
\begin{align*}
  V_\alpha^{+}(x)
  &= -2\,\frac{Q^4}{r(x)^6} + \Order{r(x)^{-5}}
   = -2\,\frac{1}{9x^2}\bigl(1 + o(1)\bigr)
     + \Order{x^{-5/3}}
     \nonumber\\
  &= -\frac{2}{9}\,x^{-2} + o\bigl(x^{-2}\bigr),
     \qquad x\downarrow0,
\end{align*}
since $x^{-5/3}=o(x^{-2})$ as $x\downarrow0$. Similarly, from
\eqref{eq:Vminus-r},
\begin{align*}
  V_\alpha^{-}(x)
  &= 4\,\frac{Q^4}{r(x)^6} + \Order{r(x)^{-5}}
   = 4\,\frac{1}{9x^2}\bigl(1 + o(1)\bigr)
     + \Order{x^{-5/3}}
     \nonumber\\
  &= \frac{4}{9}\,x^{-2} + o\bigl(x^{-2}\bigr),
     \qquad x\downarrow0.
\end{align*}
In particular, the two even-parity Einstein--Maxwell master potentials in the
optical radius $x$ have Hardy-type inverse-square cores
\begin{align*}
  V_+(x) = -\frac{2}{9}\,x^{-2} + o\bigl(x^{-2}\bigr),
  \qquad
  V_-(x) = \frac{4}{9}\,x^{-2} + o\bigl(x^{-2}\bigr),
  \qquad x\downarrow0,
\end{align*}
so that, in the notation $V_\pm(x)=C_\pm x^{-2}+o(x^{-2})$, we have
\begin{align*}
  C_+ = -\frac{2}{9} = \nu_+^2 - \frac14,
  \qquad
  C_- = \frac{4}{9}  = \nu_-^2 - \frac14,
  \qquad
  \nu_+ = \frac{1}{6},
  \qquad
  \nu_- = \frac{5}{6},
\end{align*}
with $\nu_\pm\in(0,1)$, as required for the Hardy window in
Section~\ref{sec:hardy}.

\appsection{Optical map and short-range tail of static \RN}\label{app:appC}

Here we rewrite the Ishibashi--Kodama (IK) master potentials in optical variables and isolate their behavior at large optical radius. Our goal is to verify their centrifugal-plus-short-range structure at large $x$:
\begin{align*}
V_\pm(x)=\frac{\ell(\ell+1)}{x^{2}}+W_\pm(x),\qquad W_\pm\in L^{1}\big((1,\infty)\big).
\end{align*}
This decomposition is used in the Hardy--Friedrichs analysis. Throughout, we work with four-dimensional, asymptotically flat, spherically symmetric Einstein--Maxwell backgrounds, so that the IK parameters specialize to
\begin{align}
(\Lambda,K,n)&=(0,1,2), \label{app:eq:ikn}\\
k^{2}&=\ell(\ell+n-1)=\ell(\ell+1), \label{app:eq:k2}\\
k_v^{2}&=\ell(\ell+1)-1, \label{app:eq:kv2}\\
k_v^{2}+1&=\ell(\ell+1)=k^2, \label{app:eq:kv2plus}\\
m&=k^{2}-nK=\ell(\ell+1)-2, \label{app:eq:mdef}
\end{align}
exactly as in Section~3 of \cite{IshibashiKodama2011PTPS}. We keep $\ell\ge2$ fixed and suppress the $(\ell,m)$ indices.

\subsection{Vector (axial) sector}
The IK vector-type (axial) master potentials on a general static Einstein--Maxwell background are given in Equation~5.23 of \cite{IshibashiKodama2011PTPS} as
\begin{align}
V_{V\pm}(r)
=\frac{f(r)}{r^{2}}\Bigg[
k_v^{2}
+\frac{n^{2}-2n+4}{4}\,K
-\frac{n(n-2)}{4}\,\Lambda r^{2}
+\frac{n(5n-2)Q^{2}}{4r^{2n-2}}
+\frac{\mu_\pm}{r^{\,n-1}}
\Bigg],
\label{app:eq:V_V}
\end{align}
where $f(r)$ is the usual lapse function appearing in the static metric
\begin{align}
\dd s^{2}=-f(r)\,\dd t^{2}+f(r)^{-1}\,\dd r^{2}+r^{2}\,\dd\Omega^{2},
\end{align}
and $\mu_\pm$ are the sector-dependent constants defined in Section~5 of \cite{IshibashiKodama2011PTPS}.

Specializing \eqref{app:eq:V_V} to the four-dimensional, asymptotically flat case \eqref{app:eq:ikn} proceeds term by term. First,
\begin{align}
\frac{n^{2}-2n+4}{4}\,K
=\frac{2^{2}-2\cdot2+4}{4}\cdot1
=\frac{4-4+4}{4}
=1.
\end{align}
Next,
\begin{align}
-\frac{n(n-2)}{4}\,\Lambda r^{2}
=-\frac{2\cdot0}{4}\,\Lambda r^{2}=0,
\end{align}
since $\Lambda=0$, and the charge term simplifies to
\begin{align}
\frac{n(5n-2)Q^{2}}{4r^{2n-2}}
=\frac{2(10-2)Q^{2}}{4r^{2\cdot2-2}}
=\frac{2\cdot8}{4}\,\frac{Q^{2}}{r^{2}}
=4\,\frac{Q^{2}}{r^{2}}.
\end{align}
Finally, $r^{n-1}=r^{1}=r$, so the $\mu_\pm$-term is simply $\mu_\pm/r$. Using \eqref{app:eq:kv2plus}, $k_v^{2}+1=\ell(\ell+1)$, the full bracket in \eqref{app:eq:V_V} becomes
\begin{align}
k_v^2+1+\frac{n(5n-2)Q^{2}}{4r^{2n-2}}+\frac{\mu_\pm}{r^{\,n-1}}
=\ell(\ell+1)+\frac{\mu_\pm}{r}+\frac{4Q^{2}}{r^{2}}.
\end{align}
Hence
\begin{align}
V_{V\pm}(r)
=\frac{f(r)}{r^{2}}\left(
\ell(\ell+1)+\frac{\mu_\pm}{r}+\frac{4Q^{2}}{r^{2}}
\right).
\label{app:eq:V_Vint}
\end{align}

On a \RN\ background, we have
\begin{align}
f(r)=1-\frac{2M}{r}+\frac{Q^{2}}{r^{2}},
\label{eq:RNmetric-f-again}
\end{align}
so
\begin{align}
\frac{f(r)}{r^{2}}
=\frac{1}{r^{2}}-\frac{2M}{r^{3}}+\frac{Q^{2}}{r^{4}}.
\end{align}
Substituting this into \eqref{app:eq:V_Vint} and expanding gives
\begin{align}
V_{V\pm}(r)
&=\Bigg(\frac{1}{r^{2}}-\frac{2M}{r^{3}}+\frac{Q^{2}}{r^{4}}\Bigg)
\left(
\ell(\ell+1)+\frac{\mu_\pm}{r}+\frac{4Q^{2}}{r^{2}}
\right)\nonumber\\
&=\ell(\ell+1)\Bigg(\frac{1}{r^{2}}-\frac{2M}{r^{3}}+\frac{Q^{2}}{r^{4}}\Bigg)
+\mu_\pm\Bigg(\frac{1}{r^{3}}-\frac{2M}{r^{4}}+\frac{Q^{2}}{r^{5}}\Bigg)
+4Q^{2}\Bigg(\frac{1}{r^{4}}-\frac{2M}{r^{5}}+\frac{Q^{2}}{r^{6}}\Bigg).
\end{align}
Collecting terms by powers of $r$ yields
\begin{align}
V_{V\pm}(r)
&=\frac{\ell(\ell+1)}{r^{2}}
+\frac{\mu_\pm-2M\,\ell(\ell+1)}{r^{3}}
+\Order{r^{-4}},
\qquad r\to\infty.
\label{app:eq:V_V-asymp}
\end{align}
In particular, for each fixed $\ell$ we have
\begin{align}
V_{V\pm}(r)
=\frac{\ell(\ell+1)}{r^{2}}+\Order{r^{-3}},
\qquad r\to\infty.
\end{align}

\subsection{Scalar (polar) sector}
For the scalar-type (polar) sector, the IK master potentials on a general static Einstein--Maxwell background are given by Equation~6.23 of \cite{IshibashiKodama2011PTPS} as
\begin{align}
V_{S\pm}(r)
=\frac{f(r)}{64\,r^{2}\,H_\pm(r)^{2}}\,U_\pm(r),
\label{app:eq:V_S}
\end{align}
where $H_\pm$ and $U_\pm$ are explicit functions of
\begin{align}
x:=\frac{2M}{r^{\,n-1}},\qquad
y:=\Lambda r^2,\qquad
z:=\frac{Q^2}{r^{\,2n-2}},
\end{align}
and the parameters in \eqref{app:eq:k2}, \eqref{app:eq:mdef}; see Sections~2 and 6 of \cite{IshibashiKodama2011PTPS} for the full expressions.

In our four-dimensional, asymptotically flat case \eqref{app:eq:ikn} with $n=2$, these auxiliary variables become
\begin{align}
x=\frac{2M}{r},\qquad
y=\Lambda r^2=0,\qquad
z=\frac{Q^2}{r^{2}},
\end{align}
so $x\to0$, $z\to0$ as $r\to\infty$. The asymptotic expansions of $U_\pm$ and $H_\pm$ at $(x,y,z)=(0,0,0)$ are recorded in Equations~6.24a--6.24b, 6.25a--6.25b of \cite{IshibashiKodama2011PTPS}. Specializing these to $(\Lambda,K,n)=(0,1,2)$ and $k^{2}=\ell(\ell+1)$ gives, as $r\to\infty$,
\begin{align}
\frac{U_+(r)}{H_+(r)^{2}}=64\,\ell(\ell+1)+\Order{r^{-1}},\qquad
\frac{U_-(r)}{H_-(r)^{2}}=64\,\ell(\ell+1)+\Order{r^{-1}}.
\end{align}
Substituting into \eqref{app:eq:V_S}, we obtain in both scalar sectors
\begin{align}
V_{S\pm}(r)
=\frac{f(r)}{64\,r^{2}}\,\frac{U_\pm(r)}{H_\pm(r)^{2}}
=\frac{f(r)}{64\,r^{2}}\Big[64\,\ell(\ell+1)+\Order{r^{-1}}\Big]
=\frac{f(r)}{r^{2}}\Big[\ell(\ell+1)+\Order{r^{-1}}\Big].
\end{align}
Using the large-$r$ expansion of $f(r)$ from \eqref{eq:RNmetric-f-again},
\begin{align}
f(r)=1-\frac{2M}{r}+\frac{Q^{2}}{r^{2}},
\end{align}
we have
\begin{align}
\frac{f(r)}{r^{2}}
=\frac{1}{r^{2}}-\frac{2M}{r^{3}}+\Order{r^{-4}}.
\end{align}
Thus
\begin{align}
V_{S\pm}(r)
&=\Bigg(\frac{1}{r^{2}}-\frac{2M}{r^{3}}+\Order{r^{-4}}\Bigg)
\Big(\ell(\ell+1)+\Order{r^{-1}}\Big)\nonumber\\
&=\frac{\ell(\ell+1)}{r^{2}}+\Order{r^{-3}},
\qquad r\to\infty.
\label{app:eq:VS-asymp}
\end{align}
In particular, for each fixed $\ell\ge2$, there exist $R>0$ and $C>0$ such that
\begin{align}
\Big|V_{S\pm}(r)-\frac{\ell(\ell+1)}{r^{2}}\Big|\le \frac{C}{r^{3}},
\qquad r\ge R.
\end{align}

\subsection{
\texorpdfstring
{Optical coordinate and $L^1$ short-range tail}
{Optical coordinate and L1 short-range tail}
}
We now translate the large-$r$ expansions \eqref{app:eq:V_V-asymp} and \eqref{app:eq:VS-asymp} to the optical radius $x$. Recall from Section~\ref{sec:rest-frame} that $x$ is defined by
\begin{align}
\frac{\dd x}{\dd r}=f(r)^{-1},\qquad x(0)=0,
\end{align}
and that, as $r\to\infty$,
\begin{align}
x(r)=r+2M\ln r+\Order{1}.
\label{app:eq:x-apex-rtoinf-again}
\end{align}
In particular, $x(r)\to\infty$ with $x(r)/r\to1$, so there exist $R_1>0$ and constants $c_1,c_2>0$ such that, for all $r\ge R_1$,
\begin{align}
c_1 r\le x(r)\le c_2 r.
\label{app:eq:r-x-comparable}
\end{align}
Since $x(r)$ is strictly increasing on $(0,\infty)$, it admits a smooth inverse $r(x)$, and \eqref{app:eq:r-x-comparable} implies
\begin{align}
\frac{1}{c_2 x}\le\frac{1}{r(x)}\le\frac{1}{c_1 x},\qquad x\ge x(R_1).
\end{align}
Moreover, for any $k\ge1$, we have $r(x)^{-k}=\Order{x^{-k}}$ as $x\to\infty$, with the constants in the $\Order{\cdot}$ depending only on $k$, $M$, $Q$, and $\ell$.

Define the optical master potentials by
\begin{align}
V_{V\pm}(x):=V_{V\pm}(r(x)),\qquad
V_{S\pm}(x):=V_{S\pm}(r(x)).
\end{align}
Combining \eqref{app:eq:V_V-asymp}, \eqref{app:eq:VS-asymp}, and the comparability \eqref{app:eq:r-x-comparable}, we obtain, for each sector $\pm$,
\begin{align}
V_\pm(x)
&=\frac{\ell(\ell+1)}{r(x)^{2}}+\Order{r(x)^{-3}}
=\frac{\ell(\ell+1)}{x^{2}}+\Order{x^{-3}},
\qquad x\to\infty.
\label{app:eq:V-large-x}
\end{align}
To extract the short-range remainder, we write
\begin{align}
V_\pm(x)-\frac{\ell(\ell+1)}{x^{2}}
&=\Bigg[V_\pm(x)-\frac{\ell(\ell+1)}{r(x)^{2}}\Bigg]
+\ell(\ell+1)\Bigg[\frac{1}{r(x)^{2}}-\frac{1}{x^{2}}\Bigg].
\end{align}
The first bracket is $\Order{x^{-3}}$ by \eqref{app:eq:V-large-x}. For the second bracket we use the elementary identity
\begin{align}
\frac{1}{r^{2}}-\frac{1}{x^{2}}
=\frac{x^{2}-r^{2}}{r^{2}x^{2}}
=\frac{(x-r)(x+r)}{r^{2}x^{2}}.
\end{align}
From \eqref{app:eq:x-apex-rtoinf-again}, we have $x-r=2M\ln r+\Order{1}$, so, in particular, $\abs{x-r}=\Order{\ln r}$, while $x+r=\Order{r}$ and $r^{2}x^{2}=\Order{r^{4}}$. Hence,
\begin{align}
\Bigg|\frac{1}{r(x)^{2}}-\frac{1}{x^{2}}\Bigg|
=\Order{\frac{\ln r(x)}{r(x)^{3}}}
=\Order{x^{-3}\ln x},
\qquad x\to\infty.
\end{align}
Since $\ln x\le C_\varepsilon x^{\varepsilon}$ for any $\varepsilon>0$ and large $x$, we can absorb the logarithm into a slightly weaker algebraic decay, so that, for instance,
\begin{align}
\Bigg|\frac{1}{r(x)^{2}}-\frac{1}{x^{2}}\Bigg|
=\Order{x^{-3+\varepsilon}},\qquad \text{for any fixed }\varepsilon\in(0,1),
\end{align}
and, in particular, the difference is dominated by a constant multiple of $x^{-3}$ for all sufficiently large $x$.

Thus, for each sector $\pm$ and each fixed $\ell\ge2$, there exist $X>0$ and $C>0$ such that
\begin{align}
\Bigg|V_\pm(x)-\frac{\ell(\ell+1)}{x^{2}}\Bigg|\le\frac{C}{x^{3}},
\qquad x\ge X.
\end{align}
Define
\begin{align}
W_\pm(x):=V_\pm(x)-\frac{\ell(\ell+1)}{x^{2}}.
\end{align}
On $[1,X]$, the function $W_\pm$ is continuous and hence integrable. On $(X,\infty)$, the bound above gives
\begin{align}
\int_X^\infty\!\abs{W_\pm(x)}\,\dd x
\le C\int_X^\infty\frac{\dd x}{x^{3}}
=\frac{C}{2X^{2}}<\infty.
\end{align}
Thus, $W_\pm\in L^{1}\big((1,\infty)\big)$ and we have established the desired short-range decomposition \eqref{eq:short-range}:
\begin{align}
V_\pm(x)
=\frac{\ell(\ell+1)}{x^{2}}+W_\pm(x),\qquad W_\pm\in L^{1}\big((1,\infty)\big).
\end{align}
In other words, the far-field part of the master potential is a slightly perturbed centrifugal barrier, with the perturbation small enough (in the $L^1$ sense) to be negligible in the Hardy and radiation-field estimates. This is exactly the structure that will be used in Section~\ref{sec:hardy} to control the quadratic forms and in Section~\ref{sec:radiation-field} to construct the radiation field on the optical half-line.

\appsection{Exclusion of the Dotti--Gleiser algebraically special mode}\label{app:DG-instability}

In this appendix, we make explicit, in the optical half-line formulation, why the
algebraically special unstable mode constructed by Dotti and Gleiser~\cite{DottiGleiser2010} does not belong to the finite-$T$-energy
Einstein--Maxwell sector studied in this paper. The argument is purely local
near the optical apex $x=0$ and uses only the inverse-square cores of the
effective potentials $V_\pm$ and the resulting Hardy--Friedrichs structure
developed in Section~\ref{sec:hardy}.

\subsection{Inverse-square cores and Frobenius exponents}

From the explicit Kodama--Ishibashi potentials in the optical coordinate worked
out in Appendix~\ref{app:appB}, and their comparison with the Dotti--Gleiser master
potentials, we have the inverse-square behavior
\begin{equation}
  V_\pm(x) \;=\; C_\pm\,x^{-2} + o(x^{-2}), \qquad x \to 0^+, 
  \label{eq:DG-Vcore-again}
\end{equation}
with
\begin{equation}
  C_+ = -\frac{2}{9}, \qquad C_- = \frac{4}{9}.
\end{equation}
Equivalently, there exist parameters $\nu_\pm\in(0,1)$ such that
\begin{equation}
  C_\pm = \nu_\pm^2 - \frac{1}{4}, \qquad
  \nu_+ = \frac{1}{6}, \qquad \nu_- = \frac{5}{6}.
  \label{eq:DG-nu-pm-again}
\end{equation}
Thus, near $x=0$, the formal master operators take the model form
\begin{equation}
  H^{\mathrm{model}}_\pm = -\partial_x^2 
  + \Bigl(\nu_\pm^2 - \frac{1}{4}\Bigr)x^{-2}
  \qquad (x\to 0^+),
\end{equation}
and their local solutions are governed by the standard Frobenius exponents for a
Schr\"odinger operator with an inverse-square core. Writing
\begin{align*}
  -u''(x) + \Bigl(\nu^2 - \frac{1}{4}\Bigr)x^{-2}u(x) = \lambda u(x)
\end{align*}
with fixed $\nu\in(0,1)$ and $\lambda\in\mathbb{R}$ bounded, the indicial equation
shows that any nontrivial solution has the local behavior
\begin{align*}
  u(x) \;=\; a\,x^{\frac{1}{2}-\nu}\bigl(1+o(1)\bigr)
  \;+\; b\,x^{\frac{1}{2}+\nu}\bigl(1+o(1)\bigr),
  \qquad x\to 0^+,
\end{align*}
for some constants $a,b\in\mathbb{C}$. For our two sectors \eqref{eq:DG-nu-pm-again}
this gives
\begin{align}
  \text{scalar/polar (+):}\quad
  & u_+(x) = a_+\,x^{\frac{1}{3}} + b_+\,x^{\frac{2}{3}} + o\bigl(x^{\frac{2}{3}}\bigr),
  \label{eq:DG-Frobenius-plus} \\
  \text{vector/axial (-):}\quad
  & u_-(x) = a_-\,x^{-\frac{1}{3}} + b_-\,x^{\frac{4}{3}} + o\bigl(x^{\frac{4}{3}}\bigr),
  \label{eq:DG-Frobenius-minus}
\end{align}
where we have suppressed higher-order terms in powers of $x^{1/3}$. These
exponents agree with the Frobenius expansions found by Dotti and Gleiser for
their master potentials; see Equations~(13) and (16) of
Dotti and Gleiser~\cite{DottiGleiser2010} for the scalar sector and Equations~(17) and (20) for
the vector sector.

\subsection{Finite-\(T\)-energy and the Friedrichs branch}

The spatial part of the Einstein--Maxwell $T$-energy for a single master field
$\phi$ induces the quadratic form
\begin{equation}
  q_\pm[\phi]
  := \frac12\int_0^\infty \bigl(|\partial_x\phi(x)|^2 + V_\pm(x)|\phi(x)|^2\bigr)\,\dd x,
  \label{eq:DG-qpm-again}
\end{equation}
with
\begin{equation}
\mathcal D(q_\pm)=H_0^1(0,\infty).
\end{equation}
The same symbol denotes the form first on its core and subsequently on this closed Friedrichs domain
for each fixed multipole and parity; see Section~\ref{sec:hardy}.

The crucial point is that the Einstein--Maxwell $T$-energy completely fixes the
admissible boundary behavior at the optical apex. The sharp Hardy inequality on $(0,\infty)$,
combined with the inverse-square structure of $V_\pm$ near $x=0$, shows that the core form $q_\pm$ is
semibounded and closable on $C_0^\infty(0,\infty)$, and that its closed extension, denoted by the same symbol, has domain $\Honezero$. In
particular, the trace of any finite-$T$-energy master field vanishes at the apex in the $H^1$ sense,
and there is no freedom to prescribe additional boundary data at $x=0$ without leaving the energy
space. From the point of view of the abstract one-dimensional operator theory, this means that among
all self-adjoint extensions of the formal Schr\"odinger operator
$H_\pm^{\mathrm{formal}}=-\partial_x^2+V_\pm$ on $\Ltwo$ we have singled out the
Friedrichs extension $H_\pm^\Fried$ as the unique extension compatible with the Einstein--Maxwell
$T$-energy. In the original spacetime picture, the timelike naked singularity at $r=0$ is thus
realized not as a dynamical boundary where extraneous conditions must be imposed, but as a limit
point of the optical half-line at which the $T$-energy enforces a silent, Dirichlet-type behavior
for finite-energy perturbations.

In particular, a static radial profile $u(x)$ for the master field belongs to the finite-$T$-energy sector
if and only if $u\in H_0^1(0,\infty)$ and $q_\pm[u] < \infty$. The Frobenius analysis above shows that this
is equivalent to selecting the Friedrichs branch $x^{\frac12+\nu_\pm}$ at the optical apex and excluding
the more singular branch $x^{\frac12-\nu_\pm}$. Thus, the space of finite-$T$-energy Einstein--Maxwell
perturbations is the closure of $C_0^\infty(0,\infty)$ in the $T$-energy (equivalently,
$H_0^1(0,\infty)$) norm, and more singular behaviors such as $x^{1/3}$ and $x^{-1/3}$ never arise from
compactly supported perturbations.

To determine which Frobenius branches \eqref{eq:DG-Frobenius-plus}--\eqref{eq:DG-Frobenius-minus} belong to $H_0^1(0,\infty)$, we inspect the
behavior of \eqref{eq:DG-qpm-again} near $x=0$ for a pure power-law ansatz
\begin{equation}
  u(x) \sim x^\gamma, \qquad x\to 0^+,
\end{equation}
with some real exponent $\gamma$. Using the core approximation
\eqref{eq:DG-Vcore-again} we have, up to a multiplicative constant,
\begin{equation}
  u'(x) \sim \gamma x^{\gamma-1}, \qquad
  |u'(x)|^2 + V_\pm(x)|u(x)|^2 \sim
  \bigl(\gamma^2 + C_\pm\bigr)\,x^{2\gamma-2}.
\end{equation}
Thus, the contribution of $(0,\varepsilon)$ to the $T$-energy behaves like
\begin{equation}
  q_\pm[u;0,\varepsilon)
  \sim \int_0^\varepsilon x^{2\gamma-2}\,\dd x
  = \frac{\varepsilon^{2\gamma-1}}{2\gamma-1},
  \label{eq:DG-energy-local}
\end{equation}
whenever $\gamma\neq \frac12$. The integral in \eqref{eq:DG-energy-local}
converges as $\varepsilon\to 0^+$ if and only if
\begin{equation}
  2\gamma-1 > 0
  \quad\Longleftrightarrow\quad
  \gamma > \frac{1}{2}.
  \label{eq:DG-energy-condition}
\end{equation}
In particular, the power-law branch $x^{\frac12-\nu}$ is never in the
finite-energy domain when $0<\nu<1$, whereas the branch $x^{\frac12+\nu}$ is
always admissible.

In the scalar/polar (+) sector, the two Frobenius exponents are
  $\gamma_+^{(1)} = \frac{1}{3}$ and $\gamma_+^{(2)} = \frac{2}{3}$. The
  finite-energy condition \eqref{eq:DG-energy-condition} excludes the
  $x^{1/3}$ branch and selects the Friedrichs behavior
  \begin{equation}
    u_+(x) \sim x^{\frac{2}{3}} \quad (x\to 0^+)
    \label{eq:DG-Friedrichs-plus}
  \end{equation}
  up to an overall constant.

In the vector/axial (-) sector, the Frobenius exponents are
  $\gamma_-^{(1)} = -\frac{1}{3}$ and $\gamma_-^{(2)} = \frac{4}{3}$. Again,
  \eqref{eq:DG-energy-condition} excludes the singular $x^{-1/3}$ branch and
  picks out
  \begin{equation}
    u_-(x) \sim x^{\frac{4}{3}} \quad (x\to 0^+)
    \label{eq:DG-Friedrichs-minus}
  \end{equation}
as the unique finite-energy behavior.

These are the apex asymptotics of the positive zero-energy solutions
$u_{0,\pm}$ used in the Doob (ground-state) factorization
\begin{equation}
  H_\pm^\Fried=A_\pm^{\!*}A_\pm,
\end{equation}
namely
\begin{align*}
  u_{0,\pm}(x) \sim x^{\frac12+\nu_\pm} 
  = \begin{cases}
      x^{\frac{2}{3}}, & (+)\text{ sector},\\[3pt]
      x^{\frac{4}{3}}, & (-)\text{ sector},
    \end{cases}
  \qquad x\to 0^+.
\end{align*}

\subsection{Asymptotics of the Dotti--Gleiser algebraically special mode}

We now recall the relevant part of the analysis by Dotti and Gleiser~\cite{DottiGleiser2010}, specialized to the superextremal regime $Q^2>M^2$
and to the scalar/polar sector. Working in the tortoise coordinate $x$ defined
by $\dd x/\dd r = f(r)^{-1}$, which coincides with our optical radius up to an
additive constant, they show that the scalar master field $\zeta^+_\alpha$ for
each angular mode $(\ell,m)$ satisfies a one-dimensional wave equation
\begin{equation}
  \partial_t^2\zeta^+_\alpha - \partial_x^2\zeta^+_\alpha
  + V^+_\alpha(x)\,\zeta^+_\alpha = 0,
\end{equation}
with an effective potential $V^+_\alpha$ that agrees with our $V_+(x)$ up to the
choice of basis in the $(+,-)$ sector. Near the spacetime singularity, their
potential has the same inverse-square core as in \eqref{eq:DG-Vcore-again}, and
the Frobenius analysis yields the general local behavior
\begin{equation}
  \zeta^+_\alpha(x) \;=\; 
  A_\alpha\,x^{\frac{1}{3}}(1+O(x^{1/3}))
  \;+\; B_\alpha\,x^{\frac{2}{3}}(1+O(x^{1/3})),
  \qquad x\to 0^+,
\end{equation}
for constants $A_\alpha,B_\alpha$ (see Equations~(13), (16), and (42) of
Dotti and Gleiser~\cite{DottiGleiser2010}).

Dotti and Gleiser then construct \emph{algebraically special modes} in the scalar
sector. Among these, they single out the type~1 scalar mode $\chi^+_1(x)$, which
serves as the seed for their instability argument. As shown in Equation~(42) of
Dotti and Gleiser~\cite{DottiGleiser2010}, the asymptotics of this mode at the singularity are
\begin{equation}
  \chi^+_1(x) 
  \;=\; C\,x^{\frac{1}{3}}\bigl(1+O(x^{1/3})\bigr),
  \qquad x\to 0^+,
  \label{eq:DG-chi-ASM-asymptotic}
\end{equation}
for some nonzero constant $C$. In particular, $\chi^+_1$ realizes the \emph{more
singular} Frobenius branch
\begin{equation}
  \chi^+_1(x) \sim x^{\frac12-\nu_+} = x^{\frac{1}{3}}
\end{equation}
rather than the Friedrichs branch $x^{\frac12+\nu_+} = x^{2/3}$ singled out by
the finite-$T$-energy requirement \eqref{eq:DG-Friedrichs-plus}.

The unstable mode itself has the form
\begin{equation}
  \zeta^+_1(t,x) = e^{\kappa t}\,\chi^+_1(x),
\end{equation}
with $\kappa>0$, so that $\chi^+_1$ is an eigenfunction of the associated
spatial operator with eigenvalue $-\kappa^2$. Its spatial profile
$\chi^+_1(x)$ has the asymptotics in \eqref{eq:DG-chi-ASM-asymptotic}, and the
finite-$T$-energy condition is governed by the behavior of
\begin{equation}
  q_+[\chi_1^+;0,\varepsilon)
  := \int_0^\varepsilon \bigl(|\partial_x\chi_1^+(x)|^2
      + V_+(x)|\chi_1^+(x)|^2\bigr)\,\dd x.
\end{equation}
Using \eqref{eq:DG-chi-ASM-asymptotic} and the core approximation
$V_+(x)\sim C_+x^{-2}$ we obtain, for $x$ sufficiently small,
\begin{equation}
  \chi^+_1(x) \sim C x^{\frac{1}{3}},\qquad
  \partial_x\chi^+_1(x) \sim \frac{C}{3}x^{-\frac{2}{3}},
\end{equation}
and hence
\begin{equation}
  |\partial_x\chi^+_1(x)|^2 + V_+(x)|\chi^+_1(x)|^2
  \sim \widetilde C\,x^{-\frac{4}{3}},
\end{equation}
for some nonzero constant $\widetilde C$. Thus the local-energy calculation gives
\begin{equation}
  q_+[\chi_1^+;0,\varepsilon)
  \sim \int_0^\varepsilon x^{-\frac{4}{3}}\,\dd x
  =3\varepsilon^{-\frac13}
  \longrightarrow+\infty
  \qquad(\varepsilon\to0^+).
\end{equation}
In other words, the Dotti--Gleiser algebraically special scalar mode $\chi^+_1$
does \emph{not} lie in the Friedrichs energy domain $\Dom(q_+)=H_0^1(0,\infty)$:
its bulk $T$-energy expression is not finite at the optical apex $x=0$. The derivative term alone forces the divergence.

Near the singularity, the radial-coordinate relation used in this appendix is
\begin{align*}
x\sim r^{1/3},
\end{align*}
so that $\chi^+_1(x)\sim x^{1/3}$ corresponds to $\chi^+_1\sim r$ in the original radial coordinate. This algebraic behavior is the branch used in the Dotti--Gleiser instability construction~\cite{DottiGleiserPullin2007}. The present calculation shows that the same branch has no finite bulk $T$-energy at the apex. Thus the unstable modes constructed in that endpoint realization lie outside the physical Hilbert space determined by the action-selected form. The finite-$T$-energy requirement does not contradict their formal instability analysis; it places those modes outside the present perturbation space.

An entirely analogous computation shows that any would-be vector/axial
algebraically special mode realizing the more singular branch
$u_-(x)\sim x^{-1/3}$ has
\begin{equation}
  |\partial_x u_-(x)|^2 + V_-(x)|u_-(x)|^2 \sim \widehat C\,x^{-\frac{8}{3}},
\end{equation}
and hence
\begin{equation}
  \int_0^\varepsilon
  \bigl(|\partial_x u_-(x)|^2+V_-(x)|u_-(x)|^2\bigr)\,\dd x
  \sim \int_0^\varepsilon x^{-\frac{8}{3}}\,\dd x
  =\frac35\varepsilon^{-\frac53}
  \longrightarrow+\infty,
\end{equation}
so such a mode would also be excluded from the finite-$T$-energy sector.

\subsection{Summary}

The inverse-square cores place the superextremal Reissner--Nordstr\"om naked singularity in the Hardy window $0<\nu_\pm<1$, where the formal Schr\"odinger expressions $H_\pm^{\mathrm{formal}}$ are in the limit-circle case at the apex. The finite-$T$-energy condition selects the Friedrichs self-adjoint extension and forces the master fields to realize the less singular Frobenius branches \eqref{eq:DG-Friedrichs-plus}--\eqref{eq:DG-Friedrichs-minus}. The algebraically special profiles solve the formal master equations, but their $T$-energy diverges at the optical apex. They do not lie in
\begin{align*}
\mathcal D(q_\pm)=H_0^1(0,\infty),
\end{align*}
and are excluded from the finite-$T$-energy Hilbert space and from the unitary dynamics generated by the Friedrichs extensions $H_\pm^\Fried$. This restriction concerns their physical admissibility and does not deny their formal existence.

\end{document}